\pdfoutput=1
\RequirePackage{rotating}

\newif\ifextendedversion
\extendedversiontrue
\ifextendedversion
\documentclass[acmsmall,screen,nonacm,usenames,dvipsnames]{acmart}
\else
\documentclass[acmsmall,screen,usenames,dvipsnames]{acmart}
\fi

\AtBeginDocument{%
  }

\setcopyright{acmlicensed}
\copyrightyear{2025}
\acmYear{2025}
\acmDOI{XXXXXXX.XXXXXXX}

\acmConference[PLDI '25]{Make sure to enter the correct
  conference title from your rights confirmation emai}{June 03--05,
  2025}{Woodstock, NY}
\acmISBN{978-1-4503-XXXX-X/18/06}

\usepackage{tikz}
\usetikzlibrary{tikzmark}
\usetikzlibrary{calc}
\usetikzlibrary{arrows.meta}
\usetikzlibrary{positioning}
\usetikzlibrary{patterns}
\usetikzlibrary{shapes}
\usetikzlibrary{backgrounds}
\pgfdeclarepatternformonly{my grid}{\pgfqpoint{-1pt}{-1pt}}{\pgfqpoint{2pt}{2pt}}{\pgfqpoint{1pt}{1pt}}%
{
  \pgfsetlinewidth{0.1pt}
  \pgfpathmoveto{\pgfqpoint{0pt}{0pt}}
  \pgfpathlineto{\pgfqpoint{0pt}{1.1pt}}
  \pgfpathmoveto{\pgfqpoint{0pt}{0pt}}
  \pgfusepath{stroke}
}

\usepackage{placeins}
\usepackage{siunitx}
\usepackage{listings}
\usepackage{calc}
\usepackage{cancel}
\usepackage{subcaption}
\usepackage{pgfplots}
\usepackage{pgfplotstable}
\usepgfplotslibrary{fillbetween}
\usepackage{todonotes}
\usepackage{multirow}
\usepackage{enumitem}
\usepackage{fontawesome}
\usepackage{cleveref}
\usepackage{tabularx}
\usepackage{booktabs}
\usepackage{rotating}
\usepackage{wrapfig}
\usepackage{tablefootnote}

\newsavebox\warrowsymbol
\savebox\warrowsymbol{\begin{tikzpicture} \draw (0,0) rectangle (.2cm,.2cm); \draw (0,0) -- (.2cm,.2cm); \end{tikzpicture}}
\newcommand\widen{\mathbin{\nabla}}
\newcommand\narrow{\mathbin{\Delta}}

\newcommand\join{\sqcup}

\newcommand\sem[1]{[\![ #1 ]\!]}
\newcommand\ignore[1]{}

\newcommand\Var{\ensuremath{\mathcal{X}}}
\newcommand\Dom{\ensuremath{\mathbb{D}}}
\newcommand\Context{\ensuremath{\mathbb{C}}}
\newcommand\System{\ensuremath{\mathcal{C}}}
\newcommand\Globals{\ensuremath{\mathcal{G}}}
\newcommand\Locals{\ensuremath{\mathcal{L}}}
\newcommand\Solver{\ensuremath{\mathcal{S}}}
\newcommand\Procedures{\ensuremath{\text{Proc}}}

\newcommand\Goblint{%
{\textsc{Goblint}}%
}
\newcommand\dividedsides{\ensuremath{\tau}}

\newcommand\tdside{TD\textsubscript{side}}

\newcommand\slrp{SLR\textsuperscript{+}}
\newcommand{\combineLG}{\cup}

\tikzset{cfgnode/.style={circle,thick,draw,minimum size=.7cm}}

\lstdefinestyle{pseudocaml}{
    language=caml,
    basicstyle=\small\ttfamily,
    breaklines=true,
    columns=fixed,
    escapechar=@,
    frame=single,
    numbers=left,
    numberstyle=\tiny,
    stepnumber=1,
    mathescape=true,
}

\lstnewenvironment{pseudocaml}
  {
    \lstset{
        style=pseudocaml
    }
  }
  {}

  \lstdefinestyle{ccode}{
    language=c,
    basicstyle=\footnotesize\ttfamily,
    breaklines=true,
    frame=single,
    numbers=left,
    numberstyle=\tiny,
    stepnumber=1,
    mathescape=true,
    columns=fixed,
    escapechar=@,
    commentstyle=\color{OliveGreen}\ttfamily,
    morekeywords={foreach}
}

\lstnewenvironment{ccode}
  {
    \lstset{
        style=ccode
    }
  }
  {}

\newtheorem{theorem}{Theorem}[section]
\theoremstyle{definition}
\newtheorem{definition}[theorem]{Definition}

\newtheorem{example}{Example}[section]

\newcommand{\highlight}[2]{%
  \draw[green,fill=green,opacity=0.15]%
    ([yshift=8pt]pic cs:#1) rectangle ([yshift=-4pt]pic cs:#2);%
}

\pgfplotsset{compat=1.15}

\makeatletter %new code
\pgfdeclarepatternformonly[\LineSpace,\tikz@pattern@color]{north east lines hf}{\pgfqpoint{-1pt}{-1pt}}{\pgfqpoint{\LineSpace}{\LineSpace}}{\pgfqpoint{\LineSpace}{\LineSpace}}%
{
    \pgfsetcolor{\tikz@pattern@color} %new code
    \pgfsetlinewidth{0.4pt}
    \pgfpathmoveto{\pgfqpoint{0pt}{0pt}}
    \pgfpathlineto{\pgfqpoint{\LineSpace + 0.1pt}{\LineSpace + 0.1pt}}
    \pgfusepath{stroke}
}
\makeatother %new code
\newdimen\LineSpace
\tikzset{
    line space/.code={\LineSpace=#1},
    line space=3pt
}

\definecolor{T-Q-B1}{RGB}{68,119,170} % T-Q-B1 from colorblind package
\definecolor{T-Q-B2}{RGB}{102,204,238}
\definecolor{T-Q-B3}{RGB}{34,136,51}
\definecolor{T-Q-B4}{RGB}{204,187,68}
\definecolor{T-Q-B5}{RGB}{238,102,119}
\definecolor{T-Q-B6}{RGB}{170,51,119}

\pgfplotscreateplotcyclelist{graydiscernable}{
  {T-Q-B3},
  {T-Q-B4},
  {T-Q-B5},
}

\pgfplotsset{
  cycle list name=graydiscernable
}

\let\origthelstnumber\thelstnumber
\makeatletter
\newcommand*\Suppressnumber{%
  \lst@AddToHook{OnNewLine}{%
    \let\thelstnumber\relax%
     \advance\c@lstnumber-\@ne\relax%
    }%
}

\newcommand*\Reactivatenumber[1]{%
  \setcounter{lstnumber}{\numexpr#1-1\relax}
  \lst@AddToHook{OnNewLine}{%
   \let\thelstnumber\origthelstnumber%
   \refstepcounter{lstnumber}
  }%
}

\makeatother

\colorlet{improved}{green!30}
\colorlet{worse}{red!85}
\colorlet{incomparable}{brown!90!white}

\pgfplotscreateplotcyclelist{stackedbar}{
  {fill=improved},
  {fill=worse},
  {fill=incomparable},
  {preaction={fill, improved}, pattern=north west lines},
  {preaction={fill, worse}, pattern=north west lines},
  {preaction={fill, incomparable}, pattern=north west lines},
  {preaction={fill, improved}, pattern=north east lines},
  {preaction={fill, worse}, pattern=north east lines},
  {preaction={fill, incomparable}, pattern=north east lines},
  {preaction={fill, improved}, pattern=north west lines},
  {preaction={fill, worse}, pattern=north west lines},
  {preaction={fill, incomparable}, pattern=north west lines},
}

\newlist{questions}{enumerate}{2}
\setlist[questions,1]{label=({\bfseries RQ\arabic*}), ref=({\bfseries RQ\arabic*})}
\setlist[questions,2]{label=(\alph*), ref=\thequestionsi(\alph*)}

\newcommand{\benchmarkprog}[1]{\texttt{#1}}

\newcommand{\refapp}[1]{\ifextendedversion \cref{#1}%
\else supplementary material \labelcref{#1}%
\fi}

\newcommand{\Refapp}[1]{\ifextendedversion \Cref{#1}%
\else Supplementary material \labelcref{#1}%
\fi}

\begin{document}


% \title{Combatting Precision Loss in Side-Effecting Constraint Systems}
% \title{Abstract Garbage Collection to Combat Precision Loss in Static Analysis}
% \title{Recovering Precision with Robust Warrowing and Abstract Garbage Collection}
% \title{Write Your Blackboard and Clean it, Too: Pruning Results During Static Analysis}
% \title{Don’t Worry About the Side-Effects! Pruning Accumulated Information During Static Analysis}
% \title{Housekeeping During Abstract Interpretation with Mixed Flow-Sensitivity}
%  … Configurable / Adaptive Flow-Sensitivity
% Rendering Adaptive Flow-Sensitivity Precise
% \title{Taking out the Garbage On-the-Fly: Precision Recovery for Mixed Flow-Sensitivity  in Static Analysis}
% \title{Taking out the Toxic Trash: Abstract Garbage Collection for Mixed Flow-Sensitive Static Analyses}
\title{Taking out the Toxic Trash: Recovering Precision in Mixed Flow-Sensitive Static Analyses}
% \title{Trim \& Trash: Precision Recovery for Mixed Flow-Sensitivity in Static Analysis}
% \title{Trim \& Trash: Precision Recovery for Mixed-Flow-Sensitive Static Analyses}
% \title{Trimming \& Trashing: Recovering Precision in Mixed-Flow-Sensitive Static Analyses}
% \title{Maintain \& Trim \& Trash: Retaining Precision in Mixed-Flow-Sensitive Static Analyses}
% Preserve, conserve, clean, tidy/clear
% \title{Flow-Insensitivity Tamed: Purging Outdated Writes}


%% TODO: add other authors here

\author{Fabian Stemmler}
\affiliation{%
  \institution{Technical University of Munich}
  % \city{Garching}
  \country{Germany}}
\email{fabian.stemmler@tum.de}

\author{Michael Schwarz}
\affiliation{%
  \institution{Technical University of Munich}
  % \city{Garching}
  \country{Germany}}
\email{m.schwarz@tum.de}

\author{Julian Erhard}
\affiliation{%
  \institution{Technical University of Munich}
  % \city{Garching}
  \country{Germany}}
  \affiliation{%
  \institution{LMU Munich}
  % \city{Garching}
  \country{Germany}}
\email{julian.erhard@tum.de}

\author{Sarah Tilscher}
\affiliation{%
  \institution{Technical University of Munich}
  % \city{Garching}
  \country{Germany}}
\affiliation{%
  \institution{LMU Munich}
  % \city{Garching}
  \country{Germany}}
\email{sarah.tilscher@tum.de}

\author{Helmut Seidl}
\affiliation{%
  \institution{Technical University of Munich}
  % \city{Garching}
  \country{Germany}}
\email{helmut.seidl@tum.de}


\begin{abstract}
  Static analysis of real-world programs combines flow- and context-sensitive analyses of local program states with
computation of flow- and context-insensitive invariants at \emph{globals}, that, e.g., abstract data shared by multiple threads.
The values of locals and globals may mutually depend on each other, with the analysis of local program states both making contributions
to globals and querying their values.
Usually, all contributions to globals are accumulated during fixpoint iteration, with \emph{widening} applied to enforce termination.
Such flow-insensitive information often becomes unnecessarily imprecise and can include superfluous contributions --- \emph{trash} --- which, in turn,
may be \emph{toxic} to the precision of the overall analysis.
To recover precision of globals, we propose techniques complementing each other:
\emph{Narrowing on globals} differentiates contributions by origin;
\emph{reluctant widening} limits the amount of widening applied at globals;
and finally,
\emph{abstract garbage collection} undoes contributions to globals and
propagates their withdrawal.
The experimental evaluation shows that these techniques increase the precision of mixed flow-sensitive analyses at a reasonable cost.

\end{abstract}

\begin{CCSXML}
  <ccs2012>
     <concept>
         <concept_id>10003752.10010124.10010138.10010143</concept_id>
         <concept_desc>Theory of computation~Program analysis</concept_desc>
         <concept_significance>500</concept_significance>
         </concept>
     <concept>
         <concept_id>10003752.10010124.10010138.10010142</concept_id>
         <concept_desc>Theory of computation~Program verification</concept_desc>
         <concept_significance>500</concept_significance>
         </concept>
     <concept>
         <concept_id>10003752.10010124.10010138.10011119</concept_id>
         <concept_desc>Theory of computation~Abstraction</concept_desc>
         <concept_significance>300</concept_significance>
         </concept>
   </ccs2012>
\end{CCSXML}

\ccsdesc[500]{Theory of computation~Program analysis}
\ccsdesc[500]{Theory of computation~Program verification}
\ccsdesc[500]{Theory of computation~Abstraction}

\keywords{static analysis, abstract interpretation, flow-sensitivity, flow-insensitivity}
%%  \includegraphics[width=\textwidth]{sampleteaser}
%%  \caption{Seattle Mariners at Spring Training, 2010.}
%%  \Description{Enjoying the baseball game from the third-base
%%  seats. Ichiro Suzuki preparing to bat.}

% \received{20 February 2175}
% \received[revised]{12 March 2175}
% \received[accepted]{5 June 2175}

\maketitle

\section{Introduction}
Abstract interpretation, pioneered by \citet{cousot1977}, underpins many expressive and highly scalable static analyses.
%
% Describe both extremes first to introduce terms, then go to the unified setting
Often, abstract values representing invariants are computed for each program point, yielding a \emph{flow}-sensitive analysis.
For a more precise interprocedural analysis,
abstract states at program points may further be differentiated by calling context, making the analysis additionally \emph{context}-sensitive.
Other analyses completely give up on differentiating information by program point and collect only \emph{flow-insensitive} information.

The notion of \emph{partial} flow-sensitivity has been considered as a middle ground,
where some program points are treated flow-sensitively, and others are not \cite{RoyS07,RinetzkyRSY08,ParkLR22}.
However, this term still does not capture approaches that are flow-sensitive w.r.t.\
some aspect of the state but are flow-insensitive for others.
This is, e.g., the case for thread-modular analyses of shared data.
We propose the term \emph{mixed flow-sensitivity} to encompass all of these phenomena.
Mixed flow-sensitivity is employed when some information should be aggregated regardless of origin.
This aggregation of flow-insensitive information happens at so-called \emph{globals}.
Conversely, flow-sensitive information is said to be associated with \emph{locals}, e.g., program points.

Mixed flow-sensitivity enables elegant formulations of analyses and can be a key ingredient for making them tractable.
It arises, e.g., when \emph{global store widening} \cite{shivers1991} --- sometimes considered essential for scalable analyses \cite{DaraisLNH17,JohnsonLMH13} ---
is applied not to the entire state but only to selected parts.
Similarly, an efficient yet precise pointer analysis may track only some points-to sets flow-sensitively~\cite{LhotakC11}.
The choice to switch off flow-sensitivity may be done online, e.g., because memory resources are about to be exhausted \cite{HeoOY19}.
% Their widening is not (necessarily) a widening in the Cousot&Cosout sense!
Mixed flow-sensitivity also is the method of choice for analyzing non-local control-flow such as
\texttt{setjmp/longjmp}~\cite{schwarz2023}.
% and possibly also co-routines and effects. %\cite{xie2021,alvarez2024}.
Context-sensitive analysis capturing call strings, value-based contexts, and further abstractions
can also conveniently be expressed~\cite{erhard2024}.
Still, perhaps the most widespread application is overcoming the state explosion encountered in the analysis of multi-threaded programs by analyzing each thread
flow-sensitively, while using globals to flow-insensitively abstract the communication between threads~\cite{vene2003,mine2012,mine2014,Mine2017,Mine2018,Gotsman07,StievenartNMR19,schwarz2023b,vojdani2016,schwarz2021}.
Mixed-flow sensitivity also enables thread-modular analyses of memory safety~\cite{SaanESBHTVS24a}.
% \todo[inline,color=yellow]{Shouldn't the citations be before the period? Especially as this is not the end of the paragraph?}
%             The style guides disagree. This is usually what we go with:)
% This way, the transition to the example is more natural
Abstract domains for expressing flow-insensitive properties may
come with infinite ascending chains, implying that \emph{widening}~\cite{cousot1977} at globals may be required to ensure termination.
As a result, the abstract values computed at globals may be rather coarse overapproximations.
\begin{example}\label{ex:incdec}
    Consider the analysis of multi-threaded code where globals are used to abstract shared data.
    Listings \ref{lst:incthread} and \ref{lst:decthread} show two threads concurrently incrementing, respectively decrementing, a shared variable \lstinline[style=ccode]|a|.
    Here and in following code examples of concurrent code, we assume that accesses to shared variables are atomic.
\begin{figure}
\begin{minipage}{0.465\linewidth}
\begin{lstlisting}[style=ccode,label=lst:incthread,caption=Incrementing Thread]
int a = 0;
void thread1() {
    while(1) {
        int i = a;
        if(i < 10)
            a = i + 1;
    }
}
\end{lstlisting}
\end{minipage}
\hfill
\begin{minipage}{0.465\linewidth}
\begin{lstlisting}[style=ccode,label=lst:decthread,caption=Decrementing Thread]
void thread2() {
    while(1) {
        int i = a;
        if(i > -10)
            a = i - 1;
    }
}
\end{lstlisting}
\end{minipage}
\end{figure}
%    	In all examples throughout this work, we assume the solving order to be that of the recursive top-down solver
%	\cite{seidl2021,erhard2022}.
    %
    Assume that the program points of the two threads
    are analyzed flow-sensitively, while the value of the variable \lstinline{a} shared between the two threads
    is analyzed flow-insensitively.
    Due to the guards preceding the assignments in both threads, $-10 \leq a \leq 10$ holds throughout the execution
    of both threads.
    Now assume that an interval analysis is performed where initially \lstinline{a} has the abstract value $[0,0]$,
    and termination is enforced by widening.

    For the sequence of writes to \lstinline{a} by the \emph{incrementing} thread in \cref{lst:incthread},
    already the first write causes the abstract value for
    \lstinline{a} to increase to $[0,\infty]$.
    For the sequence of writes to \lstinline{a} by the \emph{decrementing} thread
    from \cref{lst:decthread},
    widening will also give up on the lower bound, i.e.,
    % to the values of \lstinline{a}, i.e.,
    result in $[-\infty,\infty]$.
\end{example}

Even when no widening is applied to globals,
their precision may suffer from contributions later found
to be too large or withdrawn entirely.
This effect is common when the analysis has \emph{non-monotonicities},
e.g., because it uses not only widening for locals, but also \emph{narrowing} to recover lost precision, or simply, because
the analysis is \emph{context-sensitive} where different iterations over the code may refer to different abstract
calling contexts \cite{fecht1999}.
Such outdated contributions may adversely affect the precision of invariants elsewhere, meaning that on
top of being \emph{trash} (i.e., unneeded) they may be \emph{toxic} in worsening precision
elsewhere.

% In this work,
Here, we study general mechanisms to regain lost precision for flow-insensi\-tive properties
in mixed flow-sensi\-tive analyses.
% To remedy the deficiencies illustrated by \cref{ex:incdec},
For that, we encapsulate the mechanisms by which mixed flow-sensi\-tive analysis frameworks update the values of globals
by \emph{update rules}.
Update rules are generic and can, with reasonable effort, be integrated into any such framework.
Update rules may have inter\-nal state, e.g., to
track the contributions from locals individually and perform widening and narrow\-ing on these.
In \cref{ex:incdec}, this is key to establishing the invariant $-10 \leq a \leq 10$.
% for the shared variable \lstinline|a|.
Thus, update rules can make an analyzer more precise --- without necessitating changes to other components.

Our update rules are presented and evaluated in the context of
\emph{side-effecting constraint systems}~\cite{apinis2012} (recalled in \cref{sec:background}).
This formalism allows describing
% combinations of flow-sensitive with flow-insensitive
mixed flow-sensitive analyses and has been used for many of the analyses outlined above.
We introduce the general form of update rules in \cref{sec:update_rules}.
These rules can also be plugged into the frameworks for
mixed flow-sensitivity outlined in the related work (\cref{sec:relatedwork}).
As a starting point, we extract such a rule from a fixpoint solver proposed by \citet{apinis2013} whose effect has never been systematically evaluated (\cref{sec:divideandnarrow}).
The idea there is to track from where a particular global has received contributions,
and then perform widening as well as narrowing on the \emph{join} of the respective values.
That approach turns out to suffer from precision loss due to unnecessary widening already when
different locals each contribute different constant values.
% identical values are contributed to the globals -- if only the contributions from different locals differ.
% \todo[inline,color=teal]{I don't think I understand this. \emph{identical} values?
% Or is this the case where we get $1$ from one place and $2$ from the other?
% And does narrowing not save them?
% }
% \todo[inline]{Yes - identical values, but different ones from different places. And no: narrowing may sometimes,
% but not always come as a rescue.}
%
We then propose new generic update rules with increasing levels of sophistication.
Our core contributions are as follows:
\begin{itemize}
    \item We propose to
    % \emph{localize}\todo{This term is barely used later.}
    apply widening and narrowing \emph{per origin}, i.e.,
    to the contributions of locals individually (subsection~\ref{ss:localized});
    \item We propose to \emph{reluctantly} apply widening to prevent
          precision loss when the new contribution of a local is already
	  subsumed by the current value of the global;
    \item To ensure that reluctant widening still enforces a finite number of updates, we introduce
    	  the notion of \emph{strong widening}, a sufficient condition met by common widening
	  strategies such as \emph{threshold widening} (subsection \ref{ss:reluctant});
          % requirement is fulfilled by many common widening operators, though there are exceptions.
    \item We furthermore introduce \emph{abstract garbage collection}, a technique to remove \emph{withdrawn} contributions to globals to eliminate irrelevant unknowns and recursively prune further spurious contributions (subsection \ref{ss:garbage}).
\end{itemize}
We have implemented our update rules within the static analysis framework \Goblint{}~\cite{vojdani2016,SaanSAESVV21}
for multi-threaded {C} programs and evaluate the impact on precision and performance (\cref{sec:evaluation}).

\section{Preliminaries}
\label{sec:background}
% Instead of concrete program states,
Abstract interpretation in general considers abstract states from an
abstract domain $\Dom$, representing \emph{properties} of concrete program states.
% or invariants.
The set $\Dom$ comes with a partial order $\sqsubseteq$ (corresponding to implication between properties).
Here, we additionally assume $\Dom$ to be a \emph{bounded lattice}.
This means that $\Dom$ comes with a
least element $\bot$ (corresponding to the property \emph{false}),
a greatest element $\top$ (corresponding to the property \emph{true}), and binary operators
\emph{join} for computing the least upper bound (denoted by $\sqcup$) and
\emph{meet} for computing the greatest lower bound (denoted by $\sqcap$).
% The lattice may have large or even infinite ascending chains, necessitating the use of widening.

Assume that the program consists of a finite set $\Procedures$ of procedures where each procedure $p\in\Procedures$ is given by
a control flow graph $\mathcal{G}_p = (N_p, E_p)$.
The set $N_p$ collects the program points of $p$ including $\textsf{st}_p$ and $\textsf{ret}_p$, the (unique) start and return points
of procedure $p$. Each edge $e=(u,a,v)$ in $E_p \subseteq N_p \times L \times N_p$ consists of
a start point $u$ and an end point $v$, together with a label $a\in L$.
Here, we also assume that for each procedure $p$, $\textsf{ret}_p$ is reachable in $\mathcal{G}_p$ from $\textsf{st}_p$.

The set of labels of edges consists
of all actions possibly executed by the program, i.e.,
basic statements of the source language such as assignments, as well as guards. % corresponding to conditional jumps.
For interprocedural analysis, also \emph{call edges} with labels $f(e_1,\ldots,e_k)$
where $e_1,\ldots,e_k$ are side-effect free expressions not containing global variables.
% Why are side-effects or globals in these expressions a problem?
Program execution starts with a call to the procedure $\textsf{main}\in\Procedures$.

For conveniently representing the abstract semantics of mixed flow-sensitive analyses,
we chose \emph{side-effecting} constraint systems \cite{vene2003,apinis2012}.
We refer to the variables of the constraint system for which abstract values are to be computed as \emph{unknowns}.
This set of unknowns is given as the union of two disjoint sets $\Globals$ and $\Locals$.
$\Globals$ encompasses all unknowns for which abstract states are collected \emph{flow-insensitively} --- referred to as \emph{globals}, whereas $\Locals$ contains unknowns which are to be analyzed \emph{flow-sensitively} --- referred to as \emph{locals}.
This framework may serve as a foundation for all flavors of mixed flow-sensitive analyses as outlined in the introduction.
% can be instantiated in a multitude of ways as discussed in the introduction.
We briefly sketch two instantiations which we refer back to throughout this work.
In \cref{ex:incdec} globals were introduced for abstracting data shared between multiple threads.
In the literature, globals have been used to abstract individual pieces of data \cite{vene2003,StievenartNMR19,mine2012,vojdani2016} as well as clusters yielding relational information \cite{schwarz2023b}.
On the other hand, globals may also be introduced for context-sensitive analysis of procedures, where they correspond to
pairs  $(\textsf{st}_p, c)$ of \emph{start points} $\textsf{st}_p$ of procedures $p$
and contexts $c$ from some set $\Context$ of \emph{abstract calling contexts} the analysis discriminates.
The unknown $(\textsf{st}_p,c)$ then collects the abstract values passed to $p$ from all reached call sites of $p$
for the calling context $c$ \cite{apinis2012}.
% which are tokens which the analysis discriminates.
%
% allow to determine \emph{Flow-} and \emph{context-sensitive} invariants.
In context-sensitive analyses,
the set $\Locals$ of \emph{local} unknowns are pairs $(v, c)$ of program points $v$ (different from
the start point of the respective procedure) and contexts $c$.
For each local unknown, the constraint system provides constraints.
% which must all hold for a solution.
% \todo{Definition of solution follows later, so this could be skipped here}
% all to be met by the abstract values for the unknowns.
%
For interprocedural analysis,
each edge $e = (u, a, v)$ provides for every context $c$ an abstract value which must be subsumed by the control-flow successor
$v$ in context $c$, i.e., by the unknown $(v,c)$.
Depending on the label $a$, additional contributions to globals may be triggered:
\begin{itemize}
\item	Assume $a$ is an assignment to the global variable $g$ which is analyzed flow-insensitively.
	Then, a contribution to the unknown corresponding to $g$ is caused. The contribution is the abstract value of the right-hand side of the assignment;
\item	Assume $a$ is a call of some procedure $p$ with abstract calling context $c'$.
    Then, the abstract state computed for the entry of $p$ is contributed to the unknown $(\textsf{st}_p,c')$.
    % , i.e., the start point of $p$ in calling context $c'$.
% \item \todo[inline]{Do we need thread creation here? H: keep it simple - I would say: no!}
\end{itemize}
Thus, the abstract value for a local as well as the generated side-effects depend on the abstract value
at its control-flow predecessor.
In case of a call to procedure $p$ with context $c'$, also the unknown $(\textsf{ret}_p,c')$ for the
return point $\textsf{ret}_p$ of $p$ is accessed.
Furthermore, values of global unknowns corresponding to shared global variables may be queried.
Altogether, the constraint for the control-flow edge $e$ leading to some node $v$ in context $c$ thus takes the form:
\begin{equation}
(\sigma\,(v,c), \rho)  \sqsupseteq \sem{e,c}^\sharp\,(\sigma\combineLG\rho)
\label{def:edge_constraint}
\end{equation}
where $\sigma:\Locals\to\Dom$, $\rho:\Globals\to\Dom$ map locals and globals to abstract values, respectively,
and $\combineLG$ denotes the combination of the two mappings into one with domain $\Locals\cup\Globals$.
Accordingly, the right-hand side
% $\sem{e,c}^\sharp$
is a function of type
% \begin{equation}
$\sem{e,c}^\sharp:((\Locals\cup\Globals)\to\Dom)\to(\Dom\times(\Globals\to\Dom)).$
% \label{def:rhs_type}
% \end{equation}
There, the first component of the result is the contribution to the local control-flow
successor and the second describes the encountered contributions to globals.
We assume that all abstract functions $\sem{e,c}^\sharp$ are \emph{strict} in the control-flow
predecessor $(u,c)$, i.e., $\sem{e,c}^\sharp(\sigma\combineLG\rho) = (\bot,\emptyset)$ whenever
$(\sigma\combineLG\rho)\,(u,c) = \bot$ and triggers non-$\bot$ contributions to a finite set of globals only.
Thus, the contributions to globals can be represented by a finite set of tuples
$\{(g_1, d_1),\ldots,(g_r,d_r)\}$ with the understanding that globals not mentioned in the enumeration do not receive a contribution.

In addition to the constraints for edges, the analysis requires an initial value
$d_0\in\Dom$ at the start point of
\text{main} for some initial calling context $c_0$.
Therefore,
% To deal with that,
we introduce one further constraint
\begin{equation}
(\sigma\,\_\text{main},\rho)\sqsupseteq (\sigma\,(\textsf{ret}_{\text{main}},c_0),\rho_0)
\label{def:init_constraint}
\end{equation}
for a dedicated local $\_\text{main}$.
The right-hand side returns the value of the end point $\textsf{ret}_{\text{main}}$ of \text{main} in the context $c_0$
and contributes some initial abstract values to globals.
In particular, $((\textsf{st}_{\text{main}},c_0), d_0)\in\rho_0$.
The constraint system for a small example program is provided in \refapp{app:fac}.

A pair $(\sigma,\rho)$ of assignments of abstract values from $\Dom$ to locals and globals is called
\emph{solution} if it satisfies all given constraints.
The solution is \emph{finite} if $\sigma\,x \neq\bot$ and $\rho\,y\neq\bot$ only for finitely many unknowns $x\in\Locals$ and
$y\in\Globals$.
Technically, we equivalently collect all constraints for a given local $x$ into a single right-hand side
\begin{equation}
(\sigma\,x,\rho)\;\sqsupseteq f_x\,(\sigma\combineLG\rho)
\label{def:single_constraint}
\end{equation}
In case that $x \equiv \_\text{main}$,
$f_x\,(\sigma\combineLG\rho) = (\sigma\,(\textsf{ret}_{\text{main}},c_0),\rho_0)$, and otherwise
for $x\equiv (v,c)$,
$f_x$ is the least upper bound of all functions $\sem{e,c}^\sharp$, where $e$ is a control-flow edge with end point $v$.
We refer to the collection of constraints \eqref{def:single_constraint}
% together with the constraint \eqref{def:init_constraint}
as the \emph{side-effecting} constraint system $\System$ for the program.

Depending on the set $\Context$ of calling contexts distinguished by interprocedural analysis,
the constraint system $\System$ may be very large if not infinite.
Still, finite solutions of $\System$ often do exist and 
can be found by
\emph{local} solvers, such as \slrp{} \cite{apinis2013} or \tdside{} \cite{seidl2021}.
Instead of solving the entire constraint system, such solvers attempt to
provide non-$\bot$ values for just enough unknowns to determine the value of one unknown of \emph{interest}.
In the setting of program analysis, this unknown of interest is the dedicated local $\_\text{main}$.
Querying this unknown will cause the local solver to determine abstract values for all unknowns
corresponding to the program points in appropriate abstract calling contexts so that all concrete
program executions are covered~\cite{seidl2021}.
Even for finite sets of encountered unknowns, though,
the analysis may not terminate due to
infinite ascending
% (and/or descending)
chains in the domain $\Dom$.
To enforce termination, we therefore rely on \emph{widening} and \emph{narrowing} \cite{cousot1977,cousot1992}.
\begin{definition}\label{def:widen}
    An operator $\widen$ is a \emph{widening operator}, if it subsumes the join operator, i.e., $a \join b \sqsubseteq a \widen b$
    for all $a,b\in\Dom$; and, for every sequence $b_i\in\Dom$, $i\geq 1$ and $a_0\in\Dom$,
    the sequence $a_i\in\Dom, i\geq 1,$ defined by
    $a_i = a_{i-1} \widen b_i$, is ultimately stable.
    Generally, we assume that $\bot\widen b = b$.
    \footnote{We build on the \emph{original} definition of widening by \citet{cousot1977}.
    \citet{CousotC92Comparing} later reused the term for a weaker notion.
    \citet{CortesiZ11} refer to the original definition as \emph{strong} widening.
    Our strong widening proposed in \cref{def:strongwiden} is a stronger version of the original definition.
    }
\end{definition}
Analyses with aggressive widening often terminate quickly, but lose significant amounts of precision.
In some cases, the precision can be recovered by subsequent narrowing.
\begin{definition}
    An operator $\narrow$ is a \emph{narrowing operator}, if $a\sqcap b \sqsubseteq a\narrow b\sqsubseteq a$ for all
    $a, b \in \Dom$,
    and, for every sequence $b_i\in\Dom$, $i\geq 1$ and $a_0\in\Dom$,
    the sequence $a_i\in\Dom, i\geq 1,$ defined by
    $a_i = a_{i-1} \narrow b_i$, is ultimately stable.
    Moreover, we generally assume that $a\narrow\bot = \bot$.
\end{definition}

\begin{example}\label{ex:interval}
    Consider the interval domain extended with $\bot$ where $\top=\relax[-\infty,\infty]$ with the usual order and the classic widening and narrowing.
    For non-$\bot$ cases, it is given by
    \[
\begin{array}{lll}
    \relax[ l_a,u_a] \widen [l_b,u_b] &=& [l_a \leq l_b\;?\;l_a : -\infty,\quad u_a \geq u_b\;?\;u_a : \infty ]\\
    \relax[ l_a,u_a] \narrow [l_b,u_b] &=& [l_a \neq -\infty\;?\;l_a : l_b,\quad u_a \neq \infty\;?\;u_a : u_b ]
\end{array}
    \]
    This definition guarantees that chains become ultimately stable as changed bounds are forgotten for widening, and narrowing only improves infinite bounds.
\end{example}

\section{Update Rules}
\label{sec:update_rules}
Recovering precision is a common challenge across mixed-flow sensitive analysis frameworks.
We thus aim for a \emph{generic} solution, that is not bogged down by incidental details of frameworks and their solver mechanisms ---
applying equally to solvers for side-effecting constraint systems ~\cite{apinis2013,seidl2021}, solvers % amato2016
formulated as nested fixpoints in the style of~\citet{mine2012} or \citet{StievenartNMR19}, and solvers for blackboard architectures~\cite{keidel2024}.
To this end, we distill the handling of globals into \emph{update rules} which are \emph{decoupled} from any given analysis framework and its solver:
We assume that this hosting solver maintains a mapping $\rho: \Globals \rightarrow \Dom$ to store the current values
of globals.
Contributions to globals are not directly collected in $\rho$. Instead, whenever contributions are to be made,
a function \lstinline{update_globals} is called which returns a set of globals together with their new values.
The solver then updates the values of globals in $\rho$ accordingly.
Such a function \lstinline{update_globals} is called an \emph{update rule} and receives as arguments
\begin{enumerate}
    \item the local \lstinline{orig} making the contributions;
    \item the set \lstinline{contribs} of pairs $(g,b)$ of new
    contributions $b$ to globals $g$ where each global appears at most once, and finally;
    \item the map $\rho$ from globals to their current values.
\end{enumerate}
Requiring \lstinline{contribs} to contain at most one contribution per global does, in fact, not limit expressivity:
If multiple contributions are to be made to the same global,
the hosting solver can simply combine them via \emph{join}.
In case of a solver for side-effecting constraint systems, the function \lstinline{update_globals}
is called after every evaluation of a right-hand side.

\begin{example}
	Consider as a first example the update rule given in \cref{lst:always_join}.
	For each contribution $b$ to a global $g$,
	\lstinline{update_globals} checks whether $b$ is currently
	already subsumed by $\rho\,g$. If it is, no update is necessary. Otherwise,
	the value $b$ is joined into $\rho\,g$ to obtain the new value for $g$.
	The resulting set of updates is returned.
\end{example}

To arrive at a modular statement of soundness for combinations of solvers and update rules, we define requirements for both:
At a high level, we require the solver to not go wrong as long as $\rho$ always accounts for all last contributions.
Observe that the current value of $\rho$ for some global $g$ is only required to subsume the \emph{last}
contribution of each origin to $g$ --- and not \emph{all}. This enables later contributions to $g$ to
shrink the value of $\rho$ for $g$.
As long as the update rule then is \emph{sound}, i.e., it ensures that each last contribution to a global $g$ is subsumed by the produced updates,
the solver with the update rule remains sound.
More formally:
\begin{definition}[Hosting Solver]
%
% At a high level, we require the solver to not go wrong as long as $\rho$ always accounts for all last contributions.
%
We require that a \emph{hosting solver} $\Solver$ guarantees:
\begin{enumerate}[label={\bfseries(S\arabic*)}, leftmargin=1cm]
\item \label{prop:solver-call-seq} Each run of $\Solver$ performs a sequence of calls
	\lstinline{update_globals}$(o_i,C_i,\rho_i),i\geq 0,$
	with origins $o_i$, finite sets $C_i\subseteq \Globals\times\Dom$ of contributions to globals,
	and mappings $\rho_i:\Globals\to\Dom$. Each call
	returns a set $U_i\subseteq\Globals\times\Dom$
	of updates to globals so that for each $i\geq 0$,
	$\rho_{i+1} = \rho_i\oplus\{g\mapsto d\mid (g,d)\in U_i\}$.
\item \label{prop:solver-sound} The result of a run performing the sequence \lstinline{update_globals}$(o_i,C_i,\rho_i),i= 0,\ldots n-1,$ of updates
	is \emph{sound} provided that for each $j=1,\ldots,n$, each global $g$ and each origin $o$
	contributing to $g$, $\rho_j$ subsumes the \emph{latest} contribution of $o$ to $g$, i.e., if there is some $i < j$
	such that
	(1)  $o_i = o$ with $(g,d)\in C_i$, and
	(2)   $o_{i'}\neq o$ for all $i<i'<j$,
	then $d\sqsubseteq\rho_j\,g$.
\end{enumerate}
\end{definition}
\begin{definition}[Sound Update Rule]
We call an update rule \lstinline{update_globals} \emph{sound} if
for each sequence \lstinline{update_globals}$(o_i,C_i,\rho_i), i\geq 0$, of calls returning sets $U_i$ of updates
with $\rho_{i+1} = \rho_i\oplus\{g\mapsto d\mid (g,d)\in U_i\}$
% constructed as above
and each $j$-th call \lstinline{update_globals}$(o_j,C_j,\rho_j)$ in this sequence, the following two properties are met:
\begin{enumerate}[label={\bfseries(R\arabic*)}, leftmargin=1cm]
\item \label{prop:update-pair-exists}
If $(g, b)\in C_j$ where $b \not\sqsubseteq \rho_j\,g$ holds,
then $U_j$ has a pair $(g,\_)$.
%
% \item[\textbf{(R1)}]
% For each pair $(g,d)\in U_j$ and each origin $o_i, i\leq j$ with
% $(g,d_i)\in C_i$ with $o_i \neq o_{i'}$ for all $i<i'\leq j$
% (i.e., the $i$th call has provided the last contribution $d_i$ of $o$ to $g$),
% it holds that $d_i\sqsubseteq d$.
\item \label{prop:update-last-accounted-for}
For each origin $o$, let $i$ be the index of the last \lstinline{update_globals} call for $o$ in the sequence ending with the $j$th call.
Then for all  $(g,d)\in U_j$, if $(g,d_i)\in C_i$ also $d_i\sqsubseteq d$.
\end{enumerate}
\end{definition}
% \medskip
% Accordingly, we have:

\begin{theorem}\label{t:soundness}
A \emph{hosting solver} % satisfying properties \labelcref{prop:solver-call-seq} and \labelcref{prop:solver-sound}
using a sound update rule returns sound results when it terminates.
\end{theorem}

\begin{proof}
We show that during each run of a hosting solver $\Solver$ with a suitable update rule plugged in,
the assumption of property \labelcref{prop:solver-sound} is met, i.e., for each global $g$, the value
$\rho\,g$ always subsumes the latest contribution from $o$
for every encountered origin $o$.
Assume that  \lstinline{update_globals}$(o_i,C_i,\rho_i), i = 0,\ldots,n-1$, is a sequence of calls
corresponding to a run of the solver $\Solver$.
The proof is by induction on the number $j$ of updates already performed during that run.
First assume that $j=0$, i.e., no call to \lstinline{update_globals}
has been executed yet.
Therefore, the assertion for $\rho_0$ is trivially met.
Now assume that assertion for $\rho$ holds for the first $j-1$ calls to \lstinline{update_globals}, and
let $U$ denote the set of updates returned by \lstinline{update_globals}$(o_j,C_j,\rho_j)$.
By property \labelcref{prop:solver-call-seq},
$\rho_{j+1}=\rho_j\oplus\{g'\mapsto d'\mid (g',d')\in U\}$.
Consider any pair $(g,d)\in C_j$.
If $d\sqsubseteq\rho_j\,g$ and there is no pair $(g,\_)\in U$, then $d\sqsubseteq\rho_j\,g = \rho_{j+1}\,g$ (by \labelcref{prop:solver-call-seq})
and by induction hypothesis the assertion holds.
If, on the other hand, $d\not\sqsubseteq\rho_j\,g$, there is a pair $(g,\_)\in U$ by \labelcref{prop:update-pair-exists}.
Thus, for all globals appearing in $C_j$, either the property holds by induction, or an update is produced.

Consider now the set $U$ of all updates, whether for a global in $C_j$ or not.
For a pair $(g,d)\in U$, by \labelcref{prop:update-last-accounted-for}, $d$ must subsume the latest contribution from any origin $o$, and as by \labelcref{prop:solver-call-seq},
$\rho_{j+1}\,g = d$, the assertion follows for $g$.
Now, consider a global $g$ which neither appears in $C_j$ nor in $U$. For these, the assertion follows from the induction hypothesis as no new contributions occur
and $\rho_{j+1}\,g=\rho_j\,g$ (by \labelcref{prop:solver-call-seq}) in this case.
From that the assertion follows.
\end{proof}
Consider again the update rule given in \cref{lst:always_join}.
This update rule certainly satisfies properties \labelcref{prop:update-pair-exists} and \labelcref{prop:update-last-accounted-for} and therefore is sound.
However, it has a severe drawback: accumulation of values by the hosting solver at globals will often
fail to terminate when the domain $\Dom$ has infinite ascending chains.
One natural-seeming remedy is replacing the operator $\sqcup$ in line   % Better adj.?
5 with a widening operator $\widen$.
This corresponds to the still sound update rule used in \cref{ex:incdec}.
Using widening to combine the old value of a global $g$ with new, non-subsumed, contributions
guarantees that the number of updates to $g$ is finite.

\begin{figure}
    \begin{lstlisting}[style=ccode,caption={\label{lst:always_join}Update rule for globals to handle contributions triggered by right-hand sides by always joining.}]
update_globals($\Locals$ orig, $2^{\Globals\times\Dom}$ contribs, $\Globals\to\Dom$ $\rho$) {
  updates = $\emptyset$; // Updates to be propagated to the solver
  foreach((g,b) $\in$ contribs) { // Iterate over contributions
    if(b $\sqsubseteq$ $\rho$ g) continue; // Nothing to do
    updates = updates $\cup$ {(g, $\rho$ g $\sqcup$ b)}; // Compute new value by join
  }
  return updates;
}
    \end{lstlisting}
\end{figure}

\section{Precision Recovery}
\label{sec:divideandnarrow}
When a fixpoint solver performs narrowing at locals, contributions to globals may \emph{shrink} or
even disappear during fixpoint iteration.
Even without narrowing, this can happen due to non-monotonic right-hand sides.
%
% In presence of narrowing at locals or due to non-monotonicity in constraints,
% globals may intermediately receive large values which later are replaced with smaller ones or even \emph{withdrawn}.
%
The update rule~\ref{lst:always_join}, with or without $\widen$, however,
cannot recover precision when contributions shrink or disappear.
Using $\narrow$ to combine the old value of a global $g$ with new contributions $b$ when $b\sqsubseteq\rho\,g$
% in the modified version of the update rule % in \cref{lst:always_join}
is not viable:
The updated value for $g$ need not only subsume the latest contribution provided by a \emph{single}
local --- but the latest contributions provided by \emph{all} locals.

% Alas, even the unmodified version of update rule~\ref{lst:always_join} using $\join$ cannot benefit form
% \emph{shrinking} or \emph{withdrawn} contributions!
%
To remedy this deficiency, we propose the update rule in \cref{lst:updaterule_imp}.
This update rule internally maintains a data structure \lstinline|cmap|
which for each global $g$ and local \lstinline{orig} separately records
the latest contribution to $g$ originating from \lstinline|orig|.
The data structure \lstinline|cmap| is \lstinline|static| in the {C} sense: its value survives across invocations of the update rule,
but as its scope is limited to inside the update rule, it cannot be accessed or modified from outside.
% made by evaluating the constraint for $x$ last.
We demand that this hashmap \lstinline|cmap| is initialized to $\bot$ everywhere when solving begins,
and that it only represents non-$\bot$ values explicitly.
%
%% NEW
By induction on the sequence of evaluations of right-hand sides,
we verify that indeed, only the \emph{latest} contribution from each local is recorded.
The contributions from all locals are combined via join to determine the new value $d$ for a global $g$.
By construction, this update rule is thus sound as well.
% where outdated contributions are replaced in
% \lstinline|cmap| whenever a new contribution arrives.
% which replaces old contributions
% with new contributions whenever a right-hand side is re-evaluated.

\begin{figure}
    \begin{lstlisting}[style=ccode,caption={\label{lst:updaterule_imp}Update rule for globals to handle contributions triggered by right-hand sides
      keeping contributions separate and thus allowing shrinking contributions.
      The semantics of \lstinline|static| is borrowed from {C}: The value of \lstinline|cmap| survives across invocations, but the scope of \lstinline|cmap| is limited to the function.}]
update_globals($\Locals$ orig, $2^{\Globals\times\Dom}$ contribs, $\Globals\to\Dom$ $\rho$) {
  static cmap; // Hashmap from $\Globals$ and $\Locals$ to $\Dom$ with default value $\bot$

  updates = $\emptyset$; // Changes to be propagated to the solver
  foreach((g,b) $\in$ contribs) { // Iterate over contributions
    cmap[g][orig] = b; // Record contribution
    d = $\bigsqcup_{\ell \in \Locals}$ cmap[g][l]; // Compute new value
    if(d != $\rho$ g) updates = updates $\cup$ {(g,d)};
  }
  return updates;
}
    \end{lstlisting}
\end{figure}

\ignore{
Since the individual contributions from locals onto each global are maintained separately,
the implementation of \lstinline{update_globals} updates the contribution from
each local $x$ separately and thus may replace its former (possibly too large) contribution to a global
when it is outdated.
}

Splitting contributions by their origin also is done in \citet{apinis2013} % amato2016
where it is deeply baked into the \slrp\ fixpoint engine.
After abstracting away implementation details % First occurrence, ok to cite both here
and their host solver and re-casting as an update rule, we obtain \cref{lst:apinis-simple_imp}.
To deal with potential non-termination due to infinite ascending chains in
$\Dom$, they apply widening to combine the newly computed update for a global $g$
with the previous value for $g$.
Moreover, since the new value $b_\sqcup$ accounts for the latest contributions to $g$ of
\emph{all} locals, \emph{narrowing} can be applied whenever $b_\sqcup$ is subsumed by
$\rho\,g$.
% Generally applying \emph{widening} at globals may unnecessarily give up precision.
%
% Therefore, \citet{apinis2013} suggest applying \emph{narrowing} at $g$ whenever all contributions to $g$  jointly
%from all locals provide a smaller value (\cref{lst:apinis-simple_imp}). Otherwise, widening is applied.

The resulting update rule is thus sound.
One major drawback of this approach, though, is that distinct constant contributions
from several locals may result in unnecessary
widening. % while our alternative approach still may return their join.

    \begin{figure}
        % lstlsiting is unhappy unless we put curlies around the cite which specifies the section
        \begin{lstlisting}[style=ccode,caption={\label{lst:apinis-simple_imp}Updating globals using widening and narrowing --- extracted from {\cite[Section 6]{apinis2013}}.
        Equality is defined in terms of the lattice order and thus \lstinline|a == b| is shorthand for $a \sqsubseteq b \wedge b \sqsubseteq a$.}]
update_globals($\Locals$ orig, $2^{\Globals\times\Dom}$ contribs, $\Globals\to\Dom$ $\rho$) {
  static cmap; // Hashmap from $\Globals$ and $\Locals$ to $\Dom$ with default value $\bot$

  updates = $\emptyset$; // Updates to be propagated to the solver
  foreach((g,b) $\in$ contribs) { // Iterate over contributions
    if(cmap[g][orig] == b) continue;
    cmap[g][orig] = b; // Record contribution
    a = $\rho$ g;
    b$_\sqcup$ = $\bigsqcup_{\ell \in \Locals}$ cmap[g][l];
@
\begin{tikzpicture}[remember picture,overlay]
    \highlight{a47}{b47};
\end{tikzpicture}
\tikzmark{a47}
@    d = if(b$_\sqcup$ $\sqsubseteq$ a) { a $\narrow$ b$_\sqcup$ } else { a $\widen$ b$_\sqcup$ };@\hspace*{\fill}\tikzmark{b47}@
    if(d != $\rho$ g) updates = updates $\cup$ {(g,d)};
  }
  return updates;
}
        \end{lstlisting}
    \end{figure}

\begin{wrapfloat}{lstlisting}{r}{0.35\textwidth}
    \vspace{-0.8\baselineskip}
    \setlength{\abovecaptionskip}{0em}
    \setlength{\belowcaptionskip}{0.2\baselineskip}
    \begin{center}
    \begin{lstlisting}[style=ccode,xleftmargin=5.0ex,caption={\label{lst:example-global-update}Program modifying a shared variable.}]
int a = 0;
int main() {
  a = 1;
  assert(a < 2);
}
    \end{lstlisting}
    \end{center}
\end{wrapfloat}

\begin{example}\label{ex:apinis}
    Consider an interval analysis of the program in \cref{lst:example-global-update}.
    Assume that the value of variable \lstinline{a} is analyzed flow-insensitively.
    This global then receives contributions $[0,0]$ and $[1,1]$.
    When widening is performed directly for globals, we obtain $[0,\infty]$ for
    % or alternatively $[-\infty,1]$ for
    \lstinline{a} -- given that line 1 is processed by the solver first.
    Then, the assertion in line 4 cannot be shown.
    We remark that while narrowing could, in principle, help recover precision here,
    the update rule extracted from \cite{apinis2013} (\cref{lst:apinis-simple_imp}),
    does not apply narrowing in this instance.
    Since the individual contributions remain constant, the check in line 5 of the update rule holds true and leads to the constant contributions not being considered for narrowing.
    % The assertion thus remains out of reach, even with narrowing.

    % \begin{figure}
    % \begin{lstlisting}[style=ccode,caption={Distinct constant writes.},label={apinis-imprecision}]
    % int a = 0;
    % int main() {
    %    a = 1;
    %    assert(a < 2);
    % }
    % \end{lstlisting}
    % \end{figure}
    \end{example}
%       This issue exists not only for constant side-effects.

\noindent
Distinct contributions from several locals to the same global $g$ are common, e.g.,
for the analysis of multi-threaded programs, where each (potential) write to a global
variable triggers a contribution to at least one global.
Therefore, we propose performing widening and narrowing \emph{per origin}, i.e.,
for the contributions of each local unknown to $g$ separately.
\subsection{Localized Widening and Narrowing}\label{ss:local}\label{ss:localized}
%
% As the contributions to a global $g$ are collected separately for each local,

Instead of over the joined values, as done in \cref{lst:apinis-simple_imp},
we propose to perform both widening and narrowing
to the contributions from each local \lstinline|orig| to $g$ separately
(\cref{lst:update_globals_warrow_imp}).
For the interval analysis of
% or \ref{lst:update_globals_robust_warrowing_imp} is used,
the program from \cref{ex:apinis}, we find that
all contributions to \lstinline|a| originate from distinct locals.
Applying the update rule from \cref{lst:update_globals_warrow_imp}
therefore neither applies widening nor narrowing.
The abstract value found for \lstinline{a}, thus, is $[0,1]$, and the assertion is proven.
The impact of narrowing is illustrated in the subsequent example.

\begin{wrapfloat}{lstlisting}{Tr!}{0.35\textwidth}
\vspace{-0.8\baselineskip}
\setlength{\abovecaptionskip}{0em}
\setlength{\belowcaptionskip}{0.2\baselineskip}
\begin{center}
\begin{lstlisting}[style=ccode, xleftmargin=5.0ex,caption={\label{lst:fact}Factorial program.}]
int t = 1;
void fac(int i) {
  if (i > 0) {
    fac(i-1);
    t = i * t;
  }
  else assert(i == 0);
}
void main() {
  int i = 17;
  fac(i);
}
\end{lstlisting}
\end{center}
\end{wrapfloat}

\begin{figure}[b]
    \begin{lstlisting}[style=ccode,caption={\label{lst:update_globals_warrow_imp}Updating globals using localized widening and narrowing.}]
update_globals($\Locals$ orig, $2^{\Globals\times\Dom}$ contribs, $\Globals\to\Dom$ $\rho$) {
  static cmap; // Hashmap from $\Globals$ and $\Locals$ to $\Dom$ with default value $\bot$

  updates = $\emptyset$; // Updates to be propagated to the solver
  foreach((g,b) $\in$ contribs) { // Iterate over contributions
    a = cmap[g][orig];
    if(a == b) continue; // Value unmodified
@
\begin{tikzpicture}[remember picture,overlay]
    \highlight{a18}{b18};
\end{tikzpicture}
\tikzmark{a18}
@    a = if(b $\sqsubseteq$ a) { a $\narrow$ b } else { a $\widen$ b };
    cmap[g][orig] = a;@\hspace*{\fill}\tikzmark{b18}@
    d = $\bigsqcup_{\ell \in \Locals}$ cmap[g][l]; // Compute new value
    if(d != $\rho$ g) updates = updates $\cup$ {(g,d)};
  }
  return updates;
}
    \end{lstlisting}
\end{figure}

\begin{example} \label{ex:nocontext}
    % \todo[inline]{people may not be familiar with this approach, we would need to elaborate on this}
Consider an interprocedural interval analysis of the recursive factorial shown in \cref{lst:fact}.
We use \emph{global-store widening} for the program variable $t$, and
widening and narrowing operators as recapped in \cref{ex:interval}.
Since the only local variable in the program is \lstinline{i}, we also use plain intervals to represent local
states of the program.

% \begin{figure}
% \begin{minipage}[t]{0.3\textwidth}

% \end{minipage}
% \caption{\label{ex:fac}}
% \end{figure}

    Assume that the context to differentiate calls is given by the abstract interval value of the parameter
    at the call-site.
    Then for each context given by an interval $c$, \emph{every} contribution to the global
    $(\textsf{st}_{\text{fac}}, c)$ is equal to $c$.
    % $\{\mbox{\lstinline{i}}\mapsto c\}$, when ignoring the local variable $t$, which is always $\top$.
    Thus, neither widening nor narrowing is applied.
    Such \emph{full} context analysis may be expensive or even
    cause non-termination, in particular for recursive programs~\cite{erhard2024}.

    Consider instead an analysis of the same program without any context, i.e., with $\Context = \{ \bullet  \}$.
    The initial call \lstinline|fac(i)| in \text{main},
    and the subsequent recursive calls all produce contributions to the single unknown
    $(\textsf{st}_{\text{fac}}, \bullet)$.
    Specifically, successive iterations of the recursive call edge cause the contributions
    $(\textsf{st}_{\text{fac}}, \bullet) \mapsto[16,16]$ and
    $(\textsf{st}_{\text{fac}}, \bullet) \mapsto[15,16]$ %\todo{why not [15,15]?} -> because it is context-insensitive
    % $(\textsf{st}_{\text{fac}}, \bullet) \mapsto\{\mbox{\lstinline{i}}\mapsto[16,16]\}$ and
    % $(\textsf{st}_{\text{fac}}, \bullet) \mapsto\{\mbox{\lstinline{i}}\mapsto[15,16]\}$ %\todo{why not [15,15]?} -> because it is context-insensitive
    --- implying that widening is applied.
    At this point, the analysis obtains $[-\infty,16]$
    as the contribution from that call site to the global $(\textsf{st}_{\text{fac}}, \bullet)$, and altogether
    $\rho\,(\textsf{st}_{\text{fac}}, \bullet) = [-\infty,17]$.
    With this overapproximation, the assertion in line 6 cannot be proven.
    Narrowing will recover precision:
    As the guard \verb|i > 0| establishes a lower bound, the next iteration on the procedure body will trigger the contribution
    $\{(\textsf{st}_{\text{fac}}, \bullet) \mapsto[0,16]\}$ so that the contribution
    of line 3 is improved to
    $\{(\textsf{st}_{\text{fac}}, \bullet) \mapsto[0,16]\}$ ---
    and the assertion can be proven.
    % To achieve the same effect with widening delay, a counter value of at least $17$ would be required.
    A complete specification of the constraint systems for
    context-sensitive as well as context-insensitive analysis of this program is given in \refapp{app:fac}.
\end{example}
% \FloatBarrier

This update rule (\cref{lst:update_globals_warrow_imp}), as will all other enhancements of the base update rule of \cref{lst:updaterule_imp}
preserve soundness. However, there is another important issue:
The use of widening and narrowing alone does not guarantee the number of updates to a global to be finite
for update rules \ref{lst:apinis-simple_imp} and \ref{lst:update_globals_warrow_imp}  ---
even for constraint systems with monotonic right-hand sides.
% We provide an example of this issue for
% update rule \ref{lst:update_globals_warrow_imp}.
% and indicate why the problem also arises for update rule \ref{lst:apinis-simple_imp}.

\begin{example}\label{ex:nonterm}
Consider the following constraint system with $\Locals = \{x, y\}, \Globals = \{a, b\},
\Dom = \mathbb{N}_0 \cup \{\infty\}$ and $\sqsubseteq~=~\leq$ where $ a\widen b = \infty$
whenever $b\not\leq a$
% goes to $\infty$ if the second argument is not subsumed
and $a \narrow b = b$ whenever $a=\infty$ and $a$ otherwise:
\[
    \begin{array}{l}
        % (\sigma\,x,\rho) \sqsupseteq (\sigma\,y \join \sigma\,z, \emptyset) \quad\;\;
        (\sigma\,x,\rho) \sqsupseteq (\rho\,a, \{ a \mapsto (\rho\,b) + 1 \}) \\
        (\sigma\,y,\rho) \sqsupseteq (\rho\,b, \{ b \mapsto (\rho\,a) + 1 \})
    \end{array}
\]
All contributions to $a$ originate from $x$, whereas those to $b$ all originate from $y$.
% The unknowns $x$ and $y$ need to be re-evaluated whenever the value of $a$ or $b$ changes, and thus also another side-effect to $a$ or $b$ is triggered whenever their values has changed.
Consider a solver that solves this constraint system for both $x$ and $y$, and
always stabilizes the current unknown before iterating others affected by a change.
Such a solver observes an increased side-effect to $a$ when evaluating the
local $x$. From the second time on, widening is applied for this contribution,
but the lost precision is fully recovered when narrowing with the still unchanged
contribution during the subsequent re-evaluation of $x$.
%
% The subsequent iteration leads to preliminary stabilization of the value for $x$.
Because the value of $y$ is affected by the increased value for $a$,
a similar re-evaluation for $y$ is triggered once the value for $x$ has stabilized.
The new value for $b$ in turn necessitates a re-evaluation of $x$ and so on.
This results in the following infinite sequence:
% of global updates:

\vspace{0.2cm}
\noindent\begin{tabularx}{\textwidth}{l||*{5}{>{\centering\arraybackslash}X|}*{2}{c|}*{1}{>{\centering\arraybackslash}X|}*{2}{c|}*{1}{>{\centering\arraybackslash}X|}*{1}{>{\centering\arraybackslash}X}}
  %\hline
  % x & 0 & 0 & 1 & & & 1 & $\infty$ & 3 & & & & \dots \\
  a & 0 & 1 & 1 & & & $\infty$ (1 $\widen$ 3) & 3 ($\infty$ $\narrow$ 3) & 3 & & & & \dots \\
  \hline
  % y & 0 & & & 0 & 2 & & & & 2 & $\infty$ & 4 & \dots \\
  b & 0 & & & 2 & 2 & & & & $\infty$ (2 $\widen$ 4) & 4 ($\infty$ $\narrow$ 4) & 4 & \dots \\
  %\hline
\end{tabularx}
\vspace{0.1cm}

\noindent While update rule \ref{lst:apinis-simple_imp} does not cause an infinite number of widening/narrowing (W/N) switches in this example,
a slightly modified (still monotonic) constraint system triggers the problem there.
\Refapp{appendix:warrowing-apinis-nonterm} provides this example as well as a program that gives rise to constraints akin to the ones in the example above.

\ignore{
To compute a value for $x$, $y$ is solved first. In the first iteration, $\rho\,a$ is modified via the contribution
$a \mapsto 1$. In the second iteration, the same contribution is triggered, causing solving of $y$ to terminate.
Subsequent evaluation of $z$ behaves similarly, but now schedules $y$ for re-computation because it is affected by the contribution $b \mapsto 2$.
Recomputing a value for $y$, however, triggers the contribution $a \mapsto 3$.
After applying widening to the contributions from $y$ to $a$, it increases to $\infty$.
A subsequent re-computation of $y$ causes the same contribution $a\mapsto 3$ --
due to which the contribution of $y$ is narrowed to $3$.
The interplay between the locals $y$ and $z$ and the globals $a$ and $b$ thus
may result in an infinite sequence of updates to the values both of locals and globals.
}
% In this way, the values of $a$ and $b$ infinitely oscillate.
\end{example}
To cope with non-termination for our update rule \ref{lst:update_globals_warrow_imp},
we bound the number of phase switches in update rule~\ref{lst:update_globals_robust_warrowing_imp}.
For every pair $(g,x)\in\Globals\times\Locals$, we introduce a counter (called \emph{W/N gas})
which keeps track of how often such a switch from widening to narrowing has happened for
contributions of $x$ to $g$.
Then --- as soon as a particular threshold $N$ for the \emph{W/N gas} has been reached ---
$a\narrow b$ is replaced with $a$.
These counters together with the latest kind $\Box\in\{\widen,\narrow\}$ of applied
combine operator are maintained
in the data-structure \lstinline{cmap} as well.
By default, the initial value for a local $x$ in the map for global $g$ now is set to
$(\bot,\widen,0)$.
Here, the function \lstinline{first} extracts the first component from a triple of values.
A switch to narrowing is only performed finitely often for the contributions
of any local to any global.
Thus, when using the novel update rule
for the constraint system from \cref{ex:nonterm}, the number of encountered updates remains finite.

\begin{figure}[b]
    \begin{lstlisting}[style=ccode,caption={\label{lst:update_globals_robust_warrowing_imp}Updating globals with localized widening and bounded narrowing (when ignoring box in line 10). Replacing line 10 by the contents of the box yields a variant using recultant widening.}]
update_globals($\Locals$ orig, $2^{\Globals\times\Dom}$ contribs, $\Globals\to\Dom$ $\rho$) {
  static cmap; // Hashmap from $\Globals$ and $\Locals$ to $\Dom$, phase, and gas,
               // default value is $(\bot,\widen,0)$
  updates = $\emptyset$; // Updates to be propagated to the solver
  foreach((g,b) $\in$ contribs) { // Iterate over contributions
    (a,$\Box$,gas) = cmap[g][orig];
    if(a == b) continue; // Value unmodified
    if(b $\not\sqsubset$ a) {
      $\Box$ = $\widen$; // b not accounted for -> Widening
@
\begin{tikzpicture}[remember picture,overlay]
    \draw[red]([shift={(-3pt,2ex)}]pic cs:starta) rectangle([shift={(4pt,-0.65ex)}]pic cs:enda);
  \highlight{a46}{b46};
\end{tikzpicture}
\tikzmark{a46}
@      a = a $\widen$ b;     @\tikzmark{starta}@@\textbf{RELUCTANT: }@a = if(b $\sqsubseteq$ $\rho$ g) { a $\sqcup$ b } else { a $\widen$ b };@\tikzmark{enda}@  @\hspace*{\fill}\tikzmark{b46}@
    }
    else if($\Box$ == $\narrow$) a = a $\narrow$ b;
    else if($i$ >= $N$) continue; // Gas exhausted -> Do not narrow
    else (a,$\Box$,gas) = (a $\narrow$ b, $\narrow$, gas+1) // Change phase to narrow and consume gas
    cmap[g][orig] = (a,$\Box$,gas); // Update cmap
    d = $\bigsqcup_{\ell \in \Locals}$ first cmap[g][l]; // Compute new value
    if(d != $\rho$ g) updates = updates $\cup$ {(g,d)};
  }
  return updates;
}
    \end{lstlisting}
\end{figure}

\subsection{Reluctant Widening}\label{ss:reluctant}
%TODO: maybe have a control flow graph for the example in section 2, so we can refer to nodes
When re-analyzing \cref{ex:incdec} with the update rule in \cref{lst:update_globals_robust_warrowing_imp}, some but not all precision can be recovered.
% Since the other case is analogous,
Assume, e.g., that the incrementing thread is solved first.
The contribution to $a$ by the contribution from line 5 of \cref{lst:incthread}
is initially $[1,1]$.
Subsequently, it is widened with $[1,2]$, resulting in $[1,\infty]$. In the following solving iteration,
the guard \lstinline{i < 10} in line 4 allows
% lost precision to be recovered and
the contribution to be narrowed to $[1,10]$.
The join over all contributions to $a$ then amounts to $[0,10]$. Subsequent analysis of thread 2 similarly produces
contributions $[-1,9]$, $[-\infty,9]$ and finally $[-10, 9]$.
Now, the join over all contributions yields the desired
interval $[-10, 10]$ for $a$.
However, this update to $a$ forces the re-analysis of thread 1, whose analysis depended
on an outdated value of $a$. As the lower bound for $a$ has decreased, thread 1 now contributes $[-9,10]$, which widens
its contribution to $[-\infty,10]$. As thread 1 does not enforce a lower bound on the value of the local $i$, the
lost precision is not recovered, and the analysis will terminate with the imprecise invariant $a \mapsto [-\infty,10]$.

The imprecision is caused by widening contribution $[1,10]$ with $[-9, 10]$. In this case, performing
a join operation instead would have allowed the analysis to terminate immediately.
The value of $[-9,10]$
was already subsumed by the running solution $\rho~a$. This leads us to propose a \emph{reluctant} application of
widenings.  Our modified update rule as shown in \cref{lst:update_globals_robust_warrowing_imp} now replaces widening with a \emph{join}
if the new contribution is already subsumed by the current value of the global.

\begin{example}
    With this modification to \lstinline{update_globals}, the analysis of \cref{ex:incdec} indeed retains the precise
    invariant $\mbox{\lstinline{a}} \mapsto [-10,10]$,
    as for thread 1, the contributions $[1,10]$ and $[-9, 10]$ are joined, rather than widened.
\end{example}

Interestingly, for some widening operators, the update rule using reluctant widening fails to guarantee finiteness of updates to globals.
For an example where it falls short, see \refapp{appendix:reluctant-nonterm}.
However, for some widening operators, finiteness remains guaranteed.
This is, e.g., the case for a class we call \emph{strong} widening operators.

\begin{definition} \label{def:strongwiden}
    The operator $\widen: \Dom \times \Dom \rightarrow \Dom$ is a \emph{strong} widening operator,
    if for all $a,b\in\Dom$, $a\sqcup b\sqsubseteq a\widen b$ and $a\widen b =a$ whenever $b \sqsubseteq a$ and,
    for every increasing sequence $a_1\sqsubseteq a_2\sqsubseteq \ldots$ in $\Dom$
    and every sequence $b_i\in\Dom, i\geq 1$, where for all $i\geq 1$,
    $a_i\widen b_i\sqsubseteq a_{i+1}$ it holds that there is some $i_0\geq 1$
    such that $b_i\sqsubseteq a_i$ for all $i\geq i_0$.
\end{definition}

% By definition, each sequence of $a_i$ considered by definition \ref{def:strongwiden} is ascending.
Unlike in \cref{def:widen}, the left operand of $\widen$ in the ascending sequence of the $a_i$ need not be
the result of the previous widening operation, but must subsume it.
% The given definition should be viewed in the context of the conservative warrowing operator.
Such sequences, e.g., arise if
\emph{join} operations are interspersed between the widening operations.
We remark that any \emph{strong} widening operator also is a normal widening operator (\cref{def:widen}).
Demanding widenings to be strong is not unreasonable: many useful domains already provide strong widenings.

\begin{example}
    % This is easy to see,
    \emph{Threshold widening} \cite{blanchet2003} applied to an element $a_{i-1}$ with a non-subsumed element $b_i$
    always increases to the next threshold subsuming $a_{i-1}\sqcup b_i$.
    Assuming that the number of thresholds is finite, the sequence of the $a_i$ must necessarily be ultimately stable.
    Thus, threshold widenings applied to interval bounds result in strong widenings. % right reference?
    Therefore, the standard widening for intervals is strong.
    The same holds for
    the standard widening for the octagon domain~\cite{mine2006}.
    % when taking special care of when to apply closure.
\end{example}

\noindent
For the cartesian product of two domains $\mathbb{P} = \Dom_1 \times \Dom_2$
with the componentwise ordering,
the widening operator $\widen$ that applies strong widening operators $\widen_1,\widen_2$ componentwise,
is also strong.
% a widening operator can be constructed by $(a_1,a_2) \widen (b_1,b_2) = (a_1 \widen_1 b_1, a_2 \widen_2 b_2)$.
% If both $\widen_1$ and $\widen_2$ are strong, then so is $\widen$.

However, not all widening operators are strong.
Consider the cartesian product domain $\mathbb{P}$,
now ordered lexicographically.
Consider again a widening operator $\widen$ which is the componentwise combination of widening operators $\widen_1,\widen_2$ where
%
% Assume further that again,
% $(a_1,a_2) \widen (b_1,b_2) = (a_1 \widen_1 b_1, a_2 \widen_2 b_2)$ where we assume that
for both $i=1,2$,
% $\widen_i$ returns the left argument
$a_i\widen_i b_i = a_i$ whenever $b_i\sqsubseteq a_i$.
Let $d_1 \sqsubset_1 d_2 \sqsubset_1 \dots$
be an infinite strictly ascending chain
in $\Dom_1$.
For the sequence $\bar b_i = (d_i,\top)$ and the value $\bar a_0 = (d_1,\bot)$,
consider the sequence $\bar a_i, i\geq 1,$ defined by $\bar a_i  = (d_{i+1},\bot)$.
Then for all $i\geq 1$,
\[
\bar a_i = (d_{i+1},\bot)\sqsupseteq (d_i,\top) =  (d_i,\bot) \widen (d_i,\top) = \bar a_{i-1}\widen\bar b_i
\]
This sequence forms an infinite strictly ascending chain -- implying that $\widen$ cannot be strong.
With the notion of a strong widening at hand, we can now state the main theorem of this
section:
\begin{theorem}\label{t:termination}
    Let $\System$ be a side-effecting constraint system with domain $\Dom$ whose widening $\widen$ is strong,
    and assume that the solver only considers finite subsets $L\subseteq\Locals$ and $G\subseteq\Globals$.
    % and realizes updates $\rho$ only by calling \lstinline{update_globals}.
    For $i\geq 0$, assume that $\rho_i: G\to\Dom$, and $U_i$ = \lstinline{update_globals} $(x_i,\eta_i,\rho_i)$
    is determined according to \cref{lst:update_globals_robust_warrowing_imp}
    such that
    \[
    \rho_{i+1} = \rho_i\oplus\{g\mapsto d\mid (g,d)\in U_i\}
    \]
    where $\oplus$ denotes updating bindings in the left argument with new values provided in the right argument.
    Then, there is $j$, such that for all $i > j$, $U_i = \emptyset$.
    In other words, the number of updates to globals through \lstinline{update_globals} is finite.
    % supposing that only finitely many unknowns are analyzed and the host solver does not otherwise modify globals itself.
\end{theorem}
\begin{proof}
    Assume for a contradiction that the number of updates to globals is infinite,
    i.e., $\forall j: \exists i > j: U_i \neq \emptyset$.
    Since $G$ is finite, there must be some $g \in G$ such that there is an
    infinite sequence
    $i_1<\ldots<i_k<\ldots$ consisting of all $i_\kappa$
    such that $(g,d_\kappa)\in U_{i_\kappa}$ for some $d_\kappa$.
    In particular,
    $\rho_{i_\kappa+1}\,g\neq\rho_{i_\kappa}\,g$ for all $\kappa$.
    Since we assume $L$ to be finite, there must be some local $x \in L$ which occurs infinitely often among
    the $x_{i_\kappa}$.
    Let $j_1 < \ldots j_m<\ldots$ be the subsequence of the $i_\kappa$ where $x_{i_\kappa} = x$.
    Now consider the sequence of contributions $b_\mu$ of $x$ for $g$ in
    $\eta_{j_\mu}$,
    and $a_\mu,a'_\mu\in\Dom$ the values stored for $x$ in \lstinline|cmap[g][x]| before and
    after applying the update rule at $j_\mu$.
    Then for all $\mu\geq 1,$ $a_\mu\Box_\mu b_\mu = a'_\mu\sqsubseteq a_{\mu+1}$ for a sequence
    of operators $\Box_\mu\in \left\{\narrow, \join, \widen\right\}$.
    Since $\rho_{j_\mu+1}\,g\neq\rho_{j_\mu}\,g$,
    $b_\mu$ cannot be subsumed by $\rho_{j_\mu}\,g$.
    Consequently, none of the operators $\Box_\mu$ can equal $\sqcup$.
    As the number of switches from widening to narrowing is bounded, there is some
    $m$ such that for all $\mu\geq m$, $\Box_\mu\neq\narrow$, i.e., equals $\widen$.
    Recall that we assume the widening operator $\widen$ to be \emph{strong}.
    This means that for some $\mu\geq m$, $b_\mu\sqsubseteq a_\mu$ and thus
    also $b_\mu\sqsubseteq \rho_{j_\mu}\,g$.
    But then, according to our update rule, neither widening is applied to $b_\mu$
    nor any update occurs at $g$, i.e., $\rho_{j_\mu+1}\,g = \rho_{j_\mu}\,g$ ---
    contradiction.
    %
    % This concludes the proof.
\end{proof}

\subsection{Abstract Garbage Collection to Remove Withdrawn Contributions} \label{ss:garbage}

To our dismay, even after some parts of the program are found to be unreachable,
their contributions to globals may stick around as toxic trash and hurt precision for
relevant parts of the program.

% It may be of special interest to recover precision when contributions to globals disappear.

\begin{example}
Consider the multi-threaded program in \cref{lst:deaddirect} where the
shared variables are analyzed flow-insensitively.
Analyzers targeting concurrent programs often try to identify \emph{escaping} local variables, i.e., those
which may be accessed concurrently via pointers.
In line 7, the address of the local variable \lstinline|i| is written to the shared pointer variable \lstinline|a|.
This line of code is \emph{dead}, as the variable \lstinline|k| takes only values between $0$ and $10$ in the concrete.
When widening is applied at the loop head, however, the loop guard cannot improve the value of \lstinline|k|,
because it refers only to \lstinline|j|.
The widened value of \lstinline|k| can be observed for one solving iteration of the loop body.
During this iteration, line 7 provides $\{$\lstinline|&i|$\}$ as a contribution to the may-points-to set of
% things
\lstinline|a|.
% may point to.
%
% Assume that another thread, omitted here for brevity, makes two assertions about \lstinline|*a|.
As is, the assertion \lstinline|*a == 0| cannot be proven,
because the analysis cannot exclude that \lstinline|a| may point to \lstinline|i|.
% a catastrophic precision loss for \lstinline|*a|. % to become completely unknown.
%
% or just loses the upper bound?
%
% In a flow-insensitive analysis,
% the value produced by this increment is also its input, causing it to be widened.
% In summary,
The problem here is a contribution occurring in one iteration but not later-on.
\end{example}

While notions of abstract garbage collection have been investigated~\cite{MightS06,vanEs2019,JohnsonSEMH14,Germane020},
they are concerned with removing unnecessary variable bindings from local states --- which is not the issue at hand.
For a more detailed comparison, see \cref{sec:relatedwork}.

\begin{figure}
\begin{minipage}[t]{0.465\textwidth}
\begin{lstlisting}[style=ccode,label=lst:deaddirect,caption={Outdated contribution due to flow-sensitive precision recovery.}]
int zero = 0;
int *a = &zero;
void thread1() {
  int i = 1;
  for(int j = 0, k = 0; j < 10; j++) {
    if (k > 20)
      a = &i;
    k = j;
  }
}
void thread2(){
  assert(*a == 0);
}
\end{lstlisting}
\end{minipage}\hfill
\begin{minipage}[t]{0.465\textwidth}
\begin{lstlisting}[style=ccode,escapechar=@,firstnumber=3,label=lst:deadindirect,caption={Spurious contribution from a procedure in a
  context that is trash.}]
int zero = 0;
int *a = &zero;
@
\begin{tikzpicture}[remember picture,overlay]
    \highlight{fstart}{fend}
    \highlight{callstart}{callend}
\end{tikzpicture}
\tikzmark{fstart}
@void f(int k, int *i) {
  if (k > 20)
    a = &i;
}@\hspace*{\fill}\tikzmark{fend}@
void thread1() {
  int i = 1;
  for(int j = 0, k = 0; j < 10; j++) {
@\tikzmark{callstart}@    f(k, &i);@\hspace*{\fill}\tikzmark{callend}@
    k = j;
  }
}
void thread2(){
  assert(*a == 0);
}
\end{lstlisting}
\end{minipage}
\end{figure}

% \subsubsection{Omitted Side-Effects}

The constraint $f_x$ of a local unknown $x$ may generate contributions to different globals,
depending on the currently attained assignments $\sigma:\Locals\to\Dom$ and $\rho:\Globals\to\Dom$.
Thus, in general, $x$ may contribute a value $d$ to a global $g$ during one evaluation,
but no value (i.e., $\bot$) during a subsequent evaluation.
Accordingly, there is no need for the value of $g$ to subsume $d$.
We refer to this outdated contribution of $x$ to $g$ as \emph{withdrawn}.
% Our implementation uses a sparse map data-structure $\eta$ to represent
% of all contributions of $x$ to globals.
%
To enable the proposed update rules to actually remove withdrawn contributions of $x$ to $g$,
we propose an update rule \lstinline|update_globals$_\bot$|  in \cref{lst:omitted_contributions_imp}
that wraps around any of the previous update rules.
In the new update rule, we make withdrawn contributions from $x$ to $g$ explicit by
passing a contribution $(g, \bot)$ to the update rule called inside.
Technically, we introduce an internal data-structure \lstinline{old_contribs}
which provides for each local $x$ the \emph{set} of globals to which the previous call to \lstinline{update_globals}
for $x$ has provided a non-$\bot$ contribution.
The function \lstinline|update_globals$_\bot$| uses \lstinline{old_contribs} to retrieve
the set \lstinline|contribs|$_\bot$ of globals that received a non-$\bot$ contribution during the last call
of \lstinline|update_globals$_\bot$| for the local \lstinline|orig|.
For each such global $g$, a contribution $(g,\bot)$ is added to the set \lstinline|contribs| whenever
\lstinline|contribs| does not contain a contribution $g$ yet.
We denote this combination by \lstinline|contribs$_\bot$ $\sqcup$ contribs|.
In  examples, we internally use the update rule from \cref{lst:update_globals_robust_warrowing_imp}.
%
% This additional data-structure is internal to \lstinline{update_globals$_\bot$}.
% The resulting update rule is then given in \cref{lst:omitted_contributions_imp}.
%
Using \lstinline{update_globals$_\bot$} in the example program of \cref{lst:deaddirect},
the contribution of line 7 to \lstinline|a| is narrowed to $\bot$ and \lstinline|i| is
analyzed flow-sensitively. The analysis now finds that \lstinline[style=ccode]|a|
may only point to \lstinline[style=ccode]|zero| --- proving both assertions.
\begin{figure}[b]
    \begin{lstlisting}[style=ccode,caption={\label{lst:omitted_contributions_imp}Preprocessing to ensure removal of withdrawn contributions.}]
update_globals$_\bot$($\Locals$ orig, $2^{\Globals\times\Dom}$ contribs, $\Globals\to\Dom$ $\rho$) {
  static old_contribs; // Hashmap from $\Locals$ to sets of globals receiving
                       // a contribution at last evaluation
  contribs$_\bot$ = { (g,$\bot$) | g $\in$ old_contribs[x]  };
  old_contribs[x] = { g | (g,_) $\in$ contribs }; // Update old_contribs
  // If a global appears in both sets, join associated values and
  // keep all elements where global appears only in one argument
  contribs = contribs$_\bot$ $\sqcup$ contribs;
  return update_globals(orig, contribs, $\rho$);
}
    \end{lstlisting}
\end{figure}

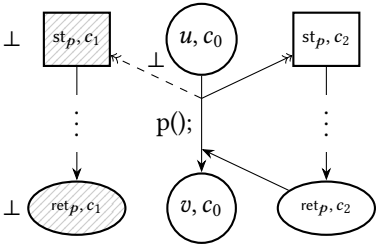
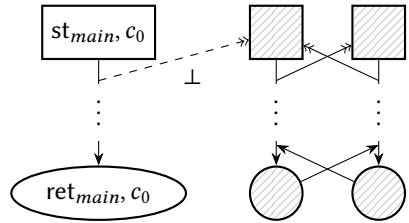
\begin{figure}[tbh]
    \centering
    \begin{subfigure}[t]{0.46\textwidth}
        \centering
        \begin{tikzpicture}
          \def\dis{1cm}
          \begin{scope}[every node/.style=cfgnode]
              \node (U1) at (0, 0) {$u,c_0$};
              \node (V1) [below=1.75cm of $(U1)$] {$v,c_0$};

              \node[shape=ellipse] (V2) [right=\dis of $(V1)$] {\tiny $\textsf{ret}_p,c_2$};
              \path let \p1 = (U1) in let \p2 = (V2) in node[shape=rectangle] (U2) at (\x2, \y1) {\scriptsize $\textsf{st}_p,c_2$};

              \node[shape=ellipse, pattern=north east lines,pattern color=black!20] (V3) [left=\dis of $(V1)$] {\tiny $\textsf{ret}_p,c_1$};
              \path let \p1 = (U1) in let \p2 = (V3) in node[shape=rectangle, pattern=north east lines,pattern color=black!20] (U3) at (\x2, \y1) {\scriptsize $\textsf{st}_p,c_1$};

              \node[draw=none] (b) [left=.5cm of $(U3)$] {$\bot$};
              \node[draw=none] (nb) [left=.5cm of $(V3)$] {$\bot$};
          \end{scope}
          \begin{scope}[>={Stealth[black]},
                      every edge/.style={draw=black}]
              \path [->] (U1) edge node [left] {p();} (V1);
              \path [->] (U2) edge node[rectangle,fill=white,minimum size=0cm,inner sep=0pt] {\raisebox{6pt}{$\vdots$}} (V2);
              \path [->] (U3) edge node[rectangle,fill=white,minimum size=0cm,inner sep=0pt] {\raisebox{6pt}{$\vdots$}} (V3);

              \path [-{to}{to}] ($(U1)!.35!(V1)$) edge node [above] {} (U2);
              \path [->]  (V2) edge node [above] {} ($(U1)!.65!(V1)$);

              \path [-{to}{to}, dashed] ($(U1)!.35!(V1)$) edge node [above] {$\bot$} (U3);
          \end{scope}
        \end{tikzpicture}
        \caption{A call to $p$ in context $c_2$ causes the withdrawal of a contribution to $p$ in the prior context $c_1$.
        % As $(\textsf{ret}_p,c_1)$ is no longer queried, this withdrawal renders all unknowns of $(v,c_1)$ for
	% program points $v$ of $p$ trash.
        Forward propagating solvers then set all unknowns $(v,c_1)$, $v$ program point of $p$, % including $(\textsf{ret}_p,c_1)$
        to $\bot$.
        % For \tdside, unless $(\textsf{ret}_p,c_1)$ is queried, however, the $\bot$ value at $(\textsf{st}_p,c_1)$ is not propagated.}
	% and the outdated writes to globals caused by the body of $p$ remain.
    }
    \label{fig:botnotprop}
    \end{subfigure}\hfill
    \begin{subfigure}[t]{0.46\textwidth}
        \centering
        \begin{tikzpicture}
            \def\dis{1cm}
            \begin{scope}[every node/.style=cfgnode]
                \node[shape=rectangle] (U1) at (0, 0) {$\textsf{st}_{main},c_0$};
                \node[shape=ellipse] (V1) [below=1.75cm of $(U1)$] {$\textsf{ret}_{main},c_0$};

                \node[shape=rectangle, pattern=north east lines,pattern color=black!20] (U2) [right=2cm of $(U1)$] {};
                \node[pattern=north east lines,pattern color=black!20] (V2) [below=1.75cm of $(U2)$] {};

                \node[shape=rectangle, pattern=north east lines,pattern color=black!20] (U3) [right=\dis of $(U2)$] {};
                \node[pattern=north east lines,pattern color=black!20] (V3) [below=1.75cm of $(U3)$] {};
            \end{scope}
            \begin{scope}[>={Stealth[black]},
                        every edge/.style={draw=black}]
                \path [->] (U1) edge node[rectangle,fill=white,minimum size=0cm,inner sep=0pt] {\raisebox{6pt}{$\vdots$}} (V1);
                \path [->] (U2) edge node[rectangle,fill=white,minimum size=0cm,inner sep=0pt] {\raisebox{6pt}{$\vdots$}} (V2);
                \path [->] (U3) edge node[rectangle,fill=white,minimum size=0cm,inner sep=0pt] {\raisebox{6pt}{$\vdots$}} (V3);

                \path [-{to}{to}, dashed] ($(U1)!.3!(V1)$) edge node [below right] {$\bot$} (U2);

                \path [-{to}{to}] ($(U2)!.3!(V2)$) edge node [above] {} (U3);
                \path [<-] ($(U2)!.7!(V2)$) edge node [above] {} (V3);

                \path [-{to}{to}] ($(U3)!.3!(V3)$) edge node [above] {} (U2);
                \path [<-] ($(U3)!.7!(V3)$) edge node [above] {} (V2);
            \end{scope}
        \end{tikzpicture}
        \caption{
	% $\bot$-ification
	Abstract garbage collection
	suffers from the same problem as reference counting. % in the realm of garbage collection.
        Unreachable unknowns may remain spuriously live by referring to each other cyclically.}
        \label{fig:callcycle}
    \end{subfigure}
    \caption{Two scenarios involving trash. \Cref{fig:botnotprop} displays the situation without
    % $\bot$-ificiation.
    abstract garbage collection.
    \Cref{fig:callcycle} displays a case where
    % $\bot$-ification cannot identify dead unknowns.
    abstract garbage collection fails.
    Hatched nodes represent unknowns that are trash. Arrows with double tips depict contributions to globals. Withdrawn contributions
    are dashed.}
\end{figure}

Withdrawing contributions may, in particular, impact context-sensitive analyses.
A procedure may at some point be analyzed in a context that later becomes irrelevant,
as the call-site creating this context turns out to be unreachable.
By analogy with contributions that are \emph{trash}, we also call an \emph{unknown} trash whenever it is found to be irrelevant
(relative to a given $\sigma$ and $\rho$).
The value of unknowns that are trash can be set to $\bot$ without compromising the analysis result.
Furthermore, all their contributions to globals should be withdrawn as these are not only \emph{trash}, but also \emph{toxic}, i.e., may cause
imprecision at other unknowns.
Taking out trash may cause further unknowns (local or global)
to become trash.
We call this successive removal of trash \emph{abstract garbage collection}.

For some solvers, the update rule~\ref{lst:omitted_contributions_imp} is sufficient to take out large fractions of the trash.
This is, for instance, the case for \emph{forward propagating} solvers (e.g., \slrp~\cite{apinis2013}): % Later cite again \cite{amato2016,FrielinghausSV18}:
When an unknown $x$ changes its value, all unknowns depending on $x$ are scheduled for re-evaluation.
If the start point $(\textsf{st}_p,c)$ receives the value $\bot$,
this value is eventually propagated to all unknowns $(v,c)$, $v\in N_p,$ and in this way, also
all contributions to globals triggered by their right-hand sides are revoked.

\begin{example}\label{ex:wewanthisbeforthestrangetdsidestuff}
Consider an interprocedural analysis with full context, i.e., with $\Context = \Dom$,
on \cref{lst:deadindirect}.
The variable \lstinline|i| only escapes if \lstinline|f| is called with \lstinline|k > 20|.
As before, the abstract value for \lstinline|k| is temporarily $[0,\infty]$ during one solving iteration
of the loop in \lstinline|thread1|.
Hence, \lstinline|f| is analyzed in some context $c$ with $k \mapsto [0,\infty]$.
In this context, \lstinline|f| provides a contribution to \lstinline|a| by which \lstinline|i| escapes. Subsequent
solving iterations no longer call \lstinline|f| in context $c$ and the value of $(\textsf{st}_f, c)$ becomes $\bot$.
For forward propagating solvers, eventually all program points in context $c$ become $\bot$ and the harmful contribution
by \lstinline|f| in context $c$ is withdrawn.
Thus, the assertion \lstinline|a == 0| can be shown.
\end{example}

The situation is different for solvers such as the local solver \tdside{} which
avoid eager re-evaluation of unknowns affected by another unknown changing its value and
instead only \emph{mark} all possibly affected unknowns as \emph{unstable}. Such unknowns
are only re-evaluated when later queried again.
In \refapp{appendix:td3agc},
we detail the problem and propose a solution to bring abstract garbage collection also to
\tdside{}-like solvers by triggering eager re-evaluation in some cases.

While we have demonstrated how abstract garbage collection
can be incorporated into any hosting solver,
% different fixpoint engines,
this solution falls short in the presence of cyclic garbage (see \cref{fig:callcycle}).
Some cyclic garbage may arise from the interprocedural analysis of directly recursive procedures:
The unknown $(\textsf{st}_p,c)$ for the start point of such a procedure $p$ in context $c$ may, after some steps,
receive all of its contributions from other unknowns $(u,c)$ where $u$ is in the same procedure, making it effectively trash.
For such cases, one solution is to distinguish between \emph{internal} and \emph{external} contributions to $(\textsf{st}_p,c)$:
Then, an unknown $(\textsf{st}_p,c)$ can be collected when all \emph{external} contributions to it are $\bot$.
As we expect abstract garbage collection to already yield meaningful results without removal of cyclic garbage,
we leave it for future work to experiment with such further extensions.

By using the update rule from \cref{lst:update_globals_robust_warrowing_imp} within \lstinline{update_globals$_\bot$},
not only soundness but also the termination guarantees carry over. This comes at the expense that ---
once gas is exhausted --- contributions are no longer withdrawn. Alternatively,
contributions could be withdrawn irrespective of the gas value and a separate
gas be introduced to bound how often an unknown becomes trash.

\section{Evaluation}
\label{sec:evaluation}
\newcommand\casecount{11222}
\newcommand\incomplete{1180}
\newcommand\incompleteourstwz{1012}
\newcommand\incompleteapinis{1179}
\newcommand\incompleteoursth{1009}
\newcommand\incompletebase{1008}

\newcommand\result{\the\numexpr\casecount - \incomplete\relax}
We implemented the update rules from \cref{lst:update_globals_robust_warrowing_imp} with reluctant widening and \cref{lst:omitted_contributions_imp} in the \Goblint\ analyzer written in \textsc{OCaml}.
As a baseline for our experiments we use the default update rule provided by \Goblint\ .
This update rule uses the first contribution to a global $g$ as initialization, joins the second increasing contribution, and widens with any further increasing contribution to $g$.
It is not able to shrink values of globals.
We are interested in the following research questions:
\begin{questions}[leftmargin=1cm]
  \item Is the new update rule (\cref{lst:update_globals_robust_warrowing_imp} with reluctant widening, referred to as \textbf{ours}) more precise than \Goblint's default update rule? \label{rq:base}
  \item Is the new update rule \textbf{ours} more precise than \textbf{apinis} as extracted from \cite{apinis2013}?
  	% (referred to as \textbf{apinis})?
  	\label{rq:apinis}
  \item Does limiting the W/N switches affect the analysis precision for update rule \textbf{ours}? \label{rq:warrowlimit}
  \item	Does the choice of the update rule impact termination behavior? \label{rq:termination}
  \item What are the impacts of abstract garbage collection introduced by \cref{lst:omitted_contributions_imp} (referred to as \textbf{ours$_\bot$})
  	% $\bot$-ification
	on analysis precision, run-time performance, and memory consumption? \label{rq:bot}
\end{questions}
We perform our experiments on a machine with an Intel Xeon 8260 CPU and 504 GB of RAM, where each instance of the \Goblint\ analyzer runs on a single core.
For \labelcref{rq:base,rq:apinis,rq:warrowlimit,rq:termination} each instance of \Goblint\ is limited to 15 GB of RAM.
\Goblint\ is configured to use intervals and points-to sets as domains, among others.
To study the precision of analyses with different update rules, we compare the numbers of unknowns for which one of the analyses yields a more precise abstract value.
Unknowns encountered in only one analysis are treated as $\bot$ in the other.
When an analysis has $n$ more precise, $m$ less precise and $k$ incomparable unknowns compared to the baseline,
%and $N$ total unknowns is strictly speaking incomparable to the baseline,
we compute the \emph{net} improvement by $\frac{n - m}{n+m+k}$.
% where $N$ is the total number of unknowns.
We consider a change in precision \emph{substantial} if at least 5\% of unknowns are improved/deteriorated.

For \labelcref{rq:base,rq:apinis,rq:warrowlimit}, we use a context-\emph{insensitive} analysis of
sequential code. Due to context-insensitivity, multiple distinct contributions can be expected
% from distinct call sites
to the start-points of procedures.
%
% For this analysis, we run
We run \Goblint\ with our new update rule \textbf{ours} and the update rule \textbf{apinis}
% extracted from \cite{apinis2013}
and compare both to a run with its default update rule.
As benchmarks we choose the ReachSafety category of the Competition on Software Verification~\cite{Beyer24} containing 11222 tasks.
On $\result$ of those tasks, all analyses run to completion within a time limit of 900s.
The results are shown in \cref{fig:reachsafety-nocontext}.
%There and in the following.

Concerning \ref{rq:base}, we find a net precision improvement in about 52\% of tasks when using the new update rule
\textbf{ours} --- both for W/N gas 3 (\textbf{ours}$_3$) and 20 (\textbf{ours}$_{20}$).
%
% then most of the improvements by the new rule are \emph{not} substantial.
13\% (\textbf{ours}$_3$) and 20\% (\textbf{ours}$_{20}$) of tasks show substantial improvements.\footnote{In
preliminary experiments, W/N gas exceeding 20 rarely improved net precision.}
Conversely, \textbf{ours}$_{\{3,20\}}$ loses net precision compared to the baseline on about 1\% of tasks ---
losses which are mostly substantial.
% of which almost all show substantial losses in precision.
%
Some of these outliers may be due to widenings avoided by the baseline.
This can occur when there are exactly two contributions
$d_1 \mathbin{\cancel{\sqsubseteq}} d_2$ from an origin to a global $g$.
Here, the baseline update rule marks $g$ as a widening point, but \emph{joins} $d_2$
without widening.
Instead, our update rules widen immediately.
The subcategory ReachSafety-Recursive greatly benefits from update rule \textbf{ours}
(see \cref{fig:reachsafety-recursive-nocontext}).
Both analyses provide substantial improvements in about 42\% of cases.
Interestingly, the share of net degraded cases is about 5\% which are all substantial.
As seen in \cref{ex:nocontext}, entries of recursive procedures are prone to
being widened, explaining these losses.
% the increased danger of precision loss.
%
We conclude w.r.t.\ \ref{rq:base}, that the new update rule significantly improves
the net precision in a considerable portion of cases. % where
% the gains clearly outweigh the losses.

W.r.t.\ \ref{rq:apinis}, the update rule \textbf{apinis}
produces a net improvement only in about 40\% of tasks, which are substantial only in 7\% of tasks.
At the same time, 19\% of the tasks have a net precision loss and 12\% a substantial loss.
% The number of tasks with a precision loss is far larger than with our new update rule, namely
% 18\% with a net precision loss and 11\% with substantial loss.
Our update rule has far fewer cases of net precision loss, improves a larger number of tasks (52\%), and tends to improve larger fractions of unknowns.
On the subset of ReachSafety-Recursive the update rule \textbf{apinis} achieves
substantial net precision improvements only in about 22\% of cases, compared to 42\% for \textbf{ours}.

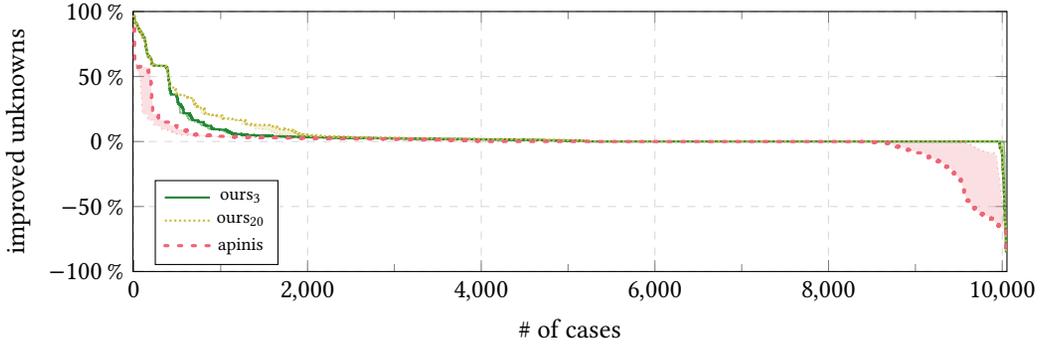
\begin{figure}
\centering
\pgfplotstableread[col sep=comma]{data/reachsafety-context0/plot2-ours3-vs-base.csv}{\tablecwth}
\pgfplotstableread[col sep=comma]{data/reachsafety-context0/plot2-ours3-vs-base-unambiguous.csv}{\tablecwthza}
\pgfplotstableread[col sep=comma]{data/reachsafety-context0/plot2-ours20-vs-base.csv}{\tablecwtw}
\pgfplotstableread[col sep=comma]{data/reachsafety-context0/plot2-ours20-vs-base-unambiguous.csv}{\tablecwtwza}
\pgfplotstableread[col sep=comma]{data/reachsafety-context0/plot2-apinis-vs-base.csv}{\tableapinis}
\pgfplotstableread[col sep=comma]{data/reachsafety-context0/plot2-apinis-vs-base-unambiguous.csv}{\tableapinisza}
% \pgfplotstableread[col sep=comma]{data/reachsafety-context0/cw3-vs-sides.csv}{\tablecwth}
% \pgfplotstableread[col sep=comma]{data/reachsafety-context0/cw3-vs-sides-unambiguous.csv}{\tablecwthza}
% \pgfplotstableread[col sep=comma]{data/reachsafety-context0/cw20-vs-sides.csv}{\tablecwtw}
% \pgfplotstableread[col sep=comma]{data/reachsafety-context0/cw20-vs-sides-unambiguous.csv}{\tablecwtwza}
% \pgfplotstableread[col sep=comma]{data/reachsafety-context0/apinis-vs-sides.csv}{\tableapinis}
% \pgfplotstableread[col sep=comma]{data/reachsafety-context0/apinis-vs-sides-unambiguous.csv}{\tableapinisza}
\begin{tikzpicture}
	\begin{axis}[
		width=0.95\linewidth,
		height=5cm,
		grid=major,
		yticklabel=\pgfmathparse{100*\tick}\pgfmathprintnumber{\pgfmathresult}\,\%,
		yticklabel style={/pgf/number format/.cd,fixed,precision=2},
		xticklabel style={/pgf/number format/fixed},
		scaled x ticks=false,
		grid style={dashed,gray!30},
		ylabel={improved unknowns},
		ylabel style={align=center},
		minor x tick num=1,
		xtick distance = 2000,
		xlabel={\# of cases},
		enlarge x limits={abs=\pgflinewidth},
		legend style={font=\scriptsize,at={(0.025,0.025)},anchor=south west},
		ymax=1,
		ymin=-1
	]

	\addplot+[const plot mark left,mark=none,thick,name path=CWTH]
		table[x=count,y=better] {\tablecwth};
	\addlegendentry{ours$_3$};
	\pgfplotsset{cycle list shift=-1}
	\addplot+[const plot mark left,forget plot,mark=none,opacity=0.5,name path=CWTHZA]
	table[x=count,y=better] {\tablecwthza};
	\addplot+[opacity=0.2,forget plot] fill between[of=CWTH and CWTHZA];
	\pgfplotsset{cycle list shift=0}

	\addplot+[const plot mark left,mark=none,densely dotted,thick,name path=CWTW]
		table[x=count,y=better] {\tablecwtw};
	\addlegendentry{ours$_{20}$};
	\pgfplotsset{cycle list shift=-1}
	\addplot+[const plot mark left,forget plot,mark=none,densely dotted,opacity=0.5,name path=CWTWZA]
	table[x=count,y=better] {\tablecwtwza};
	\addplot+[opacity=0.2,forget plot] fill between[of=CWTW and CWTWZA];
	\pgfplotsset{cycle list shift=0}

		\addplot+[const plot mark left,mark=none,line cap=round,loosely dotted,very thick,name path=APINIS]
		  table[x=count,y=better] {\tableapinis};
		\addlegendentry{apinis};
		\pgfplotsset{cycle list shift=-1}
		\addplot+[const plot mark left,forget plot,mark=none,dotted,opacity=0.5,name path=APINISZA]
		table[x=count,y=better] {\tableapinisza};
		\addplot+[opacity=0.2,forget plot] fill between[of=APINIS and APINISZA];
	  \end{axis}
	\end{tikzpicture}
  \caption{Net precision difference
  % of \textit{analysis}
  		to baseline for ReachSafety. All analyses are context-insensitive.
		   The x-axis enumerates the $\result$ cases that run to completion for all four analyses.
		   The cases are sorted by the net precision gain over the baseline.
		   % for each non-baseline analysis.
		   In some cases, the net precision difference \emph{flattens} cases with improved \emph{and}
		   degraded or with incomparable unknowns.
		   % in order to ascribe a single value to them.
		   To visualize the effect on the plot,
		   we include \emph{thin} plots in which all flattened cases are assigned a net precision difference
		   of zero, and shade the area between the two plots.
		   % To highlight the disparity between these plots and their flattened counterparts,
		   % the areas between them are filled.
		  }
  %         The x-axis is the number $n$ of cases. The value shown along the y-axis is $\text{max}_{p'\in\mathbf{R}}\{p'~|~|\{p_i(analysis, \text{base}) >= p'\}| >= n\}$.
  %         In other words, the maximal net precision gain for $n$ cases.
  \label{fig:reachsafety-nocontext}
  \end{figure}

Regarding \ref{rq:warrowlimit}, the update rule \textbf{ours} yields a net improvement on 52\% of tasks,
regardless of whether the W/N gas was set to 3 or 20.
However, with a W/N gas of 20, 20\% of tasks show substantial net improvements,
instead of 13\% with W/N gas of 3.
On ReachSafety-Recursive, with both gas values,
an equal amount of net precision is recovered (42\%).

For\ \ref{rq:termination}, to obtain more reliable performance numbers,
we ran runtime experiments with BenchExec~\cite{BeyerLW19}.\footnote{To
enable precision comparisons, the runs for \labelcref{rq:base,rq:apinis,rq:warrowlimit}
require marshalling incuring an additional overhead.}
Overall, the analyses have almost identical runtimes when they terminate, and the majority of timeouts
and other failures are shared between all considered analyses.
An exception is \textbf{ainis}, for which many outliers with high runtime overheads are observed.
The analyses with \Goblint's default update rule and \textbf{ours} with W/N gas of 3 and 20
fail to complete in a similar number of cases (972 vs. 975 or 974),
whereas \textbf{apinis} fails to complete more often (1003).
% This seems to indicate that the issue of precision recovery for interprocedural analyses
% with limited context-sensitivity is particularly pressing for recursive programs.

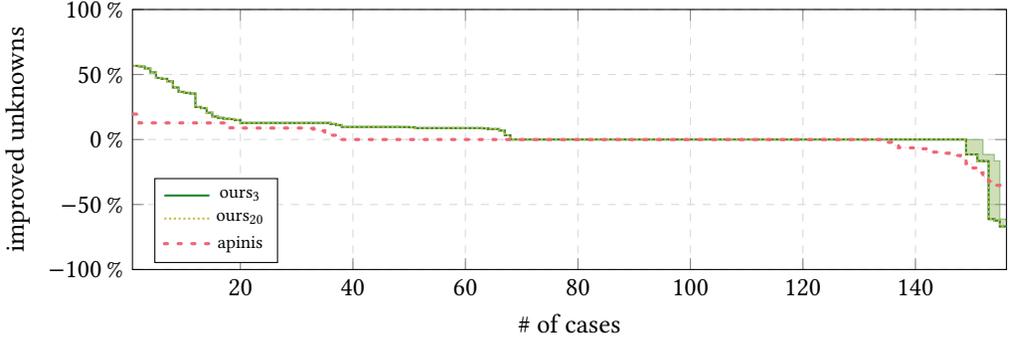
\begin{figure}
\centering
\pgfplotstableread[col sep=comma]{data/reachsafety-recursive-context0/plot3-ours3-vs-base.csv}{\tablecwth}
\pgfplotstableread[col sep=comma]{data/reachsafety-recursive-context0/plot3-ours3-vs-base-unambiguous.csv}{\tablecwthza}
\pgfplotstableread[col sep=comma]{data/reachsafety-recursive-context0/plot3-ours20-vs-base.csv}{\tablecwtw}
\pgfplotstableread[col sep=comma]{data/reachsafety-recursive-context0/plot3-ours20-vs-base-unambiguous.csv}{\tablecwtwza}
\pgfplotstableread[col sep=comma]{data/reachsafety-recursive-context0/plot3-apinis-vs-base.csv}{\tableapinis}
\pgfplotstableread[col sep=comma]{data/reachsafety-recursive-context0/plot3-apinis-vs-base-unambiguous.csv}{\tableapinisza}

\begin{tikzpicture}
	\begin{axis}[
		width=0.95\linewidth,
		height=5cm,
		grid=major,
		grid style={dashed,gray!30},
		ylabel={improved unknowns},
		ylabel style={align=center},
		yticklabel=\pgfmathparse{100*\tick}\pgfmathprintnumber{\pgfmathresult}\,\%,
		yticklabel style={/pgf/number format/.cd,fixed,precision=2},
		xlabel={\# of cases},
		enlarge x limits={abs=\pgflinewidth},
		legend style={font=\scriptsize,at={(0.025,0.025)},anchor=south west},
		ymax=1,
		ymin=-1
	]

	\addplot+[const plot mark left,mark=none,thick,name path=CWTH]
		table[x=count,y=better] {\tablecwth};
	\addlegendentry{ours$_3$};
	\pgfplotsset{cycle list shift=-1}
	\addplot+[const plot mark left,forget plot,mark=none,opacity=0.5,name path=CWTHZA]
	table[x=count,y=better] {\tablecwthza};
	\addplot+[opacity=0.2,forget plot] fill between[of=CWTH and CWTHZA];
	\pgfplotsset{cycle list shift=0}

	\addplot+[const plot mark left,mark=none,densely dotted,thick,name path=CWTW]
		table[x=count,y=better] {\tablecwtw};
	\addlegendentry{ours$_{20}$};
	\pgfplotsset{cycle list shift=-1}
	\addplot+[const plot mark left,forget plot,mark=none,densely dotted,opacity=0.5,name path=CWTWZA]
	table[x=count,y=better] {\tablecwtwza};
	\addplot+[opacity=0.2,forget plot] fill between[of=CWTW and CWTWZA];
	\pgfplotsset{cycle list shift=0}

	\addplot+[const plot mark left,mark=none,line cap=round,loosely dotted,very thick,name path=APINIS]
		table[x=count,y=better] {\tableapinis};
	\addlegendentry{apinis};
	\pgfplotsset{cycle list shift=-1}
	\addplot+[const plot mark left,forget plot,mark=none,dotted,opacity=0.5,name path=APINISZA]
	table[x=count,y=better] {\tableapinisza};
	\addplot+[opacity=0.2,forget plot] fill between[of=APINIS and APINISZA];
	\end{axis}
\end{tikzpicture}
\caption{Net precision difference to baseline for ReachSafety-Recursive. All analyses are context-insensitive.}
\label{fig:reachsafety-recursive-nocontext}
\end{figure}
% incomparable cases are described in detail by data/reachsafety-context0/ambiguity.csv

% Should Apinis even be considered in this section?
To answer \ref{rq:bot} and measure the effects of abstract garbage collection,
we pivot to analyses with \emph{full context}. The contributions to unknowns for
start-points of procedures then is given by the context.
% from distinct call-sites then necessarily agree.
%
Different analyses, however, may discover different contexts and thus may consider
different sets of unknowns.
% A direct comparison would thus conclude
% that the analysis which discovers fewer contexts is more precise, which is not a meaningful result.
Therefore, we restrict the comparison to globals
related to the \emph{flow-insensitive} analysis of variables possibly shared between threads.
For update rule \textbf{ours}, we only consider W/N gas of 3, as
more gas does not significantly impact net precision.
We conduct our experiment on two sets of large multi-threaded programs
established in prior literature.
The first set~\cite{schwarz2021,schwarz2023b} consists of six multi-threaded \textsc{Posix} programs and seven
Linux device drivers pre-processed by the LDV toolchain~\cite{DBLP:journals/pcs/ZakharovMMNPK15}.
The second set~\cite{DBLP:conf/icse/HongR23} was assembled from GitHub, originally in the context
of translation of {C} programs to \textsc{Rust}.
Program sizes range between a few hundred and a few thousand lines of code.
\Refapp{appendix:benchmarks} provides a detailed description of the benchmarks, as well as a list of those
excluded, e.g., because some configurations do not terminate within 15~min.
The analysis employs the relational thread-modular value-analysis by \citet{schwarz2023b}
instantiated with the cartesian interval domain.
\cref{fig:precision-multi-threaded} summarizes found differences in precision.
The figure shows per program and approach how many unknowns were improved, worsened, or become incomparable
relative to the baseline. The latter case only occurs with \textbf{apinis}.

When turning to net precision changes by aggregating precision losses and gains per program, we find that the update rule
\textbf{ours} improves precision in 24 out of 38 cases,
with 12 of these improvements being substantial.
\textbf{ours$_\bot$} improves precision for 32 programs, with 15 substantial improvements.
\textbf{ours$_\bot$} always improves at least as many unknowns as \textbf{ours} --- and often more.
\textbf{ours} and \textbf{ours$_\bot$} worsen net precision in only three and two cases, respectively,
with two cases showing substantial losses for both approaches.
The update rule \textbf{apinis}, on the other hand, improves net precision in only 11 cases,
with six of these being substantial.
Conversely, it worsens net precision in 18 cases, with 11 of these precision losses being substantial.

There are three outliers where all update rules are less precise than the baseline for more than
10\% of unknowns. For \texttt{ypbind}, the analysis fails to
resolve a pointer used to start a new thread leading to over 100 threads being analyzed
which indicates a critical loss of precision.
For all three programs, pinpointing the exact reason is out of reach,
given the size of the programs and the large number of steps performed in the fixpoint iteration.
However, given the non-monotonicity of the analysis, it is perhaps encouraging that only three such outliers are observed.

\begin{figure}
% \begin{figure}[t]
\begin{minipage}[b]{0.5\textwidth}
  \scalebox{0.8}{  % was 0.8
  \centering
  \pgfplotstableread[col sep=comma]{data/large-mutex-meet/plot4-differences.csv}{\tabledifferences}
  \begin{tikzpicture}[inner frame sep=0pt, outer frame sep=0pt]
    \pgfplotsset{
      compat=1.13,
      every axis/.style={
        bar shift=\shift,
        width=8cm,
        height=25cm,
        xbar stacked,
        bar width=5pt,
        grid=major,
        grid style={dashed,gray!30},
        xmin=-0.32,
        xmax=0.25,
        xlabel = {fraction of affected unknowns},
        xtick distance = 0.1,
        xticklabel=\pgfmathparse{100*\tick}\pgfmathprintnumber{\pgfmathresult}\,\%,
        xticklabel style={/pgf/number format/.cd,fixed,precision=2},
        ytick=data,
        yticklabels from table={\tabledifferences}{file},
        y tick label style={font=\small},
        enlarge y limits={abs=0.3cm},
        legend style={font=\small,at={(axis cs:-0.3,0.22)},anchor=south west},
        legend entries={ours,
                        ours$_\bot$,
                        apinis,
                        more precise,
                        less precise,
                        incomparable},
        legend columns = 1,
      },
    }
    \foreach \hide/\index/\a/\shift in {/0/ours/5pt,hide axis/3/ours-bot/0pt,hide axis/6/apinis/-5pt}{
      \begin{axis}[\hide]
        \addlegendimage{}
        \addlegendimage{pattern=north west lines}
        \addlegendimage{pattern=north east lines}
        \addlegendimage{fill=improved}
        \addlegendimage{fill=worse}
        \addlegendimage{fill=incomparable}

        \pgfplotsset{cycle list shift=\index, cycle list name=stackedbar}
        \addplot table [x expr=\thisrow{worse-\a}+0.000001, y expr=\coordindex] {\tabledifferences};
        \addplot table [x expr=-\thisrow{better-\a}, y expr=\coordindex] {\tabledifferences};
        \addplot table [x=incomparable-\a, y expr=\coordindex] {\tabledifferences};
      \end{axis}
    }
  \end{tikzpicture}}
  \centering
  \begin{minipage}{0.9\textwidth}
  \caption{Precision differences for large multi-threaded programs compared to the baseline.}
  \label{fig:precision-multi-threaded}
  \end{minipage}
\end{minipage}
\begin{minipage}[b]{0.48\textwidth}
\centering

\begin{minipage}{0.48\textwidth}
  \vspace{1.5em}
  \pgfplotstableread[col sep=comma]{data/large-mutex-meet/plot5-extra-short.csv}{\tabletimingsextrashort}
\scalebox{0.8}{
  \begin{tikzpicture}
  \begin{axis} [
    name=extrashort,
    width=3.5cm,
    height=23cm,
    xbar=\pgflinewidth,
    bar width=4pt,
    xmin=0,
    xmax=9,
    xlabel={runtime in [s].},
    grid=major,
    grid style={dashed,gray!30},
    ytick=data,
    yticklabels from table={\tabletimingsextrashort}{file},
    y tick label style={font=\small},
    enlarge y limits={abs=0.6cm},
    legend entries={apinis,
      ours$_{\bot}$,
      ours,
      base},
    reverse legend,
    legend style={font=\small,at={(70pt,-50pt)},anchor=north},
    % legend pos = south east,
    legend image code/.code={%
     \draw[#1,/tikz/.cd,yshift=-0.25em]
      (0cm,0cm) rectangle (6pt,9pt);},
    legend columns = 4
  ]
  % \foreach \a in {base,ours,ours-bot,apinis}{
  \foreach \a in {apinis,ours-bot,ours,base}{
    \addplot+ table [x=\a, y expr=\coordindex] {\tabletimingsextrashort};
  }
  \end{axis}
\end{tikzpicture}}
\end{minipage}
\begin{minipage}{0.48\textwidth}
\scalebox{0.8}{
  \begin{tikzpicture}
  \pgfplotstableread[col sep=comma]{data/large-mutex-meet/plot5-middle-low.csv}{\tabletimingsmiddlelow}
  \begin{axis} [
    xlabel={runtime in [s].},
    name=middlelow,
    width=3.5cm,
    height=10cm,
    xbar=\pgflinewidth,
    bar width=4pt,
    xmin=0,
    xmax=35,
    grid=major,
    grid style={dashed,gray!30},
    ytick=data,
    yticklabels from table={\tabletimingsmiddlelow}{file},
    y tick label style={font=\small},
    enlarge y limits={abs=0.6cm}
  ]
  \foreach \a in {apinis,ours-bot,ours,base}{
    \addplot+ table [x=\a, y expr=\coordindex] {\tabletimingsmiddlelow};
  }
  \end{axis}
  \pgfplotstableread[col sep=comma]{data/large-mutex-meet/plot5-middle.csv}{\tabletimingsmiddle}
  \begin{axis} [
    at={($(middlelow.north)+(0,1cm)$)},
    name=middle,
    anchor=south,
    width=3.5cm,
    height=10cm,
    xbar=\pgflinewidth,
    bar width=4pt,
    xmin=0,
    xmax=250,
    grid=major,
    grid style={dashed,gray!30},
    ytick=data,
    yticklabels from table={\tabletimingsmiddle}{file},
    y tick label style={font=\small},
    enlarge y limits={abs=0.6cm},
  ]
  \foreach \a in {apinis,ours-bot,ours,base}{
    \addplot+ table [x=\a, y expr=\coordindex] {\tabletimingsmiddle};
  }
  \end{axis}
  \pgfplotstableread[col sep=comma]{data/large-mutex-meet/plot5-extra.csv}{\tabletimingsextra}
  \begin{axis} [
    at={($(middle.north)+(0,1cm)$)},
    anchor=south,
    width=3.5cm,
    height=3cm,
    xbar=\pgflinewidth,
    bar width=4pt,
    xmin=0,
    xmax=1000,
    grid=major,
    grid style={dashed,gray!30},
    ytick=data,
    yticklabels from table={\tabletimingsextra}{file},
    y tick label style={font=\small},
    enlarge y limits={abs=0.1cm},
    xtick distance = 400
  ]
  \foreach \a in {apinis,ours-bot,ours,base}{
    \addplot+ table [x=\a, y expr=\coordindex] {\tabletimingsextra};
  }
  \end{axis}
\end{tikzpicture}}
\end{minipage}
\centering
\begin{minipage}{0.9\textwidth}
\caption{Runtimes of the analyses of large, multi-threaded benchmarks.}
\end{minipage}
\label{fig:timings-multi-threaded}
\end{minipage}

% \end{figure}
\end{figure}

We now consider the impact of abstract garbage collection on performance.
The run-times of the analyses are shown in \cref{fig:timings-multi-threaded}.
Again, we observe that update rule \textbf{ours} does not cause a significant overhead.
Abstract garbage collection generally comes with a moderate slowdown,
but causes a slowdown by a factor of 2.72 in the extreme.
This penalty on the run-time should be contrasted with the number of eliminated unknowns.
The next experiment therefore investigates the number of eliminated contexts per program
(\cref{fig:dead-multi-threaded}).
We find that the fraction of procedure contexts that
are identified as trash is unexpectedly high.
For 18 out of 38 cases, over 40\% of contexts are identified as trash.
During the analysis, locals at call-sites may change their abstract values,
such that procedures are re-analyzed for further contexts.
This may cause a cascading effect as the called procedures themselves may call other procedures.
When temporarily trash contexts are later rediscovered, the prior propagation of $\bot$ may not be advantageous.
For our experiments, there are very few intermittently trash contexts.
\Cref{fig:memory-multi-threaded} shows peak heap memory usage of the analyzer for each program and update rule.
Interestingly, the abstract garbage collection has an effect on the memory usage.
Withdrawing outdated writes may result in higher numbers of unknowns being set to $\bot$,
causing the garbage collector of the \textsc{OCaml} runtime to clean up their prior abstract values.
We focus on changes
of over 5\% compared to \textbf{ours}, as some deviations are expected due to the runtime system.
We find that \textbf{ours$_\bot$} \emph{increases} memory usage by 15\% for two programs, while it yields a reduction for 18 programs.
At the extremes, the magnitude of the reduction is considerable. For \verb|ypbind| and \verb|smtprc|, e.g.,
the heap memory footprint is roughly halved.

\begin{figure}
  \begin{minipage}[b]{0.64\textwidth}
    \begin{minipage}{0.49\textwidth}
    \centering
  \pgfplotstableread[col sep=comma]{data/large-mutex-meet/plot6-short.csv}{\tablememlo}
  \scalebox{0.8}{
  \begin{tikzpicture}
    \begin{axis} [
      name=shortheap,
      width=5cm,
      height=24cm,
      xbar=\pgflinewidth,
      bar width=4pt,
      xmin=0,
      xlabel={heap memory usage in [MB].},
      grid=major,
      grid style={dashed,gray!30},
      ytick=data,
      yticklabels from table={\tablememlo}{file},
      y tick label style={font=\small},
      enlarge y limits={abs=0.6cm},
      legend entries={apinis,
      ours$_{\bot}$,
      ours,
      base},
      reverse legend,
      legend pos=south east,
      legend image code/.code={%
       \draw[#1,/tikz/.cd,yshift=-0.25em]
        (0cm,0cm) rectangle (6pt,9pt);},
      legend style={font=\small}
    ]
    \foreach \a in {apinis,ours-bot,ours,base}{
      \addplot+ table [x=\a-memory, y expr=\coordindex] {\tablememlo};
    }
    \end{axis}
  \end{tikzpicture}}
\end{minipage}
\begin{minipage}{0.49\textwidth}
  \scalebox{0.8}{
    \begin{tikzpicture}
      \pgfplotstableread[col sep=comma]{data/large-mutex-meet/plot6-long-nofails.csv}{\tablememhi}
      \begin{axis} [
        anchor=west,
        width=5cm,
        height=10cm,
        xbar=\pgflinewidth,
        bar width=4pt,
        xmin=0,
        grid=major,
        grid style={dashed,gray!30},
        ytick=data,
        yticklabels from table={\tablememhi}{file},
        y tick label style={font=\small},
        enlarge y limits={abs=0.6cm},
      ]
      \foreach \a in {apinis,ours-bot,ours,base}{
        \addplot+ table [x=\a-memory, y expr=\coordindex] {\tablememhi};
      }
      \end{axis}
    \end{tikzpicture}}
  \end{minipage}
  \centering
  \begin{minipage}{0.8\textwidth}
  \caption{
    Memory footprint measured as the peak size of \textsc{OCaml}'s major heap throughout the execution. %, as reported by the \texttt{GC} module.
  }
  \label{fig:memory-multi-threaded}
  \end{minipage} %
  \end{minipage}
  \begin{minipage}[b]{0.35\textwidth}
  \scalebox{0.8}{
    \centering
    \pgfplotstableread[col sep=comma]{data/large-mutex-meet/plot7-dead.csv}{\tabledeadn}
    \begin{tikzpicture}
      \begin{axis}[
        width=5cm,
        height=22cm,
        grid=major,
        grid style={dashed,gray!30},
        xmin = 0,
        xmax = 1,
        xlabel={fraction of procedure\\entries set to $\bot$},
        xlabel style={align=center},
        xticklabel=\pgfmathparse{100*\tick}\pgfmathprintnumber{\pgfmathresult}\,\%,
        xticklabel style={/pgf/number format/.cd,fixed,precision=2},
        enlarge y limits={abs=0.6},
        xbar stacked,
        bar width=7pt,
        ytick=data,
        yticklabels from table={\tabledeadn}{file},
        y tick label style={font=\small},
        legend style={at={(0,-50pt)},nodes={scale=0.9},anchor=north, font=\small},  % at={(70pt,-50pt)},anchor=north},
      ]
      \addplot+ table[y expr=\coordindex, x expr=\thisrow{stillbot}/\thisrow{total}] {\tabledeadn};
      \addplot+ [fill=red] table[y expr=\coordindex, x expr=(\thisrow{botified}-\thisrow{stillbot})/\thisrow{total}] {\tabledeadn};
      \legend{,{intermittently trash}},
      \end{axis}
    \end{tikzpicture}}
    \begin{minipage}{\textwidth}
    \caption{Fraction of all pairs of procedure entry and context identified as trash.
    % Some contexts identified as trash are later rediscovered and reanalyzed.
    Pairs identified as trash intermittently,
    but are not trash in the final result are highlighted in red.}
    \label{fig:dead-multi-threaded}
    \end{minipage}
  \end{minipage}
\end{figure}

% There was a huge \iffalse block in here. If needed it can obtained from the git history, e.g., at 03dc4e38d1cfbbcbaed5aec33023bdc6698376bd

\paragraph{Threats to validity.}
The SV-COMP suite is an established benchmark for static analysis,
yet opinions on the generalizability to real-world programs differ.
%may be limited.
Therefore, the evaluation is supplemented with experiments on larger, real multi-threaded applications. % from the \Goblint{} benchmark suite.
The experiments are performed with one static analysis tool with specific widening operators and
flavors of context-sensitivity.
However, the analyses employ commonly used domains such as intervals and points-to sets.
Thus, we expect the results to be indicative of the impact on other mixed flow-sensitive analyzers.

\section{Related Work}
\label{sec:relatedwork}
% techniques for precision recovery.
% This section, perhaps surprisingly, first outlines its own structure, leading to the application of widening. Thus, the structure is $\top$.
% We group the related work as follows:
We first describe general techniques for improving precision, then turn to
frameworks for mixed-flow sensitivity, and lastly discuss other notions of abstract garbage collection.

\paragraph{General techniques.}
Widening and narrowing have been proposed by \citet{cousot1977,cousot1992} as general techniques
to speed up fixpoint iteration for abstract interpretation over domains with infinite ascending and descending
chains.
A variety of dedicated widening and narrowing operators as well as general techniques have been proposed to improve precision in abstract interpretation.
Here, we only mention the contributions that are most relevant for our setting.
In the context of the static analyzer \textsc{Astree}, \citet{blanchet2003} propose to delay widening
to obtain more precise results.
Arbitrary increase between widenings as we allow for reluctant widening is not considered.
Likewise, the notion of \emph{strong} widenings is new.
\citet{halbwachs2012} observe that, due to the inherent non-monotonicity of constraint solving with widening,
larger start values may lead to smaller or incomparable post fixpoints.
They propose solving the constraint system multiple times.
Each time,
% the solve is restarted,
the starting value is a modified version of the result of the prior solve.
Such restarting also has been advocated by \citet{amato2016}.
While these techniques work with a given constraint system, other approaches pioneered by trace partitioning \cite{Mauborgne05,Rival07} refine the constraint system.
Such a refinement of the unknowns or the domain allows keeping more information apart.
Trace partitioning has later been generalized to \emph{views} by \cite{KimRR18}, as well as to concurrency-sensitivity \cite{SchwarzE24}.
These techniques are orthogonal to ours and may be combined with our techniques to further
improve the precision -- at the price, though, of perhaps incurring an extra loss in efficiency.

\paragraph{Frameworks.}
Several analysis frameworks for imperative or object-oriented languages
perform a flow-insensitive \emph{points-to} analysis
in the style of \citet{andersen1994} or \citet{steensgaard1996} as the basis of more advanced analyses
of the program.
In this way, the \textsc{Java} analysis frameworks \textsc{Soot} \cite{lam2011}
and \textsc{FlowDroid} \cite{arzt2014} rely
on pointer analyses provided by \textsc{Spark} \cite{lhotak2003}, while \textsc{SootUp} \cite{karakaya2024}
relies on \textsc{QiLin} \cite{he2022a,he2022}.
\citet{CaiZ23} propose a call graph construction combining flow-insensitive points-to analysis with a flow-sensitive refinement.
Other analyses perform some flow-insensitive pre-analysis to, e.g., tune context-sensitivity
\cite{hassanshahi2017efficient}. % For the rest of them it's not true :(
%
% \todo[inline]{Should we mention CPAchecker as the most prominent example?}
%
The drawback of this approach is that the single abstraction for the heap is accumulated from the statements
of the program irrespective of whether these are reachable or not.
Analyzers such as \textsc{Astree} \cite{blanchet2003}, \textsc{Mopsa} \cite{monat2021},
\textsc{Frama-C} \cite{baudin2021}, or \textsc{Goblint} \cite{vojdani2016}
therefore perform their analyses of pointers on the fly alongside with other analyses.
Side-effecting constraint systems have been proposed by \citet{vene2003} for conveniently expressing
modular analyses of multi-threaded code where threads are analyzed flow-sensitively while
at the same time, information about shared data is collected flow-insensitively.
The potential of that approach was further elaborated in \cite{apinis2012}.
%
% Here, Goblint is explicitly cited, so it's not anonymized!
To increase efficiency at the expense of precision, \textsc{Goblint} also offers the option
to analyze global data-structures flow-insensitively even when the program is single-threaded.
This possibility is naturally supported by side-effecting constraint systems and particularly
useful when analyzing programs incrementally \cite{erhard2024a},
which amounts to a partial application of store widening~\cite{HornM10}. % TODO: Also cite erhard2015
This idea also is employed by \citet{nicolay2019}.
% StievenartNMR19 use the same framework as nicolay2019, but for multi-threaded code
% 		-> We could think about re-sorting citations so these two de Roover-things appear together?

An attempt of retaining precision during flow-insensitive accumulation of abstract values is proposed by
\citet{apinis2013}. % amato2016
They present a local solver \slrp\ which supports widening and narrowing and
also side-effects to unknowns triggered by the evaluation of right-hand sides of constraints.
Side-effect contributions to $y$ from some unknown $x$ are tracked via a helper unknown $(x,y)$
where a dedicated set maintains for $y$ the set of unknowns which have produced side-effects to $y$ so far.
The combined \emph{warrowing} operator then is applied only after the join of the individual contributions to $y$.
Withdrawal of contributions is not considered.
In contrast, our new update rule does not introduce auxiliary unknowns. It is generic in making
only little assumptions on the hosting solver. It performs widening and
narrowing per origin, it allows withdrawing outdated contributions and supports abstract garbage collection.
Another approach for dealing with globals is followed by Antoine Min\'e in \cite{mine2012,mine2014}.
When analyzing concurrent programs, Min\'e flow-insensitively collects \emph{interferences} between threads,
discovered during a flow-sensitive analysis of threads. An outer fixpoint computation then is performed until the set
of interferences stabilizes. This approach suggests an alternative fixpoint engine for combinations of
flow-sensitive with flow-insensitive analyses. Issues like withdrawal of contributions or abstract
garbage collection, however, are not discussed.
\citet{StievenartNMR19} also accumulate flow-insensitive information about global variables (via so-called \emph{effects})
during a flow-sensitive abstract interpretation of a multi-threaded system. Like the work by Min\'e, they rely on
a nested fixpoint formulation. This framework assumes finite lattices, and thus questions about widening and narrowing,
withdrawal of contributions or abstract garbage collection in our sense do not arise.
\citet{EsPSR20} report on the design of \textsc{MAF}, a framework which encompasses this analysis of multi-threaded code
as well as the analysis by \citet{nicolay2019}, and the encountered engineering challenges.
Recently, \citet{keidel2024} have
formalized the \emph{blackboard architecture} used by the tool \textsc{Opal} \cite{afonso2016,helm2020,roth2021}.
% as a generalization of flow-insensitive \emph{points-to} analysis in the style, e.g., of \citet{Andersen}.
This framework considers a global store of \emph{entities} with \emph{kinds}.
Each entity and kind is mapped to a property from some lattice corresponding to the kind.
A collection of monotonic update functions may query the store and produce contributions which are added
to the current store by \emph{join}. The analyzer is meant to compute a fixpoint, i.e.,
a store subsuming some initial store which is invariant under updates.
In our terminology, this framework deals with globals only.
Control-flow must be encoded by tagging entities, e.g., with program points and
designing the update functions appropriately.
Examples of analyses expressible in this framework are
a pointer analysis which mutually depends on a call graph analysis,
a reflection analysis which reuses the pointer analysis and on top of that,
a field and object immutability analysis.
\citet{keidel2024} assume monotonic constraints and analysis domains of finite height ---
assumptions which are not met by more expressive abstract domains.
Widening, narrowing, or withdrawal of outdated contributions are not considered.

\paragraph{Abstract garbage collection.}
\citet{mangal2014} remark that
context-sensitive interprocedural analysis may analyze procedures for contexts
which in the end do not contribute to the final result. No on-the-fly method, though,
is provided to collect the corresponding trash.
No general distinction between globals and locals is introduced. Accordingly, also no
general framework for accumulation at globals and withdrawal of contributions is provided.
Building on previous work by \citet{MightS06}, in \cite{vanEs2019}, abstract reference counting is used to identify abstract garbage
within an abstract store maintained by an abstract interpreter.
These stores, however, are analyzed flow-sensitively and thus correspond to what we call locals.
Their version of abstract garbage collection thus tries to improve \emph{within} abstract representations
of local program states. As far we can see, globals in our sense are not of concern.
Neither widening nor narrowing are explicitly dealt with. Their work has recently also found applications in other analyzers \cite{MonatOM20}.
Other recent work \cite{JohnsonSEMH14,Germane020} also relies on the notion of abstract garbage collection proposed by \citet{MightS06}.
A different form of garbage collection inside abstract values is provided by \citet{he2023} for the IFDS framework
for interprocedural analysis \cite{reps1995}. In that framework, interprocedural analysis is dissolved into the propagation
of flow facts across an \emph{exploded supergraph}. Garbage collection here aims at removing edges within that graph
which have become irrelevant.
While in spirit related to our ideas, the technique is only applicable to a restricted class of analyses and thus not as generic as ours.
Also, it is not clear how mixed flow-sensitivity can be supported.
%
% REALLY TRUE??
%
The \textsc{Goblint} system supports mixed flow-sensitive analyses, but so far has not supported abstract garbage collection during the analysis.
As it uses the top-down solver \tdside{} as its standard fixpoint engine,
the values for irrelevant unknowns need not satisfy the constraint system \cite{fecht1999,seidl2021}
and are only removed in a post-processing phase.
Their possible contributions to globals, however, are not withdrawn.
%%%

% \todo[inline]{More related work on abstract GC? More on analysis of concurrent systems?}
%               -> I think our treatment of related work is quite extensive already

%Other approaches to improve the precision of flow-insensitive focus on
%synchronized access to shared variables in multi-threaded programs \cite{vojdani2016,schwarz2021}.
%Their mechanisms are aimed at using less flow-insensitive information and producing fewer side-effects.
%Our mechanisms are orthogonal to this and can and should be used complementarily.

\section{Conclusion}
\label{sec:conclusion}
We have presented update rules to enhance the precision of mixed flow-sensitive analyses, where global unknowns are treated flow-insensitively.
They apply widening and narrowing to the individual contributions of locals.
Our more sophisticated update rules apply widening reluctantly, i.e., only if a new contribution is not subsumed by existing contributions.
We enhanced our approach with a form of abstract garbage collection to take out the toxic trash.
We have compared our more sophisticated update rules as well as an update rule extracted from earlier work specific to one solver to the treatment of globals as provided by the \Goblint{} framework.
We found for context-insensitive analyses that our new update rule results in improvements for half of the ReachSafety category of the SV-COMP benchmark suite over \Goblint's default rule, with a moderate impact on runtime.
For context-sensitive analyses, an overwhelming majority of contexts can be collected as abstract garbage during the analysis run.
The price for that is an increase in runtime by a factor of less than 3.
Future work may experiment with combining the update rules presented here with advanced widening strategies such as
delayed widening and evaluate their impact.
How to extend the techniques to collecting \emph{cyclic} abstract garbage is still to be investigated.
One may explore whether techniques
from the garbage collection literature %(see, e.g., \cite{jones2023garbage})
can be applied to collect such forms of trash.
Further, for incremental mixed flow-sensitive analyses, the new techniques for abstract garbage collection may also help withdraw contributions from outdated code.

\begin{acks}
  We thank the anonymous reviewers for their valuable feedback, which was instrumental in polishing the paper.
DeepSeek was used to help with the creation of scripts for extracting precision data from logfiles.
This work was supported in part by
Deutsche Forschungsgemeinschaft (DFG) -- 378803395/2428 \textsc{ConVeY}.

\end{acks}

\section*{Data Availability}
A virtual machine that allows for reproduction of our experiments is available on Zenodo~\cite{artifact}.
It contains the source code and binaries of \textsc{Goblint}, the benchmark programs, and scripts
to reproduce our results and figures.

\bibliographystyle{ACM-Reference-Format}
\bibliography{ref}


\begin{thebibliography}{73}

%%% the Web address.

\ifx \showCODEN    \undefined \def \showCODEN     #1{\unskip}     \fi
\ifx \showDOI      \undefined \def \showDOI       #1{#1}\fi
\ifx \showISBNx    \undefined \def \showISBNx     #1{\unskip}     \fi
\ifx \showISBNxiii \undefined \def \showISBNxiii  #1{\unskip}     \fi
\ifx \showISSN     \undefined \def \showISSN      #1{\unskip}     \fi
\ifx \showLCCN     \undefined \def \showLCCN      #1{\unskip}     \fi
\ifx \shownote     \undefined \def \shownote      #1{#1}          \fi
\ifx \showarticletitle \undefined \def \showarticletitle #1{#1}   \fi
\ifx \showURL      \undefined \def \showURL       {\relax}        \fi
% invisible to TeX
\providecommand\bibfield[2]{#2}
\providecommand\bibinfo[2]{#2}
\providecommand\natexlab[1]{#1}
\providecommand\showeprint[2][]{arXiv:#2}

\bibitem[\protect\citeauthoryear{Afonso, Bianchi, Fratantonio, Doup\'e, Polino,
  {de Geus}, Kruegel, and Vigna}{Afonso et~al\mbox{.}}{2016}]%
        {afonso2016}
\bibfield{author}{\bibinfo{person}{Vitor Afonso}, \bibinfo{person}{Antonio
  Bianchi}, \bibinfo{person}{Yannick Fratantonio}, \bibinfo{person}{Adam
  Doup\'e}, \bibinfo{person}{Mario Polino}, \bibinfo{person}{Paulo {de Geus}},
  \bibinfo{person}{Christopher Kruegel}, {and} \bibinfo{person}{Giovanni
  Vigna}.} \bibinfo{year}{2016}\natexlab{}.
\newblock \showarticletitle{Going native: Using a large-scale analysis of
  android apps to create a practical native-code sandboxing policy}. In
  \bibinfo{booktitle}{\emph{23rd Annual Network and Distributed System Security
  Symposium, NDSS 2016}}. \bibinfo{publisher}{The Internet Society}.
\newblock


\bibitem[\protect\citeauthoryear{Amato, Scozzari, Seidl, Apinis, and
  Vojdani}{Amato et~al\mbox{.}}{2016}]%
        {amato2016}
\bibfield{author}{\bibinfo{person}{Gianluca Amato}, \bibinfo{person}{Francesca
  Scozzari}, \bibinfo{person}{Helmut Seidl}, \bibinfo{person}{Kalmer Apinis},
  {and} \bibinfo{person}{Vesal Vojdani}.} \bibinfo{year}{2016}\natexlab{}.
\newblock \showarticletitle{Efficiently intertwining widening and narrowing}.
\newblock \bibinfo{journal}{\emph{Science of Computer Programming}}
  \bibinfo{volume}{120} (\bibinfo{year}{2016}), \bibinfo{pages}{1--24}.
\newblock
\showISSN{0167-6423}
\urldef\tempurl%
\url{https://doi.org/10.1016/j.scico.2015.12.005}
\showDOI{\tempurl}


\bibitem[\protect\citeauthoryear{Andersen}{Andersen}{1994}]%
        {andersen1994}
\bibfield{author}{\bibinfo{person}{Lars~Ole Andersen}.}
  \bibinfo{year}{1994}\natexlab{}.
\newblock \emph{\bibinfo{title}{Program Analysis and Specialization for the C
  Programming Language}}.
\newblock \bibinfo{thesistype}{Ph.D. Dissertation}.
\newblock


\bibitem[\protect\citeauthoryear{Apinis, Seidl, and Vojdani}{Apinis
  et~al\mbox{.}}{2012}]%
        {apinis2012}
\bibfield{author}{\bibinfo{person}{Kalmer Apinis}, \bibinfo{person}{Helmut
  Seidl}, {and} \bibinfo{person}{Vesal Vojdani}.}
  \bibinfo{year}{2012}\natexlab{}.
\newblock \showarticletitle{Side-Effecting Constraint Systems: A Swiss Army
  Knife for Program Analysis}. In \bibinfo{booktitle}{\emph{Programming
  Languages and Systems}}, \bibfield{editor}{\bibinfo{person}{Ranjit Jhala}
  {and} \bibinfo{person}{Atsushi Igarashi}} (Eds.).
  \bibinfo{publisher}{Springer Berlin Heidelberg}, \bibinfo{address}{Berlin,
  Heidelberg}, \bibinfo{pages}{157--172}.
\newblock
\showISBNx{978-3-642-35182-2}


\bibitem[\protect\citeauthoryear{Apinis, Seidl, and Vojdani}{Apinis
  et~al\mbox{.}}{2013}]%
        {apinis2013}
\bibfield{author}{\bibinfo{person}{Kalmer Apinis}, \bibinfo{person}{Helmut
  Seidl}, {and} \bibinfo{person}{Vesal Vojdani}.}
  \bibinfo{year}{2013}\natexlab{}.
\newblock \showarticletitle{How to combine widening and narrowing for
  non-monotonic systems of equations}. In \bibinfo{booktitle}{\emph{Proceedings
  of the 34th ACM SIGPLAN Conference on Programming Language Design and
  Implementation}} \emph{(\bibinfo{series}{PLDI '13})}.
  \bibinfo{publisher}{Association for Computing Machinery},
  \bibinfo{address}{New York, NY, USA}, \bibinfo{pages}{377--386}.
\newblock
\showISBNx{9781450320146}
\urldef\tempurl%
\url{https://doi.org/10.1145/2491956.2462190}
\showDOI{\tempurl}


\bibitem[\protect\citeauthoryear{Arzt, Rasthofer, Fritz, Bodden, Bartel, Klein,
  Traon, Octeau, and McDaniel}{Arzt et~al\mbox{.}}{2014}]%
        {arzt2014}
\bibfield{author}{\bibinfo{person}{Steven Arzt}, \bibinfo{person}{Siegfried
  Rasthofer}, \bibinfo{person}{Christian Fritz}, \bibinfo{person}{Eric Bodden},
  \bibinfo{person}{Alexandre Bartel}, \bibinfo{person}{Jacques Klein},
  \bibinfo{person}{Yves~Le Traon}, \bibinfo{person}{Damien Octeau}, {and}
  \bibinfo{person}{Patrick~D. McDaniel}.} \bibinfo{year}{2014}\natexlab{}.
\newblock \showarticletitle{FlowDroid: precise context, flow, field,
  object-sensitive and lifecycle-aware taint analysis for Android apps}. In
  \bibinfo{booktitle}{\emph{{ACM} {SIGPLAN} Conference on Programming Language
  Design and Implementation, {PLDI} '14, Edinburgh, United Kingdom - June 09 -
  11, 2014}}, \bibfield{editor}{\bibinfo{person}{Michael F.~P. O'Boyle} {and}
  \bibinfo{person}{Keshav Pingali}} (Eds.). \bibinfo{publisher}{{ACM}},
  \bibinfo{pages}{259--269}.
\newblock
\urldef\tempurl%
\url{https://doi.org/10.1145/2594291.2594299}
\showDOI{\tempurl}


\bibitem[\protect\citeauthoryear{Baudin, Bobot, B{\"{u}}hler, Correnson,
  Kirchner, Kosmatov, Maroneze, Perrelle, Prevosto, Signoles, and
  Williams}{Baudin et~al\mbox{.}}{2021}]%
        {baudin2021}
\bibfield{author}{\bibinfo{person}{Patrick Baudin},
  \bibinfo{person}{Fran{\c{c}}ois Bobot}, \bibinfo{person}{David B{\"{u}}hler},
  \bibinfo{person}{Lo{\"{\i}}c Correnson}, \bibinfo{person}{Florent Kirchner},
  \bibinfo{person}{Nikolai Kosmatov}, \bibinfo{person}{Andr{\'{e}} Maroneze},
  \bibinfo{person}{Valentin Perrelle}, \bibinfo{person}{Virgile Prevosto},
  \bibinfo{person}{Julien Signoles}, {and} \bibinfo{person}{Nicky Williams}.}
  \bibinfo{year}{2021}\natexlab{}.
\newblock \showarticletitle{The dogged pursuit of bug-free {C} programs: the
  Frama-C software analysis platform}.
\newblock \bibinfo{journal}{\emph{Commun. {ACM}}} \bibinfo{volume}{64},
  \bibinfo{number}{8} (\bibinfo{year}{2021}), \bibinfo{pages}{56--68}.
\newblock
\urldef\tempurl%
\url{https://doi.org/10.1145/3470569}
\showDOI{\tempurl}


\bibitem[\protect\citeauthoryear{Beyer}{Beyer}{2024}]%
        {Beyer24}
\bibfield{author}{\bibinfo{person}{Dirk Beyer}.}
  \bibinfo{year}{2024}\natexlab{}.
\newblock \showarticletitle{State of the Art in Software Verification and
  Witness Validation: {SV-COMP} 2024}. In \bibinfo{booktitle}{\emph{Tools and
  Algorithms for the Construction and Analysis of Systems - 30th International
  Conference, {TACAS} 2024, Held as Part of the European Joint Conferences on
  Theory and Practice of Software, {ETAPS} 2024, Luxembourg City, Luxembourg,
  April 6-11, 2024, Proceedings, Part {III}}} \emph{(\bibinfo{series}{Lecture
  Notes in Computer Science})}, \bibfield{editor}{\bibinfo{person}{Bernd
  Finkbeiner} {and} \bibinfo{person}{Laura Kov{\'{a}}cs}} (Eds.),
  Vol.~\bibinfo{volume}{14572}. \bibinfo{publisher}{Springer},
  \bibinfo{pages}{299--329}.
\newblock
\urldef\tempurl%
\url{https://doi.org/10.1007/978-3-031-57256-2\_15}
\showDOI{\tempurl}


\bibitem[\protect\citeauthoryear{Beyer, L{\"{o}}we, and Wendler}{Beyer
  et~al\mbox{.}}{2019}]%
        {BeyerLW19}
\bibfield{author}{\bibinfo{person}{Dirk Beyer}, \bibinfo{person}{Stefan
  L{\"{o}}we}, {and} \bibinfo{person}{Philipp Wendler}.}
  \bibinfo{year}{2019}\natexlab{}.
\newblock \showarticletitle{Reliable benchmarking: requirements and solutions}.
\newblock \bibinfo{journal}{\emph{Int. J. Softw. Tools Technol. Transf.}}
  \bibinfo{volume}{21}, \bibinfo{number}{1} (\bibinfo{year}{2019}),
  \bibinfo{pages}{1--29}.
\newblock
\urldef\tempurl%
\url{https://doi.org/10.1007/S10009-017-0469-Y}
\showDOI{\tempurl}


\bibitem[\protect\citeauthoryear{Blanchet, Cousot, Cousot, Feret, Mauborgne,
  Min\'{e}, Monniaux, and Rival}{Blanchet et~al\mbox{.}}{2003}]%
        {blanchet2003}
\bibfield{author}{\bibinfo{person}{Bruno Blanchet}, \bibinfo{person}{Patrick
  Cousot}, \bibinfo{person}{Radhia Cousot}, \bibinfo{person}{J\'{e}rome Feret},
  \bibinfo{person}{Laurent Mauborgne}, \bibinfo{person}{Antoine Min\'{e}},
  \bibinfo{person}{David Monniaux}, {and} \bibinfo{person}{Xavier Rival}.}
  \bibinfo{year}{2003}\natexlab{}.
\newblock \showarticletitle{A static analyzer for large safety-critical
  software}. In \bibinfo{booktitle}{\emph{Proceedings of the ACM SIGPLAN 2003
  Conference on Programming Language Design and Implementation}}
  \emph{(\bibinfo{series}{PLDI '03})}. \bibinfo{publisher}{Association for
  Computing Machinery}, \bibinfo{address}{New York, NY, USA},
  \bibinfo{pages}{196--207}.
\newblock
\showISBNx{1581136625}
\urldef\tempurl%
\url{https://doi.org/10.1145/781131.781153}
\showDOI{\tempurl}


\bibitem[\protect\citeauthoryear{Cai and Zhang}{Cai and Zhang}{2023}]%
        {CaiZ23}
\bibfield{author}{\bibinfo{person}{Yuandao Cai} {and} \bibinfo{person}{Charles
  Zhang}.} \bibinfo{year}{2023}\natexlab{}.
\newblock \showarticletitle{A Cocktail Approach to Practical Call Graph
  Construction}.
\newblock \bibinfo{journal}{\emph{Proc. {ACM} Program. Lang.}}
  \bibinfo{volume}{7}, \bibinfo{number}{{OOPSLA2}} (\bibinfo{year}{2023}),
  \bibinfo{pages}{1001--1033}.
\newblock
\urldef\tempurl%
\url{https://doi.org/10.1145/3622833}
\showDOI{\tempurl}


\bibitem[\protect\citeauthoryear{Cortesi and Zanioli}{Cortesi and
  Zanioli}{2011}]%
        {CortesiZ11}
\bibfield{author}{\bibinfo{person}{Agostino Cortesi} {and}
  \bibinfo{person}{Matteo Zanioli}.} \bibinfo{year}{2011}\natexlab{}.
\newblock \showarticletitle{Widening and narrowing operators for abstract
  interpretation}.
\newblock \bibinfo{journal}{\emph{Comput. Lang. Syst. Struct.}}
  \bibinfo{volume}{37}, \bibinfo{number}{1} (\bibinfo{year}{2011}),
  \bibinfo{pages}{24--42}.
\newblock
\urldef\tempurl%
\url{https://doi.org/10.1016/J.CL.2010.09.001}
\showDOI{\tempurl}


\bibitem[\protect\citeauthoryear{Cousot and Cousot}{Cousot and Cousot}{1977}]%
        {cousot1977}
\bibfield{author}{\bibinfo{person}{Patrick Cousot} {and}
  \bibinfo{person}{Radhia Cousot}.} \bibinfo{year}{1977}\natexlab{}.
\newblock \showarticletitle{Abstract interpretation: a unified lattice model
  for static analysis of programs by construction or approximation of
  fixpoints}. In \bibinfo{booktitle}{\emph{Proceedings of the 4th ACM
  SIGACT-SIGPLAN Symposium on Principles of Programming Languages}}
  \emph{(\bibinfo{series}{POPL '77})}. \bibinfo{publisher}{Association for
  Computing Machinery}, \bibinfo{address}{New York, NY, USA},
  \bibinfo{pages}{238--252}.
\newblock
\showISBNx{9781450373500}
\urldef\tempurl%
\url{https://doi.org/10.1145/512950.512973}
\showDOI{\tempurl}


\bibitem[\protect\citeauthoryear{Cousot and Cousot}{Cousot and Cousot}{1992a}]%
        {cousot1992}
\bibfield{author}{\bibinfo{person}{Patrick Cousot} {and}
  \bibinfo{person}{Radhia Cousot}.} \bibinfo{year}{1992}\natexlab{a}.
\newblock \showarticletitle{Abstract Interpretation Frameworks}.
\newblock \bibinfo{journal}{\emph{J. Log. Comput.}} \bibinfo{volume}{2},
  \bibinfo{number}{4} (\bibinfo{year}{1992}), \bibinfo{pages}{511--547}.
\newblock
\urldef\tempurl%
\url{https://doi.org/10.1093/LOGCOM/2.4.511}
\showDOI{\tempurl}


\bibitem[\protect\citeauthoryear{Cousot and Cousot}{Cousot and Cousot}{1992b}]%
        {CousotC92Comparing}
\bibfield{author}{\bibinfo{person}{Patrick Cousot} {and}
  \bibinfo{person}{Radhia Cousot}.} \bibinfo{year}{1992}\natexlab{b}.
\newblock \showarticletitle{Comparing the Galois Connection and
  Widening/Narrowing Approaches to Abstract Interpretation}. In
  \bibinfo{booktitle}{\emph{Programming Language Implementation and Logic
  Programming, 4th International Symposium, PLILP'92, Leuven, Belgium, August
  26-28, 1992, Proceedings}} \emph{(\bibinfo{series}{Lecture Notes in Computer
  Science})}, \bibfield{editor}{\bibinfo{person}{Maurice Bruynooghe} {and}
  \bibinfo{person}{Martin Wirsing}} (Eds.), Vol.~\bibinfo{volume}{631}.
  \bibinfo{publisher}{Springer}, \bibinfo{pages}{269--295}.
\newblock
\urldef\tempurl%
\url{https://doi.org/10.1007/3-540-55844-6\_142}
\showDOI{\tempurl}


\bibitem[\protect\citeauthoryear{Darais, Labich, Nguyen, and Horn}{Darais
  et~al\mbox{.}}{2017}]%
        {DaraisLNH17}
\bibfield{author}{\bibinfo{person}{David Darais}, \bibinfo{person}{Nicholas
  Labich}, \bibinfo{person}{Phuc~C. Nguyen}, {and} \bibinfo{person}{David~Van
  Horn}.} \bibinfo{year}{2017}\natexlab{}.
\newblock \showarticletitle{Abstracting definitional interpreters (functional
  pearl)}.
\newblock \bibinfo{journal}{\emph{Proc. {ACM} Program. Lang.}}
  \bibinfo{volume}{1}, \bibinfo{number}{{ICFP}} (\bibinfo{year}{2017}),
  \bibinfo{pages}{12:1--12:25}.
\newblock
\urldef\tempurl%
\url{https://doi.org/10.1145/3110256}
\showDOI{\tempurl}


\bibitem[\protect\citeauthoryear{Durumeric, Wustrow, and Halderman}{Durumeric
  et~al\mbox{.}}{2013}]%
        {durumeric2013zmap}
\bibfield{author}{\bibinfo{person}{Zakir Durumeric}, \bibinfo{person}{Eric
  Wustrow}, {and} \bibinfo{person}{J~Alex Halderman}.}
  \bibinfo{year}{2013}\natexlab{}.
\newblock \showarticletitle{{ZMap}: Fast Internet-wide scanning and its
  security applications}. In \bibinfo{booktitle}{\emph{22nd USENIX Security
  Symposium}}.
\newblock


\bibitem[\protect\citeauthoryear{Erhard, Saan, Tilscher, Schwarz, Holter,
  Vojdani, and Seidl}{Erhard et~al\mbox{.}}{2024a}]%
        {erhard2024a}
\bibfield{author}{\bibinfo{person}{Julian Erhard}, \bibinfo{person}{Simmo
  Saan}, \bibinfo{person}{Sarah Tilscher}, \bibinfo{person}{Michael Schwarz},
  \bibinfo{person}{Karoliine Holter}, \bibinfo{person}{Vesal Vojdani}, {and}
  \bibinfo{person}{Helmut Seidl}.} \bibinfo{year}{2024}\natexlab{a}.
\newblock \showarticletitle{Interactive abstract interpretation: reanalyzing
  multithreaded C programs for cheap}.
\newblock \bibinfo{journal}{\emph{International Journal on Software Tools for
  Technology Transfer}} (\bibinfo{year}{2024}), \bibinfo{pages}{1--21}.
\newblock
\urldef\tempurl%
\url{https://doi.org/10.1007/s10009-024-00768-9}
\showURL{%
\tempurl}


\bibitem[\protect\citeauthoryear{Erhard, Schinabeck, Schwarz, and Seidl}{Erhard
  et~al\mbox{.}}{2024b}]%
        {erhard2024}
\bibfield{author}{\bibinfo{person}{Julian Erhard},
  \bibinfo{person}{Johanna~Franziska Schinabeck}, \bibinfo{person}{Michael
  Schwarz}, {and} \bibinfo{person}{Helmut Seidl}.}
  \bibinfo{year}{2024}\natexlab{b}.
\newblock \showarticletitle{When to Stop Going Down the Rabbit Hole: Taming
  Context-Sensitivity on the Fly}. In \bibinfo{booktitle}{\emph{Proceedings of
  the 13th ACM SIGPLAN International Workshop on the State Of the Art in
  Program Analysis}} \emph{(\bibinfo{series}{SOAP 2024})}.
  \bibinfo{publisher}{Association for Computing Machinery},
  \bibinfo{address}{New York, NY, USA}, \bibinfo{pages}{35--44}.
\newblock
\showISBNx{9798400706219}
\urldef\tempurl%
\url{https://doi.org/10.1145/3652588.3663321}
\showDOI{\tempurl}


\bibitem[\protect\citeauthoryear{Fecht and Seidl}{Fecht and Seidl}{1999}]%
        {fecht1999}
\bibfield{author}{\bibinfo{person}{Christian Fecht} {and}
  \bibinfo{person}{Helmut Seidl}.} \bibinfo{year}{1999}\natexlab{}.
\newblock \showarticletitle{A Faster Solver for General Systems of Equations}.
\newblock \bibinfo{journal}{\emph{Sci. Comput. Program.}} \bibinfo{volume}{35},
  \bibinfo{number}{2} (\bibinfo{year}{1999}), \bibinfo{pages}{137--161}.
\newblock
\urldef\tempurl%
\url{https://doi.org/10.1016/S0167-6423(99)00009-X}
\showDOI{\tempurl}


\bibitem[\protect\citeauthoryear{Germane and Adams}{Germane and Adams}{2020}]%
        {Germane020}
\bibfield{author}{\bibinfo{person}{Kimball Germane} {and}
  \bibinfo{person}{Michael~D. Adams}.} \bibinfo{year}{2020}\natexlab{}.
\newblock \showarticletitle{Liberate Abstract Garbage Collection from the Stack
  by Decomposing the Heap}. In \bibinfo{booktitle}{\emph{Programming Languages
  and Systems - 29th European Symposium on Programming, {ESOP} 2020, Held as
  Part of the European Joint Conferences on Theory and Practice of Software,
  {ETAPS} 2020, Dublin, Ireland, April 25-30, 2020, Proceedings}}
  \emph{(\bibinfo{series}{Lecture Notes in Computer Science})},
  \bibfield{editor}{\bibinfo{person}{Peter M{\"{u}}ller}} (Ed.),
  Vol.~\bibinfo{volume}{12075}. \bibinfo{publisher}{Springer},
  \bibinfo{pages}{197--223}.
\newblock
\urldef\tempurl%
\url{https://doi.org/10.1007/978-3-030-44914-8\_8}
\showDOI{\tempurl}


\bibitem[\protect\citeauthoryear{Glaze, Labich, Might, and Horn}{Glaze
  et~al\mbox{.}}{2013}]%
        {JohnsonLMH13}
\bibfield{author}{\bibinfo{person}{Dionna~Amalie Glaze},
  \bibinfo{person}{Nicholas Labich}, \bibinfo{person}{Matthew Might}, {and}
  \bibinfo{person}{David~Van Horn}.} \bibinfo{year}{2013}\natexlab{}.
\newblock \showarticletitle{Optimizing abstract abstract machines}. In
  \bibinfo{booktitle}{\emph{{ACM} {SIGPLAN} International Conference on
  Functional Programming, ICFP'13, Boston, MA, {USA} - September 25 - 27,
  2013}}, \bibfield{editor}{\bibinfo{person}{Greg Morrisett} {and}
  \bibinfo{person}{Tarmo Uustalu}} (Eds.). \bibinfo{publisher}{{ACM}},
  \bibinfo{pages}{443--454}.
\newblock
\urldef\tempurl%
\url{https://doi.org/10.1145/2500365.2500604}
\showDOI{\tempurl}


\bibitem[\protect\citeauthoryear{Glaze, Sergey, Earl, Might, and Horn}{Glaze
  et~al\mbox{.}}{2014}]%
        {JohnsonSEMH14}
\bibfield{author}{\bibinfo{person}{Dionna~Amalie Glaze}, \bibinfo{person}{Ilya
  Sergey}, \bibinfo{person}{Christopher Earl}, \bibinfo{person}{Matthew Might},
  {and} \bibinfo{person}{David~Van Horn}.} \bibinfo{year}{2014}\natexlab{}.
\newblock \showarticletitle{Pushdown flow analysis with abstract garbage
  collection}.
\newblock \bibinfo{journal}{\emph{J. Funct. Program.}} \bibinfo{volume}{24},
  \bibinfo{number}{2-3} (\bibinfo{year}{2014}), \bibinfo{pages}{218--283}.
\newblock
\urldef\tempurl%
\url{https://doi.org/10.1017/S0956796814000100}
\showDOI{\tempurl}


\bibitem[\protect\citeauthoryear{Gotsman, Berdine, Cook, and Sagiv}{Gotsman
  et~al\mbox{.}}{2007}]%
        {Gotsman07}
\bibfield{author}{\bibinfo{person}{Alexey Gotsman}, \bibinfo{person}{Josh
  Berdine}, \bibinfo{person}{Byron Cook}, {and} \bibinfo{person}{Mooly Sagiv}.}
  \bibinfo{year}{2007}\natexlab{}.
\newblock \showarticletitle{Thread-modular shape analysis}. In
  \bibinfo{booktitle}{\emph{PLDI '07}}. \bibinfo{publisher}{ACM},
  \bibinfo{pages}{266--277}.
\newblock
\urldef\tempurl%
\url{https://doi.org/10.1145/1250734.1250765}
\showDOI{\tempurl}


\bibitem[\protect\citeauthoryear{Halbwachs and Henry}{Halbwachs and
  Henry}{2012}]%
        {halbwachs2012}
\bibfield{author}{\bibinfo{person}{Nicolas Halbwachs} {and}
  \bibinfo{person}{Julien Henry}.} \bibinfo{year}{2012}\natexlab{}.
\newblock \showarticletitle{When the Decreasing Sequence Fails}. In
  \bibinfo{booktitle}{\emph{Static Analysis - 19th International Symposium,
  {SAS} 2012, Deauville, France, September 11-13, 2012. Proceedings}}
  \emph{(\bibinfo{series}{Lecture Notes in Computer Science})},
  \bibfield{editor}{\bibinfo{person}{Antoine Min{\'{e}}} {and}
  \bibinfo{person}{David Schmidt}} (Eds.), Vol.~\bibinfo{volume}{7460}.
  \bibinfo{publisher}{Springer}, \bibinfo{pages}{198--213}.
\newblock
\urldef\tempurl%
\url{https://doi.org/10.1007/978-3-642-33125-1\_15}
\showDOI{\tempurl}


\bibitem[\protect\citeauthoryear{Hassanshahi, Ramesh, Krishnan, Scholz, and
  Lu}{Hassanshahi et~al\mbox{.}}{2017}]%
        {hassanshahi2017efficient}
\bibfield{author}{\bibinfo{person}{Behnaz Hassanshahi},
  \bibinfo{person}{Raghavendra~Kagalavadi Ramesh}, \bibinfo{person}{Padmanabhan
  Krishnan}, \bibinfo{person}{Bernhard Scholz}, {and} \bibinfo{person}{Yi Lu}.}
  \bibinfo{year}{2017}\natexlab{}.
\newblock \showarticletitle{An efficient tunable selective points-to analysis
  for large codebases}. In \bibinfo{booktitle}{\emph{Proceedings of the 6th ACM
  SIGPLAN International Workshop on State of the Art in Program Analysis}}.
  \bibinfo{pages}{13--18}.
\newblock
\urldef\tempurl%
\url{https://doi.org/10.1145/3088515.3088519}
\showDOI{\tempurl}


\bibitem[\protect\citeauthoryear{He}{He}{2022}]%
        {he2022}
\bibfield{author}{\bibinfo{person}{Dongjie He}.}
  \bibinfo{year}{2022}\natexlab{}.
\newblock \emph{\bibinfo{title}{Efficient and Precise Pointer Analysis with
  Fine-Grained Context Sensitivity}}.
\newblock \bibinfo{thesistype}{Ph.D. Dissertation}. \bibinfo{school}{University
  of New South Wales, Sydney, Australia}.
\newblock
\urldef\tempurl%
\url{https://doi.org/10.26190/UNSWORKS/23987}
\showDOI{\tempurl}


\bibitem[\protect\citeauthoryear{He, Gui, Gao, and Xue}{He
  et~al\mbox{.}}{2023}]%
        {he2023}
\bibfield{author}{\bibinfo{person}{Dongjie He}, \bibinfo{person}{Yujiang Gui},
  \bibinfo{person}{Yaoqing Gao}, {and} \bibinfo{person}{Jingling Xue}.}
  \bibinfo{year}{2023}\natexlab{}.
\newblock \showarticletitle{Reducing the Memory Footprint of {IFDS}-Based
  Data-Flow Analyses using Fine-Grained Garbage Collection}. In
  \bibinfo{booktitle}{\emph{Proceedings of the 32nd ACM SIGSOFT International
  Symposium on Software Testing and Analysis}} \emph{(\bibinfo{series}{ISSTA
  2023})}. \bibinfo{publisher}{Association for Computing Machinery},
  \bibinfo{address}{New York, NY, USA}, \bibinfo{pages}{101--113}.
\newblock
\showISBNx{9798400702211}
\urldef\tempurl%
\url{https://doi.org/10.1145/3597926.3598041}
\showDOI{\tempurl}


\bibitem[\protect\citeauthoryear{He, Lu, and Xue}{He et~al\mbox{.}}{2022}]%
        {he2022a}
\bibfield{author}{\bibinfo{person}{Dongjie He}, \bibinfo{person}{Jingbo Lu},
  {and} \bibinfo{person}{Jingling Xue}.} \bibinfo{year}{2022}\natexlab{}.
\newblock \showarticletitle{{\textsc{Qilin}}: {A} New Framework For Supporting
  Fine-Grained Context-Sensitivity in {\textsc{Java}} Pointer Analysis}. In
  \bibinfo{booktitle}{\emph{36th European Conference on Object-Oriented
  Programming, {ECOOP} 2022, June 6-10, 2022, Berlin, Germany}}
  \emph{(\bibinfo{series}{LIPIcs})}, \bibfield{editor}{\bibinfo{person}{Karim
  Ali} {and} \bibinfo{person}{Jan Vitek}} (Eds.), Vol.~\bibinfo{volume}{222}.
  \bibinfo{publisher}{Schloss Dagstuhl - Leibniz-Zentrum f{\"{u}}r Informatik},
  \bibinfo{pages}{30:1--30:29}.
\newblock
\urldef\tempurl%
\url{https://doi.org/10.4230/LIPICS.ECOOP.2022.30}
\showDOI{\tempurl}


\bibitem[\protect\citeauthoryear{Helm, K\"ubler, Reif, Eichberg, and
  Mezini}{Helm et~al\mbox{.}}{2020}]%
        {helm2020}
\bibfield{author}{\bibinfo{person}{Dominik Helm}, \bibinfo{person}{Florian
  K\"ubler}, \bibinfo{person}{Michael Reif}, \bibinfo{person}{Michael
  Eichberg}, {and} \bibinfo{person}{Mira Mezini}.}
  \bibinfo{year}{2020}\natexlab{}.
\newblock \showarticletitle{Modular collaborative program analysis in {OPAL}}.
  In \bibinfo{booktitle}{\emph{ESEC/FSE '20: 28th ACM Joint European Software
  Engineering Conference and Symposium on the Foundations of Software
  Engineering}}, \bibfield{editor}{\bibinfo{person}{P.~Devanbu},
  \bibinfo{person}{M.B. Cohen}, {and} \bibinfo{person}{T.~Zimmermann}} (Eds.).
  \bibinfo{publisher}{ACM}, \bibinfo{pages}{184--196}.
\newblock
\urldef\tempurl%
\url{https://doi.org/10.1145/3368089.3409765}
\showDOI{\tempurl}


\bibitem[\protect\citeauthoryear{Heo, Oh, and Yang}{Heo et~al\mbox{.}}{2019}]%
        {HeoOY19}
\bibfield{author}{\bibinfo{person}{Kihong Heo}, \bibinfo{person}{Hakjoo Oh},
  {and} \bibinfo{person}{Hongseok Yang}.} \bibinfo{year}{2019}\natexlab{}.
\newblock \showarticletitle{Resource-aware program analysis via online
  abstraction coarsening}. In \bibinfo{booktitle}{\emph{Proceedings of the 41st
  International Conference on Software Engineering, {ICSE} 2019, Montreal, QC,
  Canada, May 25-31, 2019}}, \bibfield{editor}{\bibinfo{person}{Joanne~M.
  Atlee}, \bibinfo{person}{Tevfik Bultan}, {and} \bibinfo{person}{Jon Whittle}}
  (Eds.). \bibinfo{publisher}{{IEEE} / {ACM}}, \bibinfo{pages}{94--104}.
\newblock
\urldef\tempurl%
\url{https://doi.org/10.1109/ICSE.2019.00027}
\showDOI{\tempurl}


\bibitem[\protect\citeauthoryear{Hong and Ryu}{Hong and Ryu}{2023}]%
        {DBLP:conf/icse/HongR23}
\bibfield{author}{\bibinfo{person}{Jaemin Hong} {and} \bibinfo{person}{Sukyoung
  Ryu}.} \bibinfo{year}{2023}\natexlab{}.
\newblock \showarticletitle{Concrat: An Automatic C-to-Rust Lock {API}
  Translator for Concurrent Programs}. In \bibinfo{booktitle}{\emph{45th
  {IEEE/ACM} International Conference on Software Engineering, {ICSE} 2023,
  Melbourne, Australia, May 14-20, 2023}}. \bibinfo{publisher}{{IEEE}},
  \bibinfo{pages}{716--728}.
\newblock
\urldef\tempurl%
\url{https://doi.org/10.1109/ICSE48619.2023.00069}
\showDOI{\tempurl}


\bibitem[\protect\citeauthoryear{Karakaya, Schott, Klauke, Bodden, Schmidt,
  Luo, and He}{Karakaya et~al\mbox{.}}{2024}]%
        {karakaya2024}
\bibfield{author}{\bibinfo{person}{Kadiray Karakaya}, \bibinfo{person}{Stefan
  Schott}, \bibinfo{person}{Jonas Klauke}, \bibinfo{person}{Eric Bodden},
  \bibinfo{person}{Markus Schmidt}, \bibinfo{person}{Linghui Luo}, {and}
  \bibinfo{person}{Dongjie He}.} \bibinfo{year}{2024}\natexlab{}.
\newblock \showarticletitle{{\textsc{SootUp}}: {A} Redesign of the
  {\textsc{Soot}} Static Analysis Framework}. In
  \bibinfo{booktitle}{\emph{Tools and Algorithms for the Construction and
  Analysis of Systems - 30th International Conference, {TACAS} 2024, Held as
  Part of the European Joint Conferences on Theory and Practice of Software,
  {ETAPS} 2024, Luxembourg City, Luxembourg, April 6-11, 2024, Proceedings,
  Part {I}}} \emph{(\bibinfo{series}{LNCS 14570})},
  \bibfield{editor}{\bibinfo{person}{Bernd Finkbeiner} {and}
  \bibinfo{person}{Laura Kov{\'{a}}cs}} (Eds.). \bibinfo{publisher}{Springer},
  \bibinfo{pages}{229--247}.
\newblock
\urldef\tempurl%
\url{https://doi.org/10.1007/978-3-031-57246-3\_13}
\showDOI{\tempurl}


\bibitem[\protect\citeauthoryear{Keidel, Helm, Roth, and Mezini}{Keidel
  et~al\mbox{.}}{2024}]%
        {keidel2024}
\bibfield{author}{\bibinfo{person}{Sven Keidel}, \bibinfo{person}{Dominik
  Helm}, \bibinfo{person}{Tobias Roth}, {and} \bibinfo{person}{Mira Mezini}.}
  \bibinfo{year}{2024}\natexlab{}.
\newblock \showarticletitle{A Modular Soundness Theory for the Blackboard
  Analysis Architecture}. In \bibinfo{booktitle}{\emph{Programming Languages
  and Systems}}, \bibfield{editor}{\bibinfo{person}{Stephanie Weirich}} (Ed.).
  \bibinfo{publisher}{Springer Nature Switzerland}, \bibinfo{address}{Cham},
  \bibinfo{pages}{361--390}.
\newblock
\showISBNx{978-3-031-57267-8}


\bibitem[\protect\citeauthoryear{Kim, Rival, and Ryu}{Kim
  et~al\mbox{.}}{2018}]%
        {KimRR18}
\bibfield{author}{\bibinfo{person}{Se{-}Won Kim}, \bibinfo{person}{Xavier
  Rival}, {and} \bibinfo{person}{Sukyoung Ryu}.}
  \bibinfo{year}{2018}\natexlab{}.
\newblock \showarticletitle{A Theoretical Foundation of Sensitivity in an
  Abstract Interpretation Framework}.
\newblock \bibinfo{journal}{\emph{{ACM} Trans. Program. Lang. Syst.}}
  \bibinfo{volume}{40}, \bibinfo{number}{3} (\bibinfo{year}{2018}),
  \bibinfo{pages}{13:1--13:44}.
\newblock
\urldef\tempurl%
\url{https://doi.org/10.1145/3230624}
\showDOI{\tempurl}


\bibitem[\protect\citeauthoryear{Lam, Bodden, Lhot\'ak, and Hendren}{Lam
  et~al\mbox{.}}{2011}]%
        {lam2011}
\bibfield{author}{\bibinfo{person}{Patrick Lam}, \bibinfo{person}{Eric Bodden},
  \bibinfo{person}{Ondrej Lhot\'ak}, {and} \bibinfo{person}{Laurie~J.
  Hendren}.} \bibinfo{year}{2011}\natexlab{}.
\newblock \showarticletitle{The {\textsc{Soot}} framework for {\textsc{Java}}
  program analysis: a retrospective}. In \bibinfo{booktitle}{\emph{Cetus Users
  and Compiler Infrastructure Workshop (CETUS 2011)}}.
\newblock


\bibitem[\protect\citeauthoryear{Lhot{\'{a}}k and Chung}{Lhot{\'{a}}k and
  Chung}{2011}]%
        {LhotakC11}
\bibfield{author}{\bibinfo{person}{Ondrej Lhot{\'{a}}k} {and}
  \bibinfo{person}{Kwok{-}Chiang~Andrew Chung}.}
  \bibinfo{year}{2011}\natexlab{}.
\newblock \showarticletitle{Points-to analysis with efficient strong updates}.
  In \bibinfo{booktitle}{\emph{Proceedings of the 38th {ACM} {SIGPLAN-SIGACT}
  Symposium on Principles of Programming Languages, {POPL} 2011, Austin, TX,
  USA, January 26-28, 2011}}, \bibfield{editor}{\bibinfo{person}{Thomas Ball}
  {and} \bibinfo{person}{Mooly Sagiv}} (Eds.). \bibinfo{publisher}{{ACM}},
  \bibinfo{pages}{3--16}.
\newblock
\urldef\tempurl%
\url{https://doi.org/10.1145/1926385.1926389}
\showDOI{\tempurl}


\bibitem[\protect\citeauthoryear{Lhot{\'{a}}k and Hendren}{Lhot{\'{a}}k and
  Hendren}{2003}]%
        {lhotak2003}
\bibfield{author}{\bibinfo{person}{Ondrej Lhot{\'{a}}k} {and}
  \bibinfo{person}{Laurie~J. Hendren}.} \bibinfo{year}{2003}\natexlab{}.
\newblock \showarticletitle{Scaling Java Points-to Analysis Using {SPARK}}. In
  \bibinfo{booktitle}{\emph{Compiler Construction, 12th International
  Conference, {CC} 2003, Held as Part of the Joint European Conferences on
  Theory and Practice of Software, {ETAPS} 2003, Warsaw, Poland, April 7-11,
  2003, Proceedings}} \emph{(\bibinfo{series}{Lecture Notes in Computer
  Science})}, \bibfield{editor}{\bibinfo{person}{G{\"{o}}rel Hedin}} (Ed.),
  Vol.~\bibinfo{volume}{2622}. \bibinfo{publisher}{Springer},
  \bibinfo{pages}{153--169}.
\newblock
\urldef\tempurl%
\url{https://doi.org/10.1007/3-540-36579-6\_12}
\showDOI{\tempurl}


\bibitem[\protect\citeauthoryear{Li}{Li}{2018}]%
        {DBLP:journals/bioinformatics/Li18}
\bibfield{author}{\bibinfo{person}{Heng Li}.} \bibinfo{year}{2018}\natexlab{}.
\newblock \showarticletitle{Minimap2: pairwise alignment for nucleotide
  sequences}.
\newblock \bibinfo{journal}{\emph{Bioinform.}} \bibinfo{volume}{34},
  \bibinfo{number}{18} (\bibinfo{year}{2018}), \bibinfo{pages}{3094--3100}.
\newblock
\urldef\tempurl%
\url{https://doi.org/10.1093/BIOINFORMATICS/BTY191}
\showDOI{\tempurl}


\bibitem[\protect\citeauthoryear{Mangal, Naik, and Yang}{Mangal
  et~al\mbox{.}}{2014}]%
        {mangal2014}
\bibfield{author}{\bibinfo{person}{Ravi Mangal}, \bibinfo{person}{Mayur Naik},
  {and} \bibinfo{person}{Hongseok Yang}.} \bibinfo{year}{2014}\natexlab{}.
\newblock \showarticletitle{A Correspondence between Two Approaches to
  Interprocedural Analysis in the Presence of Join}. In
  \bibinfo{booktitle}{\emph{Programming Languages and Systems - 23rd European
  Symposium on Programming, {ESOP} 2014, Held as Part of the European Joint
  Conferences on Theory and Practice of Software, {ETAPS} 2014, Grenoble,
  France, April 5-13, 2014, Proceedings}} \emph{(\bibinfo{series}{Lecture Notes
  in Computer Science})}, \bibfield{editor}{\bibinfo{person}{Zhong Shao}}
  (Ed.), Vol.~\bibinfo{volume}{8410}. \bibinfo{publisher}{Springer},
  \bibinfo{pages}{513--533}.
\newblock
\urldef\tempurl%
\url{https://doi.org/10.1007/978-3-642-54833-8\_27}
\showDOI{\tempurl}


\bibitem[\protect\citeauthoryear{Mauborgne and Rival}{Mauborgne and
  Rival}{2005}]%
        {Mauborgne05}
\bibfield{author}{\bibinfo{person}{Laurent Mauborgne} {and}
  \bibinfo{person}{Xavier Rival}.} \bibinfo{year}{2005}\natexlab{}.
\newblock \showarticletitle{Trace Partitioning in Abstract Interpretation Based
  Static Analyzers}. In \bibinfo{booktitle}{\emph{Programming Languages and
  Systems}}, \bibfield{editor}{\bibinfo{person}{Mooly Sagiv}} (Ed.).
  \bibinfo{publisher}{Springer Berlin Heidelberg}, \bibinfo{address}{Berlin,
  Heidelberg}, \bibinfo{pages}{5--20}.
\newblock
\showISBNx{978-3-540-31987-0}


\bibitem[\protect\citeauthoryear{Might and Shivers}{Might and Shivers}{2006}]%
        {MightS06}
\bibfield{author}{\bibinfo{person}{Matthew Might} {and} \bibinfo{person}{Olin
  Shivers}.} \bibinfo{year}{2006}\natexlab{}.
\newblock \showarticletitle{Improving flow analyses via GammaCFA: abstract
  garbage collection and counting}. In \bibinfo{booktitle}{\emph{Proceedings of
  the 11th {ACM} {SIGPLAN} International Conference on Functional Programming,
  {ICFP} 2006, Portland, Oregon, USA, September 16-21, 2006}},
  \bibfield{editor}{\bibinfo{person}{John~H. Reppy} {and}
  \bibinfo{person}{Julia Lawall}} (Eds.). \bibinfo{publisher}{{ACM}},
  \bibinfo{pages}{13--25}.
\newblock
\urldef\tempurl%
\url{https://doi.org/10.1145/1159803.1159807}
\showDOI{\tempurl}


\bibitem[\protect\citeauthoryear{Min{\'e}}{Min{\'e}}{2006}]%
        {mine2006}
\bibfield{author}{\bibinfo{person}{Antoine Min{\'e}}.}
  \bibinfo{year}{2006}\natexlab{}.
\newblock \showarticletitle{The octagon abstract domain}.
\newblock \bibinfo{journal}{\emph{Higher-Order and Symbolic Computation}}
  \bibinfo{volume}{19}, \bibinfo{number}{1} (\bibinfo{date}{01 Mar}
  \bibinfo{year}{2006}), \bibinfo{pages}{31--100}.
\newblock
\showISSN{1573-0557}
\urldef\tempurl%
\url{https://doi.org/10.1007/s10990-006-8609-1}
\showDOI{\tempurl}


\bibitem[\protect\citeauthoryear{Min{\'{e}}}{Min{\'{e}}}{2012}]%
        {mine2012}
\bibfield{author}{\bibinfo{person}{Antoine Min{\'{e}}}.}
  \bibinfo{year}{2012}\natexlab{}.
\newblock \showarticletitle{Static Analysis of Run-Time Errors in Embedded
  Real-Time Parallel {C} Programs}.
\newblock \bibinfo{journal}{\emph{Log. Methods Comput. Sci.}}
  \bibinfo{volume}{8}, \bibinfo{number}{1} (\bibinfo{year}{2012}).
\newblock
\urldef\tempurl%
\url{https://doi.org/10.2168/LMCS-8(1:26)2012}
\showDOI{\tempurl}


\bibitem[\protect\citeauthoryear{Min{\'{e}}}{Min{\'{e}}}{2014}]%
        {mine2014}
\bibfield{author}{\bibinfo{person}{Antoine Min{\'{e}}}.}
  \bibinfo{year}{2014}\natexlab{}.
\newblock \showarticletitle{Relational Thread-Modular Static Value Analysis by
  Abstract Interpretation}. In \bibinfo{booktitle}{\emph{Verification, Model
  Checking, and Abstract Interpretation}},
  \bibfield{editor}{\bibinfo{person}{Kenneth~L. McMillan} {and}
  \bibinfo{person}{Xavier Rival}} (Eds.). \bibinfo{publisher}{Springer Berlin
  Heidelberg}, \bibinfo{address}{Berlin, Heidelberg}, \bibinfo{pages}{39--58}.
\newblock
\showISBNx{978-3-642-54013-4}


\bibitem[\protect\citeauthoryear{Min{\'{e}}}{Min{\'{e}}}{2017}]%
        {Mine2017}
\bibfield{author}{\bibinfo{person}{Antoine Min{\'{e}}}.}
  \bibinfo{year}{2017}\natexlab{}.
\newblock \showarticletitle{Static Analysis of Embedded Real-Time Concurrent
  Software with Dynamic Priorities}.
\newblock \bibinfo{journal}{\emph{Electronic Notes in Theoretical Computer
  Science}}  \bibinfo{volume}{331} (\bibinfo{date}{03} \bibinfo{year}{2017}),
  \bibinfo{pages}{3--39}.
\newblock
\urldef\tempurl%
\url{https://doi.org/10.1016/j.entcs.2017.02.002}
\showDOI{\tempurl}


\bibitem[\protect\citeauthoryear{Monat}{Monat}{2021}]%
        {monat2021}
\bibfield{author}{\bibinfo{person}{Rapha{\"{e}}l Monat}.}
  \bibinfo{year}{2021}\natexlab{}.
\newblock \emph{\bibinfo{title}{Static type and value analysis by abstract
  interpretation of Python programs with native {C} libraries. (Analyse
  statique, de type et de valeur, par interpr{\'{e}}tation abstraite, de
  programmes Python utilisant des librairies {C)}}}.
\newblock \bibinfo{thesistype}{Ph.D. Dissertation}. \bibinfo{school}{Sorbonne
  University, Paris, France}.
\newblock
\urldef\tempurl%
\url{https://tel.archives-ouvertes.fr/tel-03533030}
\showURL{%
\tempurl}


\bibitem[\protect\citeauthoryear{Monat, Ouadjaout, and Min{\'{e}}}{Monat
  et~al\mbox{.}}{2020}]%
        {MonatOM20}
\bibfield{author}{\bibinfo{person}{Rapha{\"{e}}l Monat},
  \bibinfo{person}{Abdelraouf Ouadjaout}, {and} \bibinfo{person}{Antoine
  Min{\'{e}}}.} \bibinfo{year}{2020}\natexlab{}.
\newblock \showarticletitle{Value and allocation sensitivity in static Python
  analyses}. In \bibinfo{booktitle}{\emph{Proceedings of the 9th {ACM}
  {SIGPLAN} International Workshop on the State Of the Art in Program Analysis,
  SOAP@PLDI 2020, London, UK, June 15, 2020}},
  \bibfield{editor}{\bibinfo{person}{Paddy Krishnan} {and}
  \bibinfo{person}{Christoph Reichenbach}} (Eds.). \bibinfo{publisher}{{ACM}},
  \bibinfo{pages}{8--13}.
\newblock
\urldef\tempurl%
\url{https://doi.org/10.1145/3394451.3397205}
\showDOI{\tempurl}


\bibitem[\protect\citeauthoryear{Nicolay, Sti{\'{e}}venart, {de Meuter}, and
  {de Roover}}{Nicolay et~al\mbox{.}}{2019}]%
        {nicolay2019}
\bibfield{author}{\bibinfo{person}{Jens Nicolay}, \bibinfo{person}{Quentin
  Sti{\'{e}}venart}, \bibinfo{person}{Wolfgang {de Meuter}}, {and}
  \bibinfo{person}{Coen {de Roover}}.} \bibinfo{year}{2019}\natexlab{}.
\newblock \showarticletitle{Effect-Driven Flow Analysis}. In
  \bibinfo{booktitle}{\emph{Verification, Model Checking, and Abstract
  Interpretation - 20th International Conference, {VMCAI} 2019, Cascais,
  Portugal, January 13-15, 2019, Proceedings}} \emph{(\bibinfo{series}{Lecture
  Notes in Computer Science})}, \bibfield{editor}{\bibinfo{person}{Constantin
  Enea} {and} \bibinfo{person}{Ruzica Piskac}} (Eds.),
  Vol.~\bibinfo{volume}{11388}. \bibinfo{publisher}{Springer},
  \bibinfo{pages}{247--274}.
\newblock
\urldef\tempurl%
\url{https://doi.org/10.1007/978-3-030-11245-5\_12}
\showDOI{\tempurl}


\bibitem[\protect\citeauthoryear{Park, Lee, and Ryu}{Park
  et~al\mbox{.}}{2022}]%
        {ParkLR22}
\bibfield{author}{\bibinfo{person}{Jihyeok Park}, \bibinfo{person}{Hongki Lee},
  {and} \bibinfo{person}{Sukyoung Ryu}.} \bibinfo{year}{2022}\natexlab{}.
\newblock \showarticletitle{A Survey of Parametric Static Analysis}.
\newblock \bibinfo{journal}{\emph{{ACM} Comput. Surv.}} \bibinfo{volume}{54},
  \bibinfo{number}{7} (\bibinfo{year}{2022}), \bibinfo{pages}{149:1--149:37}.
\newblock
\urldef\tempurl%
\url{https://doi.org/10.1145/3464457}
\showDOI{\tempurl}


\bibitem[\protect\citeauthoryear{Reps, Horwitz, and Sagiv}{Reps
  et~al\mbox{.}}{1995}]%
        {reps1995}
\bibfield{author}{\bibinfo{person}{Thomas Reps}, \bibinfo{person}{Susan
  Horwitz}, {and} \bibinfo{person}{Mooly Sagiv}.}
  \bibinfo{year}{1995}\natexlab{}.
\newblock \showarticletitle{Precise interprocedural dataflow analysis via graph
  reachability}. In \bibinfo{booktitle}{\emph{Proceedings of the 22nd ACM
  SIGPLAN-SIGACT Symposium on Principles of Programming Languages}}
  \emph{(\bibinfo{series}{POPL '95})}. \bibinfo{publisher}{Association for
  Computing Machinery}, \bibinfo{address}{New York, NY, USA},
  \bibinfo{pages}{49--61}.
\newblock
\showISBNx{0897916921}
\urldef\tempurl%
\url{https://doi.org/10.1145/199448.199462}
\showDOI{\tempurl}


\bibitem[\protect\citeauthoryear{Rinetzky, Ramalingam, Sagiv, and
  Yahav}{Rinetzky et~al\mbox{.}}{2008}]%
        {RinetzkyRSY08}
\bibfield{author}{\bibinfo{person}{Noam Rinetzky}, \bibinfo{person}{G.
  Ramalingam}, \bibinfo{person}{Shmuel Sagiv}, {and} \bibinfo{person}{Eran
  Yahav}.} \bibinfo{year}{2008}\natexlab{}.
\newblock \showarticletitle{On the complexity of partially-flow-sensitive alias
  analysis}.
\newblock \bibinfo{journal}{\emph{{ACM} Trans. Program. Lang. Syst.}}
  \bibinfo{volume}{30}, \bibinfo{number}{3} (\bibinfo{year}{2008}),
  \bibinfo{pages}{13:1--13:28}.
\newblock
\urldef\tempurl%
\url{https://doi.org/10.1145/1353445.1353447}
\showDOI{\tempurl}


\bibitem[\protect\citeauthoryear{Rival and Mauborgne}{Rival and
  Mauborgne}{2007}]%
        {Rival07}
\bibfield{author}{\bibinfo{person}{Xavier Rival} {and} \bibinfo{person}{Laurent
  Mauborgne}.} \bibinfo{year}{2007}\natexlab{}.
\newblock \showarticletitle{The Trace Partitioning Abstract Domain}.
\newblock \bibinfo{journal}{\emph{ACM Trans. Program. Lang. Syst.}}
  \bibinfo{volume}{29}, \bibinfo{number}{5} (\bibinfo{date}{Aug.}
  \bibinfo{year}{2007}), \bibinfo{pages}{26--es}.
\newblock
\showISSN{0164-0925}
\urldef\tempurl%
\url{https://doi.org/10.1145/1275497.1275501}
\showDOI{\tempurl}


\bibitem[\protect\citeauthoryear{Roth, Helm, Reif, and Mezini}{Roth
  et~al\mbox{.}}{2021}]%
        {roth2021}
\bibfield{author}{\bibinfo{person}{Tobias Roth}, \bibinfo{person}{Dominik
  Helm}, \bibinfo{person}{Michael Reif}, {and} \bibinfo{person}{Mira Mezini}.}
  \bibinfo{year}{2021}\natexlab{}.
\newblock \showarticletitle{Cifi: Versatile analysis of class and field
  immutability}. In \bibinfo{booktitle}{\emph{36th IEEE/ACM International
  Conference on Automated Software Engineering, ASE 2021}}.
  \bibinfo{publisher}{IEEE}, \bibinfo{pages}{979--990}.
\newblock
\urldef\tempurl%
\url{https://doi.org/10.1109/ASE51524.2021.9678903}
\showDOI{\tempurl}


\bibitem[\protect\citeauthoryear{Roy and Srikant}{Roy and Srikant}{2007}]%
        {RoyS07}
\bibfield{author}{\bibinfo{person}{Subhajit Roy} {and} \bibinfo{person}{Y.~N.
  Srikant}.} \bibinfo{year}{2007}\natexlab{}.
\newblock \showarticletitle{Partial Flow Sensitivity}. In
  \bibinfo{booktitle}{\emph{High Performance Computing - HiPC 2007, 14th
  International Conference, Goa, India, December 18-21, 2007, Proceedings}}
  \emph{(\bibinfo{series}{Lecture Notes in Computer Science})},
  \bibfield{editor}{\bibinfo{person}{Srinivas Aluru}, \bibinfo{person}{Manish
  Parashar}, \bibinfo{person}{Ramamurthy Badrinath}, {and}
  \bibinfo{person}{Viktor~K. Prasanna}} (Eds.), Vol.~\bibinfo{volume}{4873}.
  \bibinfo{publisher}{Springer}, \bibinfo{pages}{245--256}.
\newblock
\urldef\tempurl%
\url{https://doi.org/10.1007/978-3-540-77220-0_25}
\showDOI{\tempurl}


\bibitem[\protect\citeauthoryear{Saan, Erhard, Schwarz, Bozhilov, Holter,
  Tilscher, Vojdani, and Seidl}{Saan et~al\mbox{.}}{2024}]%
        {SaanESBHTVS24a}
\bibfield{author}{\bibinfo{person}{Simmo Saan}, \bibinfo{person}{Julian
  Erhard}, \bibinfo{person}{Michael Schwarz}, \bibinfo{person}{Stanimir
  Bozhilov}, \bibinfo{person}{Karoliine Holter}, \bibinfo{person}{Sarah
  Tilscher}, \bibinfo{person}{Vesal Vojdani}, {and} \bibinfo{person}{Helmut
  Seidl}.} \bibinfo{year}{2024}\natexlab{}.
\newblock \showarticletitle{Goblint: Abstract Interpretation for Memory Safety
  and Termination - (Competition Contribution)}. In
  \bibinfo{booktitle}{\emph{Tools and Algorithms for the Construction and
  Analysis of Systems - 30th International Conference, {TACAS} 2024, Held as
  Part of the European Joint Conferences on Theory and Practice of Software,
  {ETAPS} 2024, Luxembourg City, Luxembourg, April 6-11, 2024, Proceedings,
  Part {III}}} \emph{(\bibinfo{series}{Lecture Notes in Computer Science})},
  \bibfield{editor}{\bibinfo{person}{Bernd Finkbeiner} {and}
  \bibinfo{person}{Laura Kov{\'{a}}cs}} (Eds.), Vol.~\bibinfo{volume}{14572}.
  \bibinfo{publisher}{Springer}, \bibinfo{pages}{381--386}.
\newblock
\urldef\tempurl%
\url{https://doi.org/10.1007/978-3-031-57256-2\_25}
\showDOI{\tempurl}


\bibitem[\protect\citeauthoryear{Saan, Schwarz, Apinis, Erhard, Seidl, Vogler,
  and Vojdani}{Saan et~al\mbox{.}}{2021}]%
        {SaanSAESVV21}
\bibfield{author}{\bibinfo{person}{Simmo Saan}, \bibinfo{person}{Michael
  Schwarz}, \bibinfo{person}{Kalmer Apinis}, \bibinfo{person}{Julian Erhard},
  \bibinfo{person}{Helmut Seidl}, \bibinfo{person}{Ralf Vogler}, {and}
  \bibinfo{person}{Vesal Vojdani}.} \bibinfo{year}{2021}\natexlab{}.
\newblock \showarticletitle{Goblint: Thread-Modular Abstract Interpretation
  Using Side-Effecting Constraints - (Competition Contribution)}. In
  \bibinfo{booktitle}{\emph{Tools and Algorithms for the Construction and
  Analysis of Systems - 27th International Conference, {TACAS} 2021, Held as
  Part of the European Joint Conferences on Theory and Practice of Software,
  {ETAPS} 2021, Luxembourg City, Luxembourg, March 27 - April 1, 2021,
  Proceedings, Part {II}}} \emph{(\bibinfo{series}{Lecture Notes in Computer
  Science})}, \bibfield{editor}{\bibinfo{person}{Jan~Friso Groote} {and}
  \bibinfo{person}{Kim~Guldstrand Larsen}} (Eds.),
  Vol.~\bibinfo{volume}{12652}. \bibinfo{publisher}{Springer},
  \bibinfo{pages}{438--442}.
\newblock
\urldef\tempurl%
\url{https://doi.org/10.1007/978-3-030-72013-1\_28}
\showDOI{\tempurl}


\bibitem[\protect\citeauthoryear{Schwarz and Erhard}{Schwarz and
  Erhard}{2024}]%
        {SchwarzE24}
\bibfield{author}{\bibinfo{person}{Michael Schwarz} {and}
  \bibinfo{person}{Julian Erhard}.} \bibinfo{year}{2024}\natexlab{}.
\newblock \showarticletitle{The digest framework: concurrency-sensitivity for
  abstract interpretation}.
\newblock \bibinfo{journal}{\emph{Int. J. Softw. Tools Technol. Transf.}}
  \bibinfo{volume}{26}, \bibinfo{number}{6} (\bibinfo{year}{2024}),
  \bibinfo{pages}{727--746}.
\newblock
\urldef\tempurl%
\url{https://doi.org/10.1007/S10009-024-00773-Y}
\showDOI{\tempurl}


\bibitem[\protect\citeauthoryear{Schwarz, Erhard, Vojdani, Saan, and
  Seidl}{Schwarz et~al\mbox{.}}{2023a}]%
        {schwarz2023}
\bibfield{author}{\bibinfo{person}{Michael Schwarz}, \bibinfo{person}{Julian
  Erhard}, \bibinfo{person}{Vesal Vojdani}, \bibinfo{person}{Simmo Saan}, {and}
  \bibinfo{person}{Helmut Seidl}.} \bibinfo{year}{2023}\natexlab{a}.
\newblock \showarticletitle{When Long Jumps Fall Short: Control-Flow Tracking
  and Misuse Detection for Non-local Jumps in C}. In
  \bibinfo{booktitle}{\emph{Proceedings of the 12th ACM SIGPLAN International
  Workshop on the State Of the Art in Program Analysis}}
  \emph{(\bibinfo{series}{SOAP 2023})}. \bibinfo{publisher}{Association for
  Computing Machinery}, \bibinfo{address}{New York, NY, USA},
  \bibinfo{pages}{20--26}.
\newblock
\showISBNx{9798400701702}
\urldef\tempurl%
\url{https://doi.org/10.1145/3589250.3596140}
\showDOI{\tempurl}


\bibitem[\protect\citeauthoryear{Schwarz, Saan, Seidl, Apinis, Erhard, and
  Vojdani}{Schwarz et~al\mbox{.}}{2021}]%
        {schwarz2021}
\bibfield{author}{\bibinfo{person}{Michael Schwarz}, \bibinfo{person}{Simmo
  Saan}, \bibinfo{person}{Helmut Seidl}, \bibinfo{person}{Kalmer Apinis},
  \bibinfo{person}{Julian Erhard}, {and} \bibinfo{person}{Vesal Vojdani}.}
  \bibinfo{year}{2021}\natexlab{}.
\newblock \showarticletitle{Improving Thread-Modular Abstract Interpretation}.
  In \bibinfo{booktitle}{\emph{Static Analysis}},
  \bibfield{editor}{\bibinfo{person}{Cezara Dr{\u{a}}goi},
  \bibinfo{person}{Suvam Mukherjee}, {and} \bibinfo{person}{Kedar Namjoshi}}
  (Eds.). \bibinfo{publisher}{Springer International Publishing},
  \bibinfo{address}{Cham}, \bibinfo{pages}{359--383}.
\newblock
\showISBNx{978-3-030-88806-0}


\bibitem[\protect\citeauthoryear{Schwarz, Saan, Seidl, Erhard, and
  Vojdani}{Schwarz et~al\mbox{.}}{2023b}]%
        {schwarz2023b}
\bibfield{author}{\bibinfo{person}{Michael Schwarz}, \bibinfo{person}{Simmo
  Saan}, \bibinfo{person}{Helmut Seidl}, \bibinfo{person}{Julian Erhard}, {and}
  \bibinfo{person}{Vesal Vojdani}.} \bibinfo{year}{2023}\natexlab{b}.
\newblock \showarticletitle{Clustered Relational Thread-Modular Abstract
  Interpretation with Local Traces}. In \bibinfo{booktitle}{\emph{Programming
  Languages and Systems}}, \bibfield{editor}{\bibinfo{person}{Thomas Wies}}
  (Ed.). \bibinfo{publisher}{Springer Nature Switzerland},
  \bibinfo{address}{Cham}, \bibinfo{pages}{28--58}.
\newblock
\showISBNx{978-3-031-30044-8}


\bibitem[\protect\citeauthoryear{Seidl, Vene, and M{\"u}ller-Olm}{Seidl
  et~al\mbox{.}}{2003}]%
        {vene2003}
\bibfield{author}{\bibinfo{person}{Helmut Seidl}, \bibinfo{person}{Varmo Vene},
  {and} \bibinfo{person}{Markus M{\"u}ller-Olm}.}
  \bibinfo{year}{2003}\natexlab{}.
\newblock \showarticletitle{Global Invariants for Analysing Multi-Threaded
  Applications}. In \bibinfo{booktitle}{\emph{Proceedings-Estonian Academy Of
  Sciences Physics Mathematics}}, Vol.~\bibinfo{volume}{52}.
  \bibinfo{publisher}{Estonian Academy Publishers}, \bibinfo{pages}{413--436}.
\newblock


\bibitem[\protect\citeauthoryear{Seidl and Vogler}{Seidl and Vogler}{2021}]%
        {seidl2021}
\bibfield{author}{\bibinfo{person}{Helmut Seidl} {and} \bibinfo{person}{Ralf
  Vogler}.} \bibinfo{year}{2021}\natexlab{}.
\newblock \showarticletitle{Three improvements to the top-down solver}.
\newblock \bibinfo{journal}{\emph{Mathematical Structures in Computer Science}}
  \bibinfo{volume}{31}, \bibinfo{number}{9} (\bibinfo{year}{2021}),
  \bibinfo{pages}{1090--1134}.
\newblock
\urldef\tempurl%
\url{https://doi.org/10.1017/S0960129521000499}
\showDOI{\tempurl}


\bibitem[\protect\citeauthoryear{Shivers}{Shivers}{1991}]%
        {shivers1991}
\bibfield{author}{\bibinfo{person}{Olin~Grigsby Shivers}.}
  \bibinfo{year}{1991}\natexlab{}.
\newblock \emph{\bibinfo{title}{Control-flow analysis of higher-order languages
  or taming lambda}}.
\newblock \bibinfo{thesistype}{Ph.D. Dissertation}. \bibinfo{school}{Carnegie
  Mellon University}.
\newblock


\bibitem[\protect\citeauthoryear{Steensgaard}{Steensgaard}{1996}]%
        {steensgaard1996}
\bibfield{author}{\bibinfo{person}{Bjarne Steensgaard}.}
  \bibinfo{year}{1996}\natexlab{}.
\newblock \showarticletitle{Points-to analysis in almost linear time}. In
  \bibinfo{booktitle}{\emph{Proceedings of the 23rd ACM SIGPLAN-SIGACT
  Symposium on Principles of Programming Languages}}
  \emph{(\bibinfo{series}{POPL '96})}. \bibinfo{publisher}{Association for
  Computing Machinery}, \bibinfo{address}{New York, NY, USA},
  \bibinfo{pages}{32--41}.
\newblock
\showISBNx{0897917693}
\urldef\tempurl%
\url{https://doi.org/10.1145/237721.237727}
\showDOI{\tempurl}


\bibitem[\protect\citeauthoryear{Stemmler, Schwarz, Erhard, Tilscher, and
  Seidl}{Stemmler et~al\mbox{.}}{2025}]%
        {artifact}
\bibfield{author}{\bibinfo{person}{Fabian Stemmler}, \bibinfo{person}{Michael
  Schwarz}, \bibinfo{person}{Julian Erhard}, \bibinfo{person}{Sarah Tilscher},
  {and} \bibinfo{person}{Helmut Seidl}.} \bibinfo{year}{2025}\natexlab{}.
\newblock \bibinfo{title}{Artifact for 'Taking out the Toxic Trash: Recovering
  Precision in Mixed Flow-Sensitive Static Analyses'}.
\newblock   (\bibinfo{date}{April} \bibinfo{year}{2025}).
\newblock
\urldef\tempurl%
\url{https://doi.org/10.5281/zenodo.15047000}
\showDOI{\tempurl}


\bibitem[\protect\citeauthoryear{Sti{\'{e}}venart, Nicolay, {de Meuter}, and
  {de Roover}}{Sti{\'{e}}venart et~al\mbox{.}}{2019}]%
        {StievenartNMR19}
\bibfield{author}{\bibinfo{person}{Quentin Sti{\'{e}}venart},
  \bibinfo{person}{Jens Nicolay}, \bibinfo{person}{Wolfgang {de Meuter}}, {and}
  \bibinfo{person}{Coen {de Roover}}.} \bibinfo{year}{2019}\natexlab{}.
\newblock \showarticletitle{A general method for rendering static analyses for
  diverse concurrency models modular}.
\newblock \bibinfo{journal}{\emph{J. Syst. Softw.}}  \bibinfo{volume}{147}
  (\bibinfo{year}{2019}), \bibinfo{pages}{17--45}.
\newblock
\urldef\tempurl%
\url{https://doi.org/10.1016/J.JSS.2018.10.001}
\showDOI{\tempurl}


\bibitem[\protect\citeauthoryear{Suzanne and Min{\'{e}}}{Suzanne and
  Min{\'{e}}}{2018}]%
        {Mine2018}
\bibfield{author}{\bibinfo{person}{Thibault Suzanne} {and}
  \bibinfo{person}{Antoine Min{\'{e}}}.} \bibinfo{year}{2018}\natexlab{}.
\newblock \showarticletitle{Relational Thread-Modular Abstract Interpretation
  Under Relaxed Memory Models}. In \bibinfo{booktitle}{\emph{APLAS '18}},
  Vol.~\bibinfo{volume}{LNCS 11275}. \bibinfo{publisher}{Springer},
  \bibinfo{pages}{109--128}.
\newblock
\urldef\tempurl%
\url{https://doi.org/10.1007/978-3-030-02768-1_6}
\showDOI{\tempurl}


\bibitem[\protect\citeauthoryear{{van Es}, Sti{\'{e}}venart, and {de
  Roover}}{{van Es} et~al\mbox{.}}{2019}]%
        {vanEs2019}
\bibfield{author}{\bibinfo{person}{Noah {van Es}}, \bibinfo{person}{Quentin
  Sti{\'{e}}venart}, {and} \bibinfo{person}{Coen {de Roover}}.}
  \bibinfo{year}{2019}\natexlab{}.
\newblock \showarticletitle{Garbage-Free Abstract Interpretation Through
  Abstract Reference Counting}. In \bibinfo{booktitle}{\emph{33rd European
  Conference on Object-Oriented Programming, {ECOOP} 2019, July 15-19, 2019,
  London, United Kingdom}} \emph{(\bibinfo{series}{LIPIcs})},
  \bibfield{editor}{\bibinfo{person}{Alastair~F. Donaldson}} (Ed.),
  Vol.~\bibinfo{volume}{134}. \bibinfo{publisher}{Schloss Dagstuhl -
  Leibniz-Zentrum f{\"{u}}r Informatik}, \bibinfo{pages}{10:1--10:33}.
\newblock
\urldef\tempurl%
\url{https://doi.org/10.4230/LIPICS.ECOOP.2019.10}
\showDOI{\tempurl}


\bibitem[\protect\citeauthoryear{{van Es}, {van der Plas}, Sti{\'{e}}venart,
  and {de Roover}}{{van Es} et~al\mbox{.}}{2020}]%
        {EsPSR20}
\bibfield{author}{\bibinfo{person}{Noah {van Es}}, \bibinfo{person}{Jens {van
  der Plas}}, \bibinfo{person}{Quentin Sti{\'{e}}venart}, {and}
  \bibinfo{person}{Coen {de Roover}}.} \bibinfo{year}{2020}\natexlab{}.
\newblock \showarticletitle{{MAF:} {A} Framework for Modular Static Analysis of
  Higher-Order Languages}. In \bibinfo{booktitle}{\emph{20th {IEEE}
  International Working Conference on Source Code Analysis and Manipulation,
  {SCAM} 2020, Adelaide, Australia, September 28 - October 2, 2020}}.
  \bibinfo{publisher}{{IEEE}}, \bibinfo{pages}{37--42}.
\newblock
\urldef\tempurl%
\url{https://doi.org/10.1109/SCAM51674.2020.00009}
\showDOI{\tempurl}


\bibitem[\protect\citeauthoryear{{Van Horn} and Might}{{Van Horn} and
  Might}{2010}]%
        {HornM10}
\bibfield{author}{\bibinfo{person}{David {Van Horn}} {and}
  \bibinfo{person}{Matthew Might}.} \bibinfo{year}{2010}\natexlab{}.
\newblock \showarticletitle{Abstracting abstract machines}. In
  \bibinfo{booktitle}{\emph{Proceeding of the 15th {ACM} {SIGPLAN}
  international conference on Functional programming, {ICFP} 2010, Baltimore,
  Maryland, USA, September 27-29, 2010}},
  \bibfield{editor}{\bibinfo{person}{Paul Hudak} {and}
  \bibinfo{person}{Stephanie Weirich}} (Eds.). \bibinfo{publisher}{{ACM}},
  \bibinfo{pages}{51--62}.
\newblock
\urldef\tempurl%
\url{https://doi.org/10.1145/1863543.1863553}
\showDOI{\tempurl}


\bibitem[\protect\citeauthoryear{Vojdani, Apinis, R\~{o}tov, Seidl, Vene, and
  Vogler}{Vojdani et~al\mbox{.}}{2016}]%
        {vojdani2016}
\bibfield{author}{\bibinfo{person}{Vesal Vojdani}, \bibinfo{person}{Kalmer
  Apinis}, \bibinfo{person}{Vootele R\~{o}tov}, \bibinfo{person}{Helmut Seidl},
  \bibinfo{person}{Varmo Vene}, {and} \bibinfo{person}{Ralf Vogler}.}
  \bibinfo{year}{2016}\natexlab{}.
\newblock \showarticletitle{Static race detection for device drivers: the
  {\textsc{Goblint}} approach}. In \bibinfo{booktitle}{\emph{Proceedings of the
  31st IEEE/ACM International Conference on Automated Software Engineering}}
  \emph{(\bibinfo{series}{ASE '16})}. \bibinfo{publisher}{Association for
  Computing Machinery}, \bibinfo{address}{New York, NY, USA},
  \bibinfo{pages}{391--402}.
\newblock
\showISBNx{9781450338455}
\urldef\tempurl%
\url{https://doi.org/10.1145/2970276.2970337}
\showDOI{\tempurl}


\bibitem[\protect\citeauthoryear{Zakharov, Mandrykin, Mutilin, Novikov,
  Petrenko, and Khoroshilov}{Zakharov et~al\mbox{.}}{2015}]%
        {DBLP:journals/pcs/ZakharovMMNPK15}
\bibfield{author}{\bibinfo{person}{Ilja~S. Zakharov},
  \bibinfo{person}{Mikhail~U. Mandrykin}, \bibinfo{person}{Vadim~S. Mutilin},
  \bibinfo{person}{Evgeny Novikov}, \bibinfo{person}{Alexander~K. Petrenko},
  {and} \bibinfo{person}{Alexey~V. Khoroshilov}.}
  \bibinfo{year}{2015}\natexlab{}.
\newblock \showarticletitle{Configurable toolset for static verification of
  operating systems kernel modules}.
\newblock \bibinfo{journal}{\emph{Program. Comput. Softw.}}
  \bibinfo{volume}{41}, \bibinfo{number}{1} (\bibinfo{year}{2015}),
  \bibinfo{pages}{49--64}.
\newblock
\urldef\tempurl%
\url{https://doi.org/10.1134/S0361768815010065}
\showDOI{\tempurl}


\end{thebibliography}

\appendix
\clearpage

\section{Constraint Systems for the Factorial Program}\label{app:fac}

Consider the factorial program from \cref{lst:fact} now equipped with a numbering of its program points
which we use when constructing the constraint system for interprocedural analysis:
%\vspace{-0.8\baselineskip}

\setlength{\abovecaptionskip}{0em}
\setlength{\belowcaptionskip}{0.2\baselineskip}
\begin{center}
\begin{lstlisting}[style=ccode, xleftmargin=5.0ex,caption={\label{lst:fact-1}Factorial program.},numbers=none]
int t = 1;
void fac(int i) {
/* 0 */  if (i > 0) {
/* 1 */     fac(i-1);
/* 2 */     t = i * t;
  }
/* 3 */  else assert(i == 0);
/* 4 */
}
void main() {
/* 5 */  int i = 17;
/* 6 */  fac(i);
/* 7 */
}
\end{lstlisting}
\end{center}

For the analysis, we use global store widening selectively for variable \lstinline{t}, i.e., track the abstract value for
\lstinline{t} flow-insensitively.
Since there is also just the single local program variable \lstinline{i},
it suffices to maintain single intervals for both local and global unknowns.
The constraint system for context-sensitive analysis of the program as outlined in \cref{sec:background}
is given by:
\[
\begin{array}{llll}
(A): \qquad\qquad\qquad & (\sigma\,\_\text{main},\rho)	&\sqsupseteq&	(\sigma\,(7,\top),\{(\textsf{t},[1,1]), ((5,\top),\top)\})\\

(B): & (\sigma\,(6,\top),\rho)	&\sqsupseteq&	\textbf{if}\;\rho\,(5,\top)=\bot\;\textbf{then}\;(\bot,\emptyset)\\
& 			&&		\textbf{else}\;([17,17],\emptyset)		\\
(C): & (\sigma\,(7,\top),\rho)	&\sqsupseteq&	\textbf{if}\;\sigma\,(6,\top)=\bot\;\textbf{then}\;(\bot,\emptyset)\\
&			&&		\textbf{else}\,\begin{array}[t]{l}
					\textbf{let}\;d'=\sigma\,(6,\top)\;\textbf{in} \\
					\textbf{let}\;d'' =  \begin{array}[t]{l}
						\textbf{if}\;\sigma\,(4,d') =\bot\;\textbf{then}\;\bot\\
						\textbf{else}\;\sigma\,(6,\top)\;\textbf{in}
						\end{array}\\
					(d'', \{((0,d'),d')\})
					\end{array}\\[4ex]
& && \\
(D1): & (\sigma\,(1,d),\rho)	&\sqsupseteq&	(\rho\,(0,d)\sqcap[1,\infty], \emptyset)	\\
(E): &(\sigma\,(2,d),\rho)	&\sqsupseteq&	\textbf{if}\;\sigma\,(1,d) = \bot\;\textbf{then}\;(\bot,\emptyset)\\
 &			&&		\textbf{else}\,\begin{array}[t]{l}
					\textbf{let}\;d'=\sigma\,(1,d)\,{-^\sharp}\,[1,1]\;\textbf{in} 	\\
					\textbf{let}\;d'' = \begin{array}[t]{l}
						\textbf{if}\;(\sigma\,(4,d') =\bot)\;\textbf{then}\;\bot\\
						\textbf{else}\;\sigma\,(1,\top)\;\textbf{in}
						\end{array}\\
					(d'', \{((0,d'),d')\})
					\end{array}	\\
(D2): & (\sigma\,(3,d),\rho)	&\sqsupseteq&	(\rho\,(0,d)\sqcap[-\infty,0],\emptyset)	\\
(F): & (\sigma\,(4,d),\rho)	&\sqsupseteq&	(\sigma\,(2,d),\{(t,\sigma\,(2,d) \,{*^\sharp}\,\rho\,t)\})\,
					\sqcup\,	(\sigma\,(3,d),\emptyset)
\end{array}
\]
The letters on the left-hand side do not form part of the constraints system and are used to label the constraints, so they can be referenced in the following explanations.
For convenience, we have numbered the program points consecutively
where $0$ and $5$ represent the start points $\textsf{st}_{\text{fac}}$ and $\textsf{st}_{\text{main}}$ of the
procedures \lstinline{fac} and \lstinline{main}, respectively, while
$4$ and $7$ represent their return points.
$\top = [-\infty,\infty]$ and $\bot=\emptyset$ represent the top and bottom elements of the interval domain,
while $d$ ranges over arbitrary non-empty intervals, and $-^\sharp, *^\sharp$ are the abstract operators for
subtraction and multiplication of intervals, respectively.
Thus, the set of globals consists of \lstinline{t} together with $(0,\top)$ and $(5,d)$ ($d$ an interval)
for the start points of procedures \lstinline{main} and \lstinline{fac}, respectively.
%
% Note that the constraints for procedure calls must be carefully crafted to enforce strictness in
% the unknown corresponding to the respective control-flow predecessors.

In detail, the unknowns $(5,\top),(6,\top),(7,\top)$ correspond to the program points of procedure
\lstinline{main} in context $\top$, where
\begin{itemize}
\item[(A)]	The constraint for $\_\text{main}$ asks for the return value of the procedure \lstinline{main}
	for calling context $\top$ while contributing the side-effect $\top$ to its start point
	for the same calling context; additionally, the initial abstract value $[1,1]$ is contributed to
	the unknown for \lstinline{t};
\item[(B)]	The constraint for $(6,\top)$ checks whether its control-flow predecessor $(5,\top)$ is reachable.
	If so, the abstract value (of \lstinline{i}) is set to $[17,17]$;
\item[(C)]	The constraint for $(7,\top)$ models the abstract effect of the call $\text{fac}(i)$.
	If the control-flow predecessor $(6,\top)$ is reachable,
	its local state $d'$ (i.e., the value of \lstinline{i} before the call) is contributed to
	$(0,d')$, i.e., the start point of \lstinline{fac} in context $d'$,
	where the local state after the call either is $\bot$ (if
	the call returns $\bot$) or equal to the abstract state before the call.
\end{itemize}
For any calling context $d$, the unknowns $(0,d),\ldots,(4,d)$ correspond to the program points of procedure
\lstinline{fac} in context $d$. Note that due to the recursive calls, multiple such contexts might be
encountered during the analysis. In detail,
\begin{itemize}
\item[(D1/2)]	The constraints for $(1,d)$ and $(3,d)$ correspond to positive and negative guards
	checking whether the parameter \lstinline{i} exceeds 0 or not;
\item[(E)]	The constraint for $(2,d)$ corresponds to the other call of the procedure \lstinline{fac}
	where now for the actual parameter is \lstinline{i - 1}. Accordingly, the
	abstract value $d'$ of $\sigma(1,d)\,{-^\sharp}\,[1,1]$ for the control-flow predecessor $(1,d)$ is
	contributed to the unknown $(0,d')$, i.e., the start point of \lstinline{fac} in calling context $d'$.
	Again, the local state after the call either is $\bot$ (if
	the call returns $\bot$) or equal to the abstract state before the call;
\item[(F)]	The right-hand side of the constraint for $(4,d)$ is the join of the effects of two
	control-flow edges.
	The effect for the edge from $(2,d)$ accounts for the update of the program variable \lstinline{t}
	by providing a contribution to \lstinline{t}, the abstract value
	$\sigma(2,d)\,{*^\sharp}\,\rho$\lstinline{t}.
	The effect for the edge from $(3,d)$ just returns the abstract value for $(3,d)$.
\end{itemize}
The constraint system for context-insensitive analysis of the factorial is constructed analogously --
with the only modification is that now the context component of each unknown equals the trivial context $\bullet$:
\[
\begin{array}{lll}
(\sigma\,\_\text{main},\rho)   	&\sqsupseteq&   (\sigma\,(7,\bullet),\{(\textsf{t},[1,1]), ((5,\bullet),\top)\})\\
(\sigma\,(1,\bullet),\rho)    	&\sqsupseteq&   (\rho\,(0,\bullet)\sqcap[1,\infty], \emptyset)      \\
(\sigma\,(2,\bullet),\rho)    	&\sqsupseteq&   \textbf{if}\;\sigma\,(1,\bullet)=\bot\;
						\textbf{then}\;(\bot,\emptyset)\\
				&&		\textbf{else}\,\begin{array}[t]{l}
					\textbf{let}\;d' = \sigma\,(1,\bullet)\,{-^\sharp}\,[1,1]\;\textbf{in}     \\
	                                \textbf{let}\;d'' = \begin{array}[t]{l}
                                                \textbf{if}\;(\sigma\,(4,\bullet) =\bot)\;\textbf{then}\;\bot\\
                                                \textbf{else}\;\sigma\,(1,\bullet)\;\textbf{in}
                                                \end{array}\\
	                                (d'', \{((0,\bullet),d')\})
					\end{array}	\\
(\sigma\,(3,\bullet),\rho)    	&\sqsupseteq&   (\rho\,(0,\bullet)\sqcap[-\infty,0],\emptyset)      \\
(\sigma\,(4,\bullet),\rho)    	&\sqsupseteq&   (\sigma\,(2,\bullet),\{(t,\sigma\,(2,\bullet)\,{*^\sharp}\,\rho\,t)\})
	                        \;\sqcup\;      (\sigma\,(3,\bullet),\emptyset)	\\
(\sigma\,(6,\bullet),\rho)	&\sqsupseteq&	\textbf{if}\;\rho\,(5,\bullet)=\bot\;
						\textbf{then}\;(\bot,\emptyset)\\
				&&		\textbf{else}\;([17,17],\emptyset)		\\
(\sigma\,(7,\bullet),\rho) 	&\sqsupseteq&   \textbf{let}\;d' = \sigma\,(6,\bullet)\;\textbf{in} \\
                        	&&              \textbf{if}\;d'=\bot\;\textbf{then}\;(\bot,\emptyset)\\
			        &&              \textbf{else} \,\begin{array}[t]{l}
					\textbf{let}\;d'' =  \begin{array}[t]{l}
					\textbf{if}\;(\sigma\,(4,\bullet) =\bot)\;\textbf{then}\;\bot\\
					\textbf{else}\;\sigma\,(6,\bullet)\;\textbf{in}
                                        \end{array}\\
	                              (d'', \{((0,\bullet),d')\})
				      \end{array}
\end{array}
\]
The resulting constraint system consists of finitely many constraints only where
the set of globals is given by \textsf{t}, $(0,\bullet)$ and $(5,\bullet)$.

\section{Examples for infinite W/N switches for update rule \ref{lst:update_globals_warrow_imp}}\label{appendix:warrowing-apinis-nonterm}

Consider again \cref{ex:nonterm}.
In this example, the same contribution to the unknowns $a$ and $b$ is triggered
twice in a row.
A program for which the analysis may result in such a behavior is given in \cref{lst:program_warrowing_nonterm}.

\begin{lstlisting}[style=ccode,label=lst:program_warrowing_nonterm,caption={Program that may cause nontermination for update rule \ref{lst:update_globals_warrow_imp}.}]
void thread1(){
  while(true) {
    u = a;
    a = b + 1;
  }
}

void thread2(){
  while(true){
    v = b;
    b = a + 1;
  }
}
\end{lstlisting}

While this causes a narrowing to be applied in update rule
\ref{lst:update_globals_warrow_imp}, the update rule \ref{lst:apinis-simple_imp}
does not perform widening and narrowing unless the contribution has changed.
Thus, we modify the constraint system from \cref{ex:nonterm} as follows:
\[
    \begin{array}{lll}
        (\sigma\,x,\rho) & \sqsupseteq & (\rho\,a, \{
             a \mapsto \textbf{if } (\rho\,a ) = \infty \textbf{ then } (\rho\,b) + 2 \textbf{ else } (\rho\,b) + 1 \})\\
        (\sigma\,y,\rho) & \sqsupseteq & (\rho\,b, \{
             b \mapsto \textbf{if } (\rho\,b ) = \infty \textbf{ then } (\rho\,a) + 2 \textbf{ else } (\rho\,a) + 1 \})
    \end{array}
\]
This modified constraint system remains monotonic, but now when $a$ gets widened to $\infty$, the contribution caused
by $x$ upon re-evaluation is different from the contribution in the last iteration, which causes narrowing to be applied.
Through a further iteration narrowing, the value then gets narrowed back down to $(\rho\,b) + 1$.
The sequence thus is the same one as obtained with update rule \ref{lst:update_globals_warrow_imp} for the unmodified
constraint system from \cref{ex:nonterm} with the difference that one more intermediate value is attained in each round.

\section{Example where reluctant does not ensure finite updates}\label{appendix:reluctant-nonterm}
This section provides an example where the update rule using reluctant widening does not ensure
a finite number of updates, unless the operator is strong.
Consider the domain $\Dom = (\mathbb{N} \cup \infty) \times \{ 0,1,2 \}$ where $\mathbb{N}$ are
the natural numbers. The domain uses a lexicographical ordering, where for each component, the ordering
follows the usual one for the natural numbers.
Since the order on $\Dom$ is total, the join can be defined as the maximum.
We use the following widening operator:
\[
    (a,b) \nabla (c,d) = \begin{cases}
        (\infty, 2) & \text{if }  \textsf{max}(b,d) = 2 \\
        (\textsf{max}(a,c), \textsf{max}(b,d) +1) & \text{otherwise}
    \end{cases}
\]
This operator is a widening, but not a \emph{strong} widening.
The second component of the domain can be thought of as a sort of \emph{widening gas} encoded in the domain.

Now assume we have unknowns $x$ and $y$ that make contributions to a global $g$, and that the reluctant variant of the update rule from
\cref{lst:update_globals_robust_warrowing_imp} is used.
The following table illustrates a pattern of updates to the unknown $g$ that may be infinitely prolonged.
The first column is the overall value of the global $g$.
The fourth and fifth columns indicate new contributions that are obtained when evaluating the right-hand sides of $x$ and $y$, respectively.
The second and third columns are the values stored in \lstinline|cmap| for the contributions by $x$ and $y$ to $g$, respectively.
Empty values in the second and third columns indicate an unchanged value.
\begin{center}
    \begin{tabular}{lll|lll}
        g & \lstinline|cmap[g][x]| & \lstinline|cmap[g][y]| & x contribution & y contribution & combination op
        \\
        \bottomrule
        \addlinespace
        0,0 & 0,0 & 0,0 & & & \\
        1,1 & 1,1 & & 1,0 & & $\nabla$ (x cont. $\sqsupset$ g) \\
        2,1 & & 2,1 & & 2,0 & $\nabla$ (y cont. $\sqsupset$ g) \\
        2,1 & 2,0 & & 2,0 & & $\sqcup$ (x cont. $\sqsubseteq$ g) \\
        3,1 & 3,1 & & 3,0 & & $\nabla$ (x cont. $\sqsupset$ g) \\
        3,1 & & 3,0 & & 3,0 & $\sqcup$ (y cont. $\sqsubseteq$ g) \\
        4,1 & & 4,1 & & 4,0 & $\nabla$ (y cont. $\sqsupset$ g) \\
        4,1 & 4,0 & & 4,0 & & $\sqcup$ (x cont. $\sqsubseteq$ g) \\
        5,1 & 5,1 & & 5,0 & & $\nabla$ (x cont. $\sqsupset$ g) \\
        5,1 & & 5,0 & & 5,0 & $\sqcup$ (y cont. $\sqsubseteq$ g) \\
        \multicolumn{6}{c}{...}
     \end{tabular}
\end{center}
In this example, the ability to join in values into the \lstinline|cmap| is used to reset the widening gas so that a
subsequent widening on the same component of \lstinline|cmap| will not go to the top value given by $(\infty, 2)$.

\section{Abstract Garbage Collection for the solver \tdside{}}\label{appendix:td3agc}
Consider again \cref{ex:wewanthisbeforthestrangetdsidestuff} from \cref{ss:garbage}.
For forward-propagating solvers, plugging in the enhanced update rule yields the desired result out of the box.

The situation is different for the local solver \tdside{}.
That solver avoids eager re-evaluation of unknowns affected by an unknown $x$ changing its value and instead only \emph{marks} all possibly affected unknowns as \emph{unstable}.
Such unstable unknowns will only be re-evaluated later in case they are queried again.
More concretely, consider the constraint corresponding to a call edge $e=(u,p,v)$:
\[
(\sigma\,(v,c'),\rho)	\;\sqsupseteq\;
	\begin{array}[t]{l}
	\textbf{let}\;(c,d) = \textsf{enter}_e^\sharp\,c'\,(\sigma\,(u,c'))\;\textbf{in}	\\
	\textbf{let}\;d' =\textsf{combine}_e^\sharp\,(\sigma\,(u,c'))\,(\sigma\,(\textsf{ret}_p,c))\;\textbf{in}\\
	(d',\{(\textsf{st}_p,c)\mapsto d\})
	\end{array}
\]
where we assume that the function $\textsf{enter}_e^\sharp:\Context\to\Dom\to(\Context\times\Dom)$
implements abstract passing of parameters.
It takes the preceding context $c'$ together with the abstract value before the call and
returns the new context $c$ for the called procedure $p$ together with entry state
for which $p$ is called.
Moreover, the function $\textsf{combine}_e^\sharp:\Dom\to\Dom\to\Dom$ takes the abstract value before the call
together with the abstract value returned for $p$ in context $c$ to construct a value for the endpoint of the edge $e$
in the caller's context $c'$.
If, due to larger values for the unknown $(u,c')$, the contribution to the global $(\textsf{st}_p, c)$ is withdrawn, also the value
$(\textsf{ret}_p, c)$  for the return point $\textsf{ret}_p$ in context $c$ will no longer be queried.
If the unknown $(\textsf{ret}_p, c)$ is no longer queried elsewhere by the solver,
the out-dated values of $(\textsf{ret}_p, c)$ as well as of all other unknowns $(v,c)$ are still
preserved together with their out-dated contributions to globals.

After having found the unknown $(\textsf{st}_p,c)$ to receive the value $\bot$,
we therefore let the solver \tdside{} not only \emph{destabilize} all unknowns possibly influenced by the update
(i.e., mark them for re-evaluation), but explicitly initiate re-evaluation of $(\textsf{ret}_p, c)$ ---
in the formulation of \tdside{} from \cite{erhard2024a} ---
by calling \lstinline{query} $(\textsf{ret}_p, c)$.
Due to the strictness of all constraints $\sem{e,c}^\sharp$ for edges $e$ in the control-flow predecessor
% $(u,c)$ ($u$ start point of the edge $e$),
this re-evaluation will perform the abstract garbage collection
and set the values of $\sigma$ for all unknowns $(v,c), v\in N_p,$ to $\bot$.

\begin{example}
  Consider again \cref{ex:wewanthisbeforthestrangetdsidestuff} and assume that \tdside{} is used.
  As with other solvers, the abstract value for $k$ in line $10$ temporarily becomes $[0, \infty]$ due to widening on locals,
  and the call to the procedure $f$ is analyzed with the context $c$ where $k \mapsto [0, \infty]$.
  Later, narrowing on locals is performed by \tdside{} and the value of $k$ for line $10$ becomes $[0, 9]$. The call to $f$ is now analyzed in another context where $k \mapsto [0, 9]$.
  Using the update rule with abstract garbage collection, the previous non-bottom contribution to $(\textsf{st}_f, c)$ is withdrawn.
  To ensure that $\bot$ is propagated for program points of the procedure $f$ in the context $c$, the unknown  $(\textsf{ret}_f, c)$ needs to be queried.
  This way, the recursive \tdside{} solver will descend into evaluating the procedure $f$ in the context $c$ again.
  As the abstract value of $(\textsf{st}_f, c)$ has been set to $\bot$, the $\bot$ value is propagated to the other unknowns consisting of nodes in $f$ and the context $c$.
  Then, the assertion \lstinline|a == 0| can be shown.
\end{example}

\section{Description of benchmarks for (RQ5)}\label{appendix:benchmarks}

The benchmarks used for answering \ref{rq:bot} are multi-threaded real-world {C} programs from literature.
Here, we provide a description of the benchmark sets and report on those benchmarks where not all approaches terminated successfully.

The first set was used to benchmark multi-threaded value analyses~\cite{schwarz2021,schwarz2023b} in the \textsc{Goblint} static analyzer
and consists of six \textsc{Posix} programs as well seven Linux device drivers which where pre-processed by the LDV toolchain~\cite{DBLP:journals/pcs/ZakharovMMNPK15}.
These tasks are part of the \texttt{c/ldv-linux-3.14-races} subset of the \textsc{ConcurrencySafety-Main} category of the \textsc{SV-COMP} benchmark suite~\cite{Beyer24}.
\Cref{tab:benchmarks} lists the benchmarks, gives a short description, and reports the number of lines of the source file (Lines) as well as the logical lines of code (LLoC)
obtained by only counting lines that contain executable code (thus, excluding not only comments and empty lines, but also struct definitions or typedefs).

\begin{table}
    \centering
    \caption{Benchmarks used to evaluate \textsc{Goblint}~\cite{schwarz2021,schwarz2023b}.}
    \label{tab:benchmarks}
        \footnotesize
        \begin{tabular}{lrrl}%
            \multicolumn{4}{l}{\textsc{Posix} Programs}\\
            \bottomrule
            \addlinespace
            Name  & Lines & LLoC & Description \\
            \addlinespace
            \midrule
            \benchmarkprog{pfscan} & 1295 & 562 & Parallel file scanner \\
            \benchmarkprog{aget} & 1280 & 587 & Multi-threaded HTTP download accelerator \\
            \benchmarkprog{ctrace} & 1407 & 657 & C Tracing library sample program \\
            \benchmarkprog{knot} & 2255 & 981 & Multi-threaded webserver \\
            \benchmarkprog{ypbind} & 6588 & 992 & Linux NIS binding process \\
            \benchmarkprog{smptrc} & 5787 & 3037 & SMTP Open Relay Checker \\
            \addlinespace
            \addlinespace
            \multicolumn{4}{l}{\textsc{Linux Device Drivers} (After processing by LDV toolchain~\cite{DBLP:journals/pcs/ZakharovMMNPK15})}\\
            \bottomrule
            \addlinespace
            Name  & Lines & LLoC & Devices\\
            \addlinespace
            \midrule
            \benchmarkprog{iowarrior} & 7687 & 1345 & IOWarrior chips for I/O via USB from Code Mercenaries \\
            \benchmarkprog{adutux} & 8114 & 1520 & ADU devices for I/O from Ontrak Control Systems \\
            \benchmarkprog{w83977af} & 10071 & 1501 & Winbond W83977AF Super I/O chip for data transmission in noisy environments  \\
            \benchmarkprog{tegra20} & 7111 & 1547 & Nvidia's Tegra20/Tegra30 SLINK Controller (for chip-to-chip communication) \\
            \benchmarkprog{nsc} &12778  & 2379 & NSC PC'108 and PC'338 IrDA chipsets (for infrared communications) \\
            \benchmarkprog{marvell1} & 12246 & 2465 & CMOS camera controller in Marvell 88ALP01 chip \\
            \benchmarkprog{marvell2} & 12256 & 2465 & CMOS camera controller in Marvell 88ALP01 chip \\
            \bottomrule
        \end{tabular}
\end{table}

% Be careful! Some references are to Goblint other are to WizWoz (\Goblint)!!!!
The second, larger, set of benchmarks was assembled by \citet{DBLP:conf/icse/HongR23} to evaluate \textsc{Concrat}, which is a tool aimed at the automatic translation of {C} programs to Rust.
The benchmarks were collected by searching public repositories on GitHub which have at most 500~kB of C code, use the \textsc{pthread} locking facilities, have at least 1000 stars and are not intended for educational purposes.
Additionally, only programs which can be translated by the \textsc{C2Rust} tool were considered.
The set comprises 46 tasks: Out of those, 5 lack a main procedure and had to be excluded as \Goblint\ targets whole program analysis. 5 further programs were excluded as they are statically known to be single-threaded.
This leaves 36 benchmarks, which are listed in \cref{tab:benchmarks2} along with their characteristics.
Four programs are marked with a $\star$: For these programs, the \textsc{Goblint} maintainers identified issues with how the programs were merged from multiple files to obtain a single input
file --- we use the fixed version provided by the \textsc{Goblint} maintainers here.

\begin{table}
    \centering
    \caption{Benchmarks from the \textsc{Concrat}~\cite{DBLP:conf/icse/HongR23} suite. A $\star$ indicates a fixed version as maintained in the \textsc{Goblint} benchmark repository.}
    \label{tab:benchmarks2}
        \footnotesize
        \begin{tabular}{lrrl}%
            \multicolumn{4}{l}{}\\
            \bottomrule
            \addlinespace
            Name  & Lines & LLoC & Description\\
            \addlinespace
            \midrule
            \benchmarkprog{AirConnect} & 17954 & 7512 & Bridge between AirPlay devices and UPnP speakers \\
            \benchmarkprog{axel} & 6004 & 2716 & Download accelerator \\
            \benchmarkprog{brubeck}& 5879 & 2240 & Statistics aggregator \\
            \benchmarkprog{C-Thread-Pool} & 749 & 241 & Minimal \textsc{Posix} threadpool implementation \\
            \benchmarkprog{cava} & 4858 & 2011 & Cross-platform Audio Visualizer            \\
            \benchmarkprog{clib} & 25773 & 11090 & Package manager for C \\
            \benchmarkprog{dnspod-sr} $\star$ & 9473 & 4698 & Recursive DNS Server \\
            \benchmarkprog{dump1090} & 4777 & 2079 & Decoder for Software Defined Radio \\
            \benchmarkprog{EasyLogger} & 2140 & 839 & High-performance C log library \\
            \benchmarkprog{fzy} & 2765 & 1077 & Fuzzy finder for the terminal\\
            \benchmarkprog{klib} & 736 & 293 & Lightweight C library            \\
            \benchmarkprog{level-ip} & 5699 & 2452 & A userspace TCP/IP stack\\
            \benchmarkprog{libaco} & 1302 & 667 & C asymmetric coroutine library \\
            \benchmarkprog{libfaketime} & 528 & 143 & Modifies the system time for a single application \\
            \benchmarkprog{libfreenect} & 646 & 245 & Drivers and libraries for the Xbox Kinect device\\
            \benchmarkprog{lmdb} & 11021 & 5748 & Memory-Mapped Database \\
            \benchmarkprog{minimap2} & 17596 & 9081 & Aligner for DNA or mRNA sequences~\cite{DBLP:journals/bioinformatics/Li18}\\
            \benchmarkprog{Mirai-Source-Code}\tablefootnote{Abbreviated to \benchmarkprog{mirai} in plots.} $\star$ & 1876 & 820 & Mirai malware \\
            \benchmarkprog{nnn} & 12293 & 6712 & Terminal file manager \\
            \benchmarkprog{phpspy} & 19695 & 9551 & Sampling PHP profiler \\
            \benchmarkprog{pianobar} & 11663 & 4382 & Console-based music streaming player \\
            \benchmarkprog{pigz} $^\star$ & 9232 & 5014 & Parallel implementation of gzip compression alogrithm \\
            \benchmarkprog{pingfs} & 2403 & 913 & Filesystem storing information in ICMP ping packets \\
            \benchmarkprog{ProcDump-for-Linux}\tablefootnote{Abbreviated to \benchmarkprog{ProcDump} in plots.} & 4220 & 2157 & Linux version of ProcDump \\
            \benchmarkprog{Remotery} & 7562 & 3531 & CPU/GPU profiler with Remote Web View \\
            \benchmarkprog{shairport} & 8902 & 3791 & AirPlay audio player for Linux  \\
            \benchmarkprog{siege} & 19880 & 9239 & Load tester for HTTP servers \\
            \benchmarkprog{snoopy} & 3638 & 1938 & Library to log program executions  \\
            \benchmarkprog{sshfs} & 7451 & 3258 & Network filesystem client \\
            \benchmarkprog{streem} & 20803 & 9185 & Prototype stream-based programming language \\
            \benchmarkprog{sysbench} $\star$ & 16340 & 3575 & Database and system performance benchmark \\
            \benchmarkprog{the\_silver\_searcher} & 7396 & 3615 &\benchmarkprog{ack}-like code search \\
            \benchmarkprog{uthash} & 822 & 476 & Utilities for working with hashtables in C \\
            \benchmarkprog{vanitygen} & 11163 & 5160 & Vanity address generator for Bitcoin \\
            \benchmarkprog{wrk} & 8883 & 3747 & Load tester for HTTP servers \\
            \benchmarkprog{zmap} & 17908 & 7183 & Network scanner~\cite{durumeric2013zmap} \\
            \bottomrule
        \end{tabular}
\end{table}

For the benchmarks from the first suite, all approaches terminated successfully. For the second suite, there were some outliers:
\begin{itemize}
    \item For \benchmarkprog{axel}, only the baseline terminated within 900~s.
    \item For \benchmarkprog{pigz}, \benchmarkprog{sshfs}, and \benchmarkprog{minimap}, none of the configurations terminated within 900~s.
    \item For \benchmarkprog{lmdb}, \benchmarkprog{remotery}, \benchmarkprog{streem}, \benchmarkprog{the-silver-searcher}, and \benchmarkprog{airConnect},
            all configurations encountered a stack overflow. The problem seems to be caused by a deep recursion in the programs.
            While techniques such as a $k$-callstring or context gas~\cite{erhard2024} may help mitigate this issue, such considerations are orthogonal to the contributions in this paper.
    \item For \benchmarkprog{phpspy}, the \textbf{apinis} configuration timed out, while all others encountered a stack overflow.
    \item For \benchmarkprog{clib}, the \Goblint\ analyzer terminated with an exception claiming the program to be ill-typed for all configurations. This likely is due to a bug in \Goblint.
\end{itemize}
We have excluded these cases from the evaluation in \cref{sec:evaluation}.

\end{document}